\documentclass[fontsize=12pt,a4paper,headings=normal,
twoside=false,leqno,parskip=half-,abstract=true]{scrartcl}
\usepackage[english]{babel}

\usepackage[utf8]{inputenc}
\usepackage{hyperref}
\hypersetup{
 pdftitle={},
 pdfauthor={},
 colorlinks=true,
 linkcolor=blue,
 citecolor=blue,
 filecolor=blue,
 urlcolor=blue}

\usepackage[pagewise]{lineno}%\linenumbers
\usepackage[version=4]{mhchem}
\usepackage{mathtools} 
\usepackage[format=plain,labelfont=bf,font=small]{caption}
\usepackage{subfigure}%replaced by subcaption
\usepackage{xcolor}
\usepackage[arrow, matrix, curve]{xy}
\usepackage{tikz}
\usepackage{float}
\usepackage{orcidlink}

\usepackage{academicons}
\definecolor{orcidlogocol}{HTML}{A6CE39}
\usepackage{tikz-cd}
\usepackage{caption}
\usepackage{amsmath,amsthm}
\usepackage{amssymb} 
\usepackage[normalem]{ulem}
\usepackage{enumitem}
\usepackage[normalem]{ulem}
\usepackage[version=4]{mhchem}

\newtheorem{theorem}{Theorem}[section]
\newtheorem{lemma}[theorem]{Lemma}
\newtheorem{cor}[theorem]{Corollary}
\newtheorem{prop}[theorem]{Proposition}
\newtheorem{conjecture}[theorem]{Conjecture}
\newtheorem{obs}[theorem]{Observation} 

\newtheorem{innercustomthm}{Recipe}
\newenvironment{customthm}[1]
  {\renewcommand\theinnercustomthm{#1}\innercustomthm}
  {\endinnercustomthm}

\theoremstyle{definition}
\newtheorem{definition}[theorem]{Definition}

\theoremstyle{remark}
\newtheorem{remark}[theorem]{Remark}

\numberwithin{equation}{section}

\DeclareMathAlphabet{\mathpzc}{OT1}{pzc}{m}{it}

\newcommand{\be}{\begin{equation}}
\newcommand{\ee}{\end{equation}}
\newcommand{\bea}{\begin{eqnarray}}
\newcommand{\eea}{\end{eqnarray}}

\title{Stoichiometric recipes for periodic oscillations in reaction networks}
\author{Alexander Blokhuis\thanks{Instituto IMDEA Nanociencia, Madrid, Spain \texttt{alexander.willem@imdea.org}}\;, Peter F. Stadler\thanks{Universität Leipzig, Germany; Max Planck Institute for Mathematics in the Sciences, Leipzig, Germany; Institute for Theoretical Chemistry, University of Vienna, Austria; Facultad de Ciencias, Universidad National de Colombia; The Santa Fe Institute, New Mexico, USA.}\;,
Nicola Vassena\thanks
{Universität Leipzig, Germany, \texttt{nicola.vassena@uni-leipzig.de}}
 }
\date{\today}

\begin{document}

\maketitle

\begin{abstract}
 Oscillatory chemical reactions are functional components in a variety of biological contexts. In chemistry, the construction and identification of even rudimentary oscillators remain elusive and lack a general framework. Using parameter-rich kinetics – a methodology enabling the disentanglement of parametric dependencies from structural analysis – we investigate the stoichiometry of chemical oscillators. We introduce the concept of \emph{oscillatory cores}: minimal subnetworks that guarantee the potential for oscillations in any reaction network containing them. These cores fall into two classes, depending on whether they involve positive or negative feedback. In particular, the latter class unveils a family of oscillators - yet to be synthesized - that require a minimum number of reaction steps to exhibit oscillations, a phenomenon we refer to as the \emph{principle of length}. We identify several mechanisms through which catalysis promotes oscillations: (I) furnishing instability (e.g. autocatalysis), (II) lifting dependencies, (III) lowering length thresholds. Notwithstanding this mechanistic ubiquity, we show that oscillators can also be realized without employing any catalysis. Our results highlight branches of chemistry where oscillators are likely to arise by chance, suggest new strategies for their design, and point to novel classes of oscillators yet to be realized experimentally.
\end{abstract}

\newpage

\tableofcontents

\section{Introduction}

Chemical oscillators were first observed more than a century ago in the Morgan \cite{Morgan:1916} and Bray–Liebhafsky \cite{Bray:1921}
reactions. Lying dormant for more than fifty years, this phenomenon gained
broader interest in the 1970s as an intriguing example of nonlinear
dynamical systems \cite{Noyes:1972}. Beyond classical examples such as the
Belousov--Zhabotinsky reactions \cite{Zhabotinsky:64}, new chemical
oscillators continue to be developed
\cite{Kovacs2007,Semenov2015,Semenov2016,Semenov2021}. This remains an
arduous task, however, due to the lack of a clear structural understanding
of the mechanisms underlying oscillations. While chemical reaction network
theory \cite{Fei19} has yielded results such as the celebrated
deficiency-zero theorem that guarantees stability, little is known that
connects network structure to oscillatory behavior. A structural
characterization of autocatalysis \cite{blokhuis20} and the development of
the concept of \emph{unstable cores} \cite{VasStad23} (which in particular
includes \emph{autocatalytic cores}) have provided sufficient conditions for
instability. Which of these unstable motifs are essential for realizing
specific types of nonlinear dynamics still remains an open question
\cite{Gosh2024,Ivan25}.

It is often thought that the most natural way to analyze the dynamics of a
chemical reaction network is through the use of the law of mass action. At first
glance, one might expect mass-action systems to be easier to analyze, given
that only one parameter appears in each reaction rate. Paradoxically  - but
not quite - the opposite is true: in mass-action systems, the fluxes and
reactivities - i.e., the evaluated reaction functions and their first
derivatives - are necessarily intertwined at a steady-state through their parametric dependencies and coproduction \cite{blokhuis2025datadimensionchemistry}. These intricate relationships makes it difficult to extract
clear criteria or even heuristics from the stoichiometry of the system alone.
\emph{Parameter-rich} models with monotone rate functions $\mathbf{r}$, on the
other hand, provide sufficient flexibility to specify the coordinates of a
(positive) steady-state  $\bar{\mathbf{x}}$ and the non-zero numerical values of
the reactivities
at $\bar{\mathbf{x}}$ independently of each other, see Sec.~\ref{sec:prichmain} below for a more formal definition.  Although parameter-richness might appear to be somewhat restrictive, it is readily satisfied by important kinetic models such as Michaelis--Menten \cite{MM13}, Hill \cite{Hill10} and generalized mass-action kinetics \cite{Muller:12}, see \cite{VasStad23, blokhuis2025datadimensionchemistry}. This approach has thus intrinsic validity as a
design framework for oscillators in biochemical systems.

 Motivated by this, we investigate stoichiometric recipes - simple network motifs that ensure the potential for periodic oscillations in any reaction network containing them. The parameter-rich framework allows us to formulate conditions directly at the stoichiometric level, providing accessible criteria for initial analysis. More detailed kinetic aspects, such as thermodynamic feasibility, must be assessed case by case and lie beyond our current scope. Recipe \ref{recipe:1main} involves an unstable-positive feedback within a stable subnetwork; the feedback may be autocatalytic or not. Recipe \ref{recipe:2main} relies on an unstable-negative feedback, which is always non-autocatalytic: oscillations may appear only when a sufficient number of reaction steps are present (\emph{Principle of length}, see Sec.~\ref{sec:principlelength}). Finally and independently, we observe that instability always leads to oscillations if multistationarity is excluded. This broader mechanism is captured in Recipe \ref{recipe:0main}. See Fig.~\ref{fig:cores} for a first bird-eye overview on the three recipes.

The paper is organized as follows. Section~\ref{sec:mathresmain} presents the main mathematical result in concise form and without proofs. Section~\ref{sec:principlelength} introduces and focuses on the \emph{principle of length} for negative-feedback oscillators. The roles of catalysis are discussed in Section~\ref{sec:catalysis}, and perspectives on constructing actual oscillators in the laboratory are presented in Section~\ref{sec:buildingosci}. Section~\ref{sec:discussion} concludes with a wrap-up discussion. The proofs of the main results stated in Section~\ref{sec:mathresmain} are given in Section~\ref{sec:proofs}.

The Supplementary Material (SM) consists of three parts. Part~I presents the background mathematical framework in greater detail, including proofs of straightforward statements as well as results already established elsewhere, for self-containment. Part~II provides a detailed analysis of the examples discussed in the main text and introduces additional examples to further support our work. Part~III explicitly illustrates how to identify periodic solutions \emph{in silico} for parameter-rich systems.

\paragraph{Acknowledgments.}A.B. acknowledges fruitful conversations with Thomas Hermans. A.B. acknowledges the EU (MSCA-PF-2023 "KENA"
  no. 101155395). N.V. and P.F.S. ackowledge the MATOMIC consortium,
  funded by the Novo Nordisk Foundation, grant NNF21OC0066551. Research in
  the Stadler labs is supported by the German Federal Ministry of Education
  and Research BMBF through DAAD project 5761681433 (SECAI, School of
  Embedded Composite AI).

\section{Mathematical Results}\label{sec:mathresmain}

\subsection{D-Hopf matrices and oscillations} 

Throughout this contribution we assume a finite set of chemical reactions
$R$ on a finite set $M$ of chemical species.  For ease of presentation we
exclude here explicit catalysis, i.e., in any given reaction, a species
does not appear as both a reactant and a reaction product. We refer to the Supplementary Material (SM) for a presentation of the theory in full generality, without such an assumption. The reaction
network $\mathbf{\Gamma}=(M,E)$ is therefore described by its
\emph{stoichiometric matrix} $\mathbb{S}$ with entries $\mathbb{S}_{mj}<0$
if $\mathsf{m}\in M$ is a reactant in reaction $j\in E$,
$\mathbb{S}_{mj}>0$ if $\mathsf{m}$ is a product of $j$, and
$\mathbb{S}_{mj}=0$ if $\mathsf{m}$ does not take part in $j$.

\begin{figure} %[tbhp]
\centering
\includegraphics{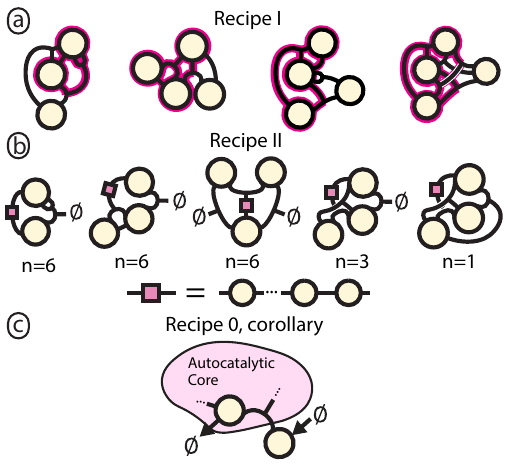}
\caption{a) Via Recipe \ref{recipe:1main}, oscillatory cores obtained by extending autocatalytic cores to stable subnetworks. b) Via Recipe \ref{recipe:2main}, oscillatory cores based on negative-feedback unstable cores, which are always nonautocatalytic. c) A simple setup to make an oscillator following Recipe \ref{recipe:0main}: to an autocatalytic core, a replenishment of an off-core reactant consumed by an autocatalyst is added, along with a degradation step for the same autocatalyst.}
\label{fig:cores}
\end{figure}
%\firstpage{1}[46]

We consider well-mixed, spatially homogeneous reaction networks, whose time
evolution for the concentration vector $\mathbf{x}=\mathbf{x}(t)\ge0$ is described by the system of ordinary
differential equations (ODE): \begin{equation}\label{eq:maineqmain}
\dot{\mathbf{x}} = \mathbb{S}\mathbf{r}(\mathbf{x}).   
\end{equation}  The nonnegative reaction rates
$\mathbf{r}(\mathbf{x})\ge0$ depends on the vector $\mathbf{x}$ of concentrations.  We consider only so-called monotone reaction rates, which satisfy the following natural assumptions: each rate $r_j(\mathbf{x}) \ge 0$ depends only on the concentrations of its reactants and is zero if and only if at least one of their concentrations is zero. If all of them are present at non-zero concentrations, moreover, the rate increases with the concentrations. That is, the \emph{reactivity}
$R_{jm}\coloneqq \partial r_j/\partial x_m$ is positive for all $\mathsf{m}\in X$
with $\mathbb{S}_{mj}<0$ and $R_{jm}=0$ otherwise, provided $x_m>0$ for all
$m\in X$. The fact that the non-zero reactivities $R_{jm}$ for any positive
concentration vector $\mathbf{x}>0$ are determined by $\mathbb{S}$, and thus by the
structure of the network $\mathbf{\Gamma}$ alone, plays a key role in this
work.

We focus on \emph{consistent} reaction networks \cite{Ang07}, which admit at least one positive steady-state, i.e., whose stoichiometric matrix $\mathbb{S}$ has at least one positive kernel vector $v>0$: 
\begin{equation}
\mathbb{S}v=0.   \end{equation}

The theorem of Hartman and Grobman \cite{Hsubook} then ensures that the local behavior of the dynamical system 
\eqref{eq:maineqmain}
near a hyperbolic steady state $\bar{\mathbf{x}}$ is entirely determined by the spectrum of its linearization - the Jacobian matrix 
\begin{equation}
G(\bar{\mathbf{x}}) = \mathbb{S} R(\bar{\mathbf{x}}), 
\end{equation}
where the \emph{reactivity matrix} $R(\bar{\mathbf{x}})$ encodes all reactivities $R_{jm}$, i.e., the partial derivatives of the reaction rates. In particular, Hurwitz-stability of $G(\bar{\mathbf{x}})$, i.e. all of its eigenvalues have negative real-part, implies local dynamical stability of $\bar{\mathbf{x}}$. 
Accordingly, we identify the onset of periodic orbits by locating \emph{Hopf bifurcation points} $\bar{\mathbf{x}}_H>0$, where the Jacobian loses hyperbolicity due to the presence of purely imaginary eigenvalues. Instead of relying on the classical approach based on \emph{local} Hopf bifurcation \cite{GuHo84}, we leverage the theory on \emph{global} Hopf bifurcation \cite{Fiedler85PhD} (see Fiedler's Thm.~\ref{thm:fiedlermain}), which provides a more computationally tractable criterion for the existence of periodic orbits.
To this aim, we build on the linear algebra concept of $D$-stability. A matrix $A$ is said to be \emph{$D$-stable} if $AD$ is Hurwitz-stable for every positive diagonal matrix $D$, and \emph{$D$-unstable} if there exists a positive diagonal matrix $D$ such that $AD$ is Hurwitz-unstable. We now recall the concept of \emph{inertia} of a matrix \cite{ostrowski62}.

\begin{definition}[Inertia of a matrix]
The inertia of an $n \times n$ square matrix $A$ is a nonnegative triple 
    $$\operatorname{inertia}(A) \coloneqq (\sigma^-_A,\sigma^+_A,\sigma^0_A),$$
where $\sigma^-_A$, $\sigma^+_A$, and $\sigma^0_A$ represent the number of eigenvalues of $A$ with negative real part, positive real part, and zero real part, respectively. The eigenvalues are counted according to their algebraic multiplicities so that $\sigma^+_A + \sigma^-_A + \sigma^0_A = n$.
\end{definition}

We introduce some necessary notation: let $\kappa$ denote any selection of $k \leq n$ indices from $\{1,\dots,n\}$. $A[\kappa]$ represents the principal submatrix of $A$, obtained by retaining the columns and rows with common indices given by those in $\kappa$; its determinant is referred to as \emph{a $k$-principal minor} of $A$. We state a new definition of $D$-Hopf matrix as follows.

\begin{definition}[D-Hopf matrix]\label{def:dhopfmain}
A $n\times n$ matrix $A$ is called \emph{D-Hopf} if there exists an invertible $k$-principal submatrix $A[\kappa]$ of $A$, and two positive $k\times k$ diagonal matrices $D_1$ and $D_2$ such that 
\begin{equation}\label{eq:changeofinertiamain}
\operatorname{inertia}A[\kappa]D_1\neq\operatorname{inertia}A[\kappa]D_2.
\end{equation}
\end{definition}
The motivation for introducing $D$-Hopf matrices lies in their link to purely imaginary eigenvalues. 
\begin{lemma}\label{lem:Dhopfmain}
If $A$ is $D$-Hopf, then there exists positive diagonal matrix $D$ such that
$AD$ has purely imaginary eigenvalues.
\end{lemma}
See the SM, Lemma~\ref{lem:Dhopfpurely}, for a proof. Our main mathematical result, leveraging global Hopf bifurcation and the structure of reaction networks, lifts Lemma \ref{lem:Dhopfmain} from linear algebra to nonlinear dynamics.
\begin{theorem}\label{thm:mainmain}
Under a mild nondegeneracy condition, if there exists a parameter choice for which the Jacobian $G(\bar{\mathbf{x}})=\mathbb{S}R(\bar{\mathbf{x}})$ is $D$-Hopf at a positive steady-state $\bar{\mathbf{x}}>0$, then the system \eqref{eq:maineqmain} admits periodic orbits.
\end{theorem}
The proof is presented in Sec.~\ref{sec:proofs}. The nondegeneracy condition simply requires that the system's dimension equals the number of species minus the number of independent linear conservation laws, see Def.~\ref{def:nondegnetmain} for details. Theorem~\ref{thm:mainmain} also holds for mass-action kinetics, but to derive sufficient conditions that can be easily interpreted in terms of the stoichiometry, without any algebraic pre-processing, we assume parameter-rich kinetics from this point onward.

\subsection{Parameter-rich kinetics and stoichiometric interpretations}\label{sec:prichmain}
Concrete rate models
$\mathbf{r}(\mathbf{x})$, and thus reactivities, typically depend on a set of positive parameters
$\mathbf{p}>0$. We write $\mathbf{r}(\mathbf{x};\mathbf{p})$ to emphasize this fact.  Parameter-rich kinetic models are defined as follows.
\begin{definition}\label{def:prichmain}
 A parametric kinetic model $\mathbf{r}(\mathbf{x};\mathbf{p})$ is  said to be
\emph{parameter-rich} if, for any prescribed steady-state concentration
$\bar{\mathbf{x}}>0$ and any assignment of the non-zero reactivities $R_{jm}>0$, there is a choice of parameters $\bar{\mathbf{p}}$ such that:
\begin{equation}
  \begin{cases}
    \mathbb{S} \mathbf{r}(\bar{\mathbf{x}}; \bar{\mathbf{p}}) = 0, \\
    \dfrac{\partial r_j(\mathbf{x};\bar{\mathbf{p}})}{\partial x_m}
    \bigg|_{\mathbf{x}=\bar{\mathbf{x}}} = R_{jm}.
  \end{cases}
\end{equation}   
\end{definition}
Therefore, under the assumption of parameter-rich kinetics, $R(\bar{\mathbf{x}})$ is
interpreted as a non-negative matrix with non-zero entries determined entirely by
$\mathbb{S}$, whose numerical value can be chosen independently of
$\bar{\mathbf{x}}$. Thus we simply write $R$ instead of $R(\bar{\mathbf{x}})$ and call
$R$ the \emph{symbolic reactivity matrix}. Correspondingly, the Jacobian $G=\mathbb{S}R$ also becomes a symbolic
matrix. The parameter-rich framework thus permits to study changes in the
spectrum of the \emph{symbolic Jacobian} $G$ directly. This amounts to
assessing the spectral properties of $G$ as we scan all possible
choices of the reactivities $R$. This approach allows us to
derive simple, stoichiometry-based sufficient recipes for the onset of
periodic orbits. Since Michaelis--Menten kinetics, the Hill model, and
generalized mass-action kinetics all fall within the class of
parameter-rich models \cite{VasStad23}, this framework serves as a sound basis for understanding biochemical oscillators.

\subsubsection{Child Selections and Cores}
The parameter-rich framework allows the values of reactivities to be chosen so that certain subsets of reactions and reactants dominate the linearized dynamics. This, in turn, enables the analysis of the spectrum of symbolic Jacobians by isolating $k \times k$ stoichiometric submatrices that - under suitable parameter choices - approximate $k$ dominant eigenvalues of the full Jacobian. The same strategy has previously been used to identify stoichiometric sources of network instability \cite{VasStad23}. We consider subsets of $k$ reactants $\kappa
\subseteq M$ and $k$ reactions $E_\kappa\subseteq E$ such that there is 1-1 correspondence (i.e. a bijection $J$) between species 
$\mathsf{m} \in \kappa$ and a reaction $j=J(m) \in E_\kappa$ where $\mathsf{m}$ participates as a reactant. We call the triple $\pmb{\kappa}=(\kappa,E_\kappa,J)$ a \emph{Child-Selection triple} (CS) and we define the associated \emph{Child-Selection matrix} (CS-matrix) $\mathbb{S}[\pmb{\kappa}]$ as:
\begin{equation}
\mathbb{S}[\pmb{\kappa}]_{mn}:=\mathbb{S}_{mJ(n)}.
\end{equation}
In the absence of explicit catalysis, i.e. species appearing both as a reactant and a product of a reaction, CS-matrices correspond to square submatrices of the stoichiometric matrix $\mathbb{S}$ that, after an appropriate column permutation, yield matrices $\mathbb{S}[\pmb{\kappa}]$ with a strictly negative diagonal.

Again following \cite{VasStad23}, an analysis of the characteristic polynomial of the symbolic Jacobian  $G = \mathbb{S}R$, using the Cauchy--Binet formula, reveals the following: by tuning all reactivities $R_{jm}$ with $(\mathsf{m},j) \neq (\mathsf{m}, J(\mathsf{m}))$ sufficiently small, the product  
\begin{equation}\label{eq:Dmain}
\mathbb{S}[\pmb{\kappa}] \cdot \operatorname{diag}(R_{m_1J(m_1)}, \dots, R_{m_kJ(m_k)}), \quad \kappa=\{m_1,...,m_k\}
\end{equation}
approximates the spectrum of the principal submatrix $G[\kappa]$ associated to the set $\kappa$ and, moreover,  approximates the $k$ dominant eigenvalues of the full Jacobian $G$.  In particular, then, the existence of $D$-unstable (and in particular of any Hurwitz-unstable) CS-matrix $\mathbb{S}[\pmb{\kappa}]$ is sufficient for the existence of a parameter choice such that the network admits a Hurwitz-unstable steady-state. To make this contribution self-contained, we prove this statement as Prop.~\ref{pro:csunstable} in the SM. The formulation \eqref{eq:Dmain} as a product with a positive diagonal matrix is pivotal in our emphasis on $D$-Hopf matrices as well.

The fact that any restriction $\pmb{\kappa'}=(\kappa'\subset \kappa, E_{\kappa'}\subset E_\kappa, J)$ of a CS $\pmb{\kappa}$ is itself a CS makes minimal  CS well-defined. At the level of the associated CS-matrices, the restriction $\mathbb{S}[\pmb{\kappa}']$ appears as a principal submatrix of $\mathbb{S}[\pmb{\kappa}]$. 
We can then define \emph{cores} in the network based on the following matrix definition.
\begin{definition}
Let $\mathbb{P}$ be a matrix property. A $\mathbb{P}$-core is a
CS-matrix $\mathbb{S}[\pmb{\kappa}]$ with property $\mathbb{P}$ that
does not have a proper principal submatrix with property $\mathbb{P}$.
\end{definition}
Unstable cores (resp. D-unstable cores) are CS-matrices that are minimal
with the property of being Hurwitz-unstable (resp. D-unstable)
\cite{VasStad23}. A $k\times k$ unstable core $\mathbb{S}[\pmb{\kappa}]$ is called
an \emph{unstable-positive feedback} if
$\operatorname{sign}\operatorname{det}\mathbb{S}[\pmb{\kappa}]=(-1)^{k-1}$ and
\emph{unstable-negative feedback} if
$\operatorname{sign}\operatorname{det}\mathbb{S}[\pmb{\kappa}]=(-1)^{k}.$ This definition generalizes the convention \cite{Tyson2002} by which feedback loops are built up from products of coefficients in the Jacobian as they appear in the characteristic polynomial. Whether a feedback loop is positive or negative is determined by the sign of this product. 
Unstable-positive feedbacks always have one single real positive
eigenvalues, while unstable-negative feedbacks only possess a conjugate
complex pair of eigenvalues with positive real part, but no real positive
eigenvalues.  Finally, a notion of autocatalytic cores was introduced in
\cite{blokhuis20} to classify minimal autocatalytic subsystems in reaction
networks. Autocatalytic cores then turned out to be a special case of
unstable-positive feedbacks. More precisely, the autocatalytic
cores in the sense of \cite{blokhuis20} are exactly the
unstable cores in the sense of \cite{VasStad23} whose restricted
stoichiometric matrix is a Metzler matrix, i.e., all off-diagonal elements
are non-negative (while the diagonal has only negative entries).

\begin{figure}%[!tbhp]
\centering
\includegraphics{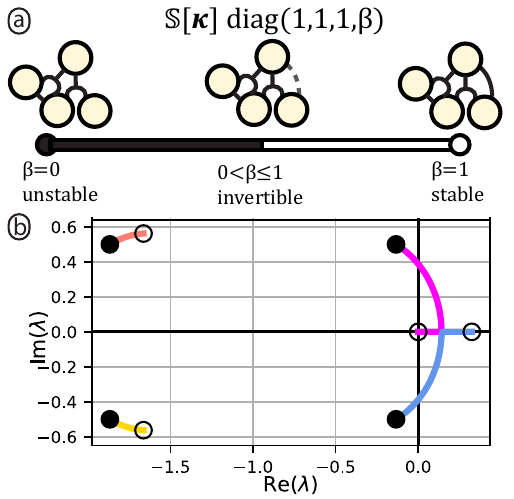}
\caption{Recipe \ref{recipe:1main} illustrated for the motif $\ce{X}\rightarrow\ce{Y}+\ce{Z}$; $\ce{Y}\rightarrow \ce{Z}$; $\ce{Z}+\ce{W}\rightarrow \ce{X}$; $\ce{W}\rightarrow \ce{X}$. The first three reactions on species $(\ce{X},\ce{Y},\ce{Z})$ form an autocatalytic core with stoichiometric matrix as in \eqref{eq:autcorII}. Any consistent reaction network that includes such motif has the capacity for periodic oscillations. The bifurcation process follows a parameter $\beta\in(0,1]$. At $\beta=1$ the associated CS-matrix is Hurwitz stable, while at $\beta\approx0$ is unstable due to the presence of the autocatalytic core. The product matrix is invertible throughout: change of stability happens at purely imaginary eigenvalues. a) The bifurcation process is intuitively illustrated in the network as a removal of the reactivity of one reaction. b) The trajectory of the eigenvalues is displayed where a purely imaginary crossing can be seen.}
\label{fig:recipe_1}
\end{figure}

\subsubsection{Oscillatory Cores}
The concept of cores as minimal topological substructures conveying the
potential for qualitative dynamical properties of reaction networks is
investigated here for oscillatory behavior.
\begin{definition}\label{def:Oscillatorycoremain}
 A $D$-Hopf CS-matrix
$\mathbb{S}[\pmb{\kappa}]$ is called an \textbf{oscillatory core} if none of its principal submatrices is $D$-Hopf.   
\end{definition}
The existence of any
$D$-Hopf CS-matrix $\mathbb{S}[\pmb{\kappa}]$, irrespective of its size, is
sufficient for the assumptions of Thm.~\ref{thm:mainmain} to hold, under
parameter-rich kinetics, guaranteeing the %insurgence
emergence of periodic orbits, as the following corollary to Thm.~\ref{thm:mainmain} states.
\begin{cor}\label{cor:dhopfoscillationsmain}
  Let $\pmb{\Gamma}$ be a network with parameter-rich kinetics.  If
  any CS-matrix $\mathbb{S}[\pmb{\kappa}]$ is $D$-Hopf, then the system
  admits periodic solutions.
\end{cor}
See Sec.~\ref{sec:proofscor} for a proof. While the $D$-Hopf property can be in general computationally difficult to
verify, there exist several explicit and tractable configurations in which
it can be asserted directly. Here, we focus on two complementary and
illustrative cases.  Consider any $k\times k$ CS-matrix
$\mathbb{S}[\pmb{\kappa}]$ with
$\operatorname{sign}\operatorname{det}\mathbb{S}[\pmb{\kappa}]=(-1)^k,$ and
let $\mathbb{S}[\pmb{\kappa}-1]$ indicate any of its $(k-1)\times (k-1)$
principal submatrix, i.e. without loss of generality:
\begin{equation}
\mathbb{S}[\pmb{\kappa}]=\begin{pmatrix}
    \mathbb{S}[\pmb{\kappa}-1] & \vdots\\
    ... & -
\end{pmatrix}.\end{equation}
If the Hurwitz-stability of $\mathbb{S}[\pmb{\kappa}]$ and $\mathbb{S}[\pmb{\kappa}-1]$ differs, then $\mathbb{S}[\pmb{\kappa}]$ is $D$-Hopf (See Cor.~\ref{cor:Ocores1} and Cor.~\ref{cor:Ocores2b} in the SM for a straightforward proof) and thus \emph{the network admits periodic orbits}. Naturally, two complementary configurations arise, depending on whether $\mathbb{S}[\pmb{\kappa}]$ is Hurwitz-stable and $\mathbb{S}[\pmb{\kappa}-1]$ is Hurwitz-unstable, or \textit{vice versa}.

%\textcolor{teal}{NV: I felt like including the one matrix example (the one of Fig.~\ref{fig:recipe_1} and Fig.~\ref{fig:recipe_2}) is the best way to be clear. One could also remove them, of course... but I could not easily explain the matter without them...}
Generalizing the above perspective, we define  \textbf{Oscillatory core of class I} as follows.
\begin{definition}[Oscillatory Cores of Class I]\label{def:Ocores1main} An \textbf{Oscillatory Core of Class I} is a CS-matrix $\mathbb{S}[\pmb{\kappa}]$ that is minimal with the property of being Hurwitz-stable and possessing an unstable-positive feedback as a principal submatrix $\mathbb{S}[\pmb{\kappa}']$. 
\end{definition}
Since autocatalysis is a specific
instance of unstable-positive feedback, oscillatory cores of class I
in particular include the autocatalytic examples. 
%may include autocatalysis.
However, as shown in \cite{VasStad23}, unstable-positive feedback can also
occur in non-autocatalytic forms, so autocatalysis is a natural but not
strictly necessary component: examples with non-autocatalytic oscillatory cores of class I can be found in \cite[Example B]{VasStad23} for parameter-rich kinetics and \cite[Example I]{Vassena2025} for mass-action kinetics. In this paper, as an illustration, we constructed five examples of oscillatory cores of class I by extending the five autocatalytic cores from \cite{blokhuis20} into Hurwitz-stable matrices. See Table \ref{tab:posmain} for an overview and Fig.~\ref{fig:recipe_1} for one detailed case, where we enlarged the $3\times3$ autocatalytic core
\begin{equation}\label{eq:autcorII}
\mathbb{S}[\pmb{\kappa}']=
    \begin{pmatrix}
        -1 & 0 & 1\\
        1 & -1 & 0 \\
        1 & 1 & -1
    \end{pmatrix},
\end{equation}
corresponding to reactions $\ce{X}\rightarrow \ce{Y}+\ce{Z}$, $\ce{Y}\rightarrow \ce{Z}$, $\ce{Z}\rightarrow \ce{X}$, to a  $4\times 4$ Hurwitz-stable CS-matrix:
\begin{equation}
\mathbb{S}[\pmb{\kappa}]=
    \begin{pmatrix}
        -1 & 0 & 1 & 1\\
        1 & -1 & 0 & 0\\
        1 & 1 & -1 & 0\\
        0 & 0 & -1 & -1
    \end{pmatrix},
\end{equation}
corresponding to reactions $\ce{X}\rightarrow\ce{Y}+\ce{Z}$; $\ce{Y}\rightarrow \ce{Z}$; $\ce{Z}+\ce{W}\rightarrow \ce{X}$; $\ce{W}\rightarrow \ce{X}$. A complete analysis of these five constructions is provided in the SM, Sec.~\ref{sec:recipeI}. The presence of an oscillatory core of class I defines \textbf{Recipe \ref{recipe:1main}} for oscillations:
\begin{customthm}{I}\label{recipe:1main}
 Assume that the network possesses an Oscillatory Core of Class I. Then the system admits nonstationary periodic orbits.   
\end{customthm}

\begin{table}[h!]
    \centering
    \begin{tabular}{c|c|c|c|c}        Osc. Core (I,a)  & {\scriptsize$\vcenter{\hbox{$\begin{pmatrix}
        -1 & 2 & 1\\
         1 & -1 & 1\\
         0 & -1 & -1\\
    \end{pmatrix}$}}$}& {\scriptsize$
\begin{cases}
\ce{X} + ... &\rightarrow \quad\ce{Y} + ...\\
\ce{Y} + \ce{Z} + ... &\rightarrow \quad 2 \ce{X} + ...\\
\ce{Z}+...&\rightarrow \quad \ce{X}+\ce{Y}
\end{cases}
$}&$\vcenter{\hbox{\includegraphics[scale=0.3]{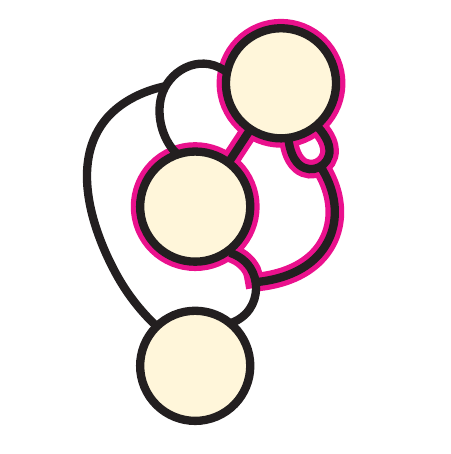}}}$ &$\vcenter{\hbox{\includegraphics[scale=0.2]{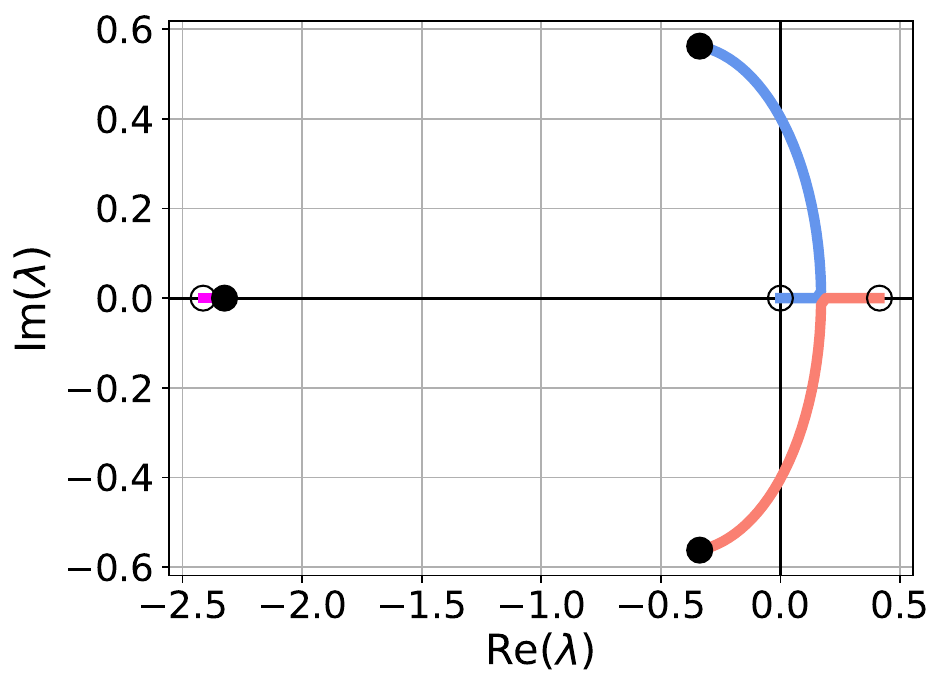}}}$ \\
        \hline   
    \\
              Osc.  Core (I,b)  & {\scriptsize $\vcenter{\hbox{$\begin{pmatrix}
        -1 & 0 & 1 & 1\\
         1 & -1 & 0& 0\\
         1 & 1 & -1& 0\\
         0 & 0 & -1 & -1
        \end{pmatrix}$}}$}  & {\scriptsize $\begin{cases}
\ce{X} + ... &\rightarrow \quad\ce{Y} + \ce{Z} + ...\\
\ce{Y}  + ... &\rightarrow \quad  \ce{Z} + ...\\
\ce{Z}+\ce{W}+...&\rightarrow \quad \ce{X}+...\\
\ce{W}+... &\rightarrow \quad \ce{X}+...
\end{cases}$} &$\vcenter{\hbox{
       \includegraphics[scale=0.3]{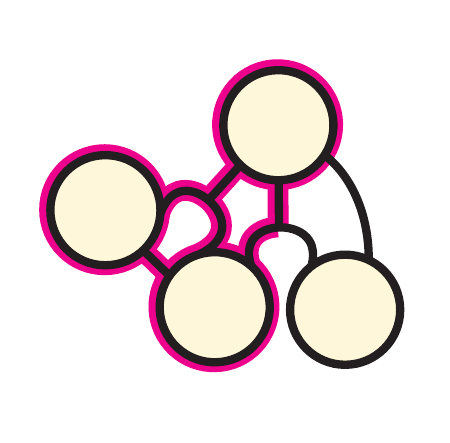}}}$ &$\vcenter{\hbox{\includegraphics[scale=0.2]{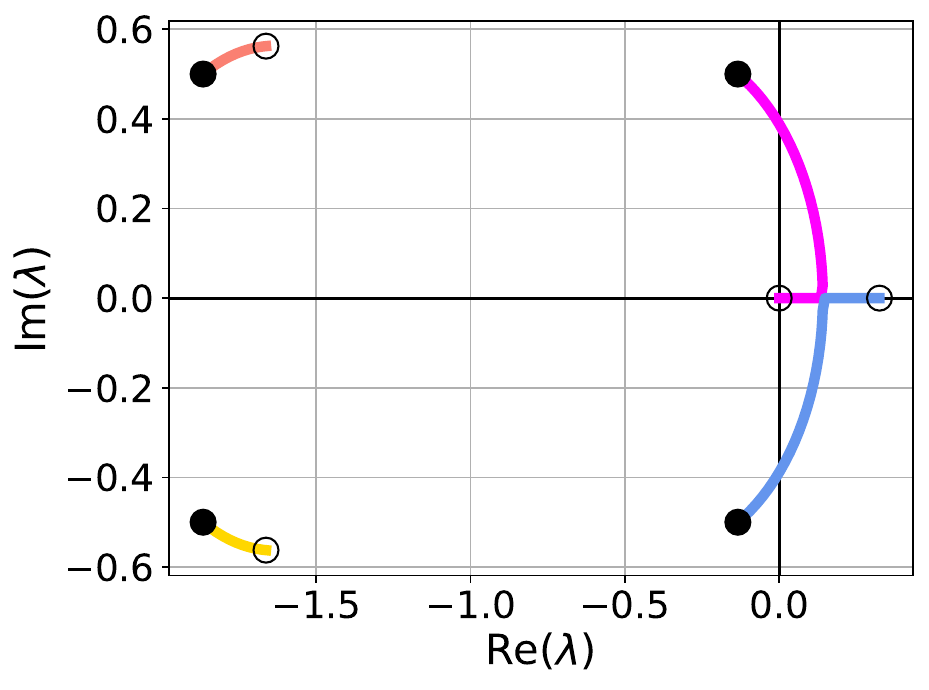}}}$ \\
        \hline \\
              Osc. 
               Core (I,c) & {\scriptsize $\vcenter{\hbox{$\begin{pmatrix}
        -1 & 1 & 1 & 2\\
         1 & -1 & 0& 0\\
         1 & 0 & -1& 0\\
         0 & 0 & -1 & -1 
        \end{pmatrix}$}}$} & {\scriptsize $\begin{cases}
\ce{X} + ... &\rightarrow \quad\ce{Y} + \ce{Z} +...\\
\ce{Y}  + ... &\rightarrow \quad  \ce{X} + ...\\
\ce{Z}+\ce{W}+...&\rightarrow \quad \ce{X}+...\\
\ce{W}+... &\rightarrow \quad 2\ce{X}+...
\end{cases}$} &  $\vcenter{\hbox{\includegraphics[scale=0.3]{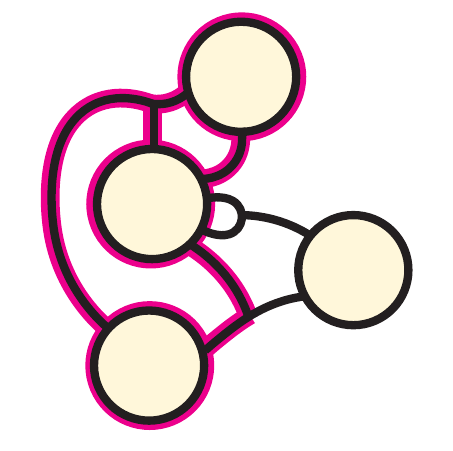}}}$ &$\vcenter{\hbox{\includegraphics[scale=0.2]{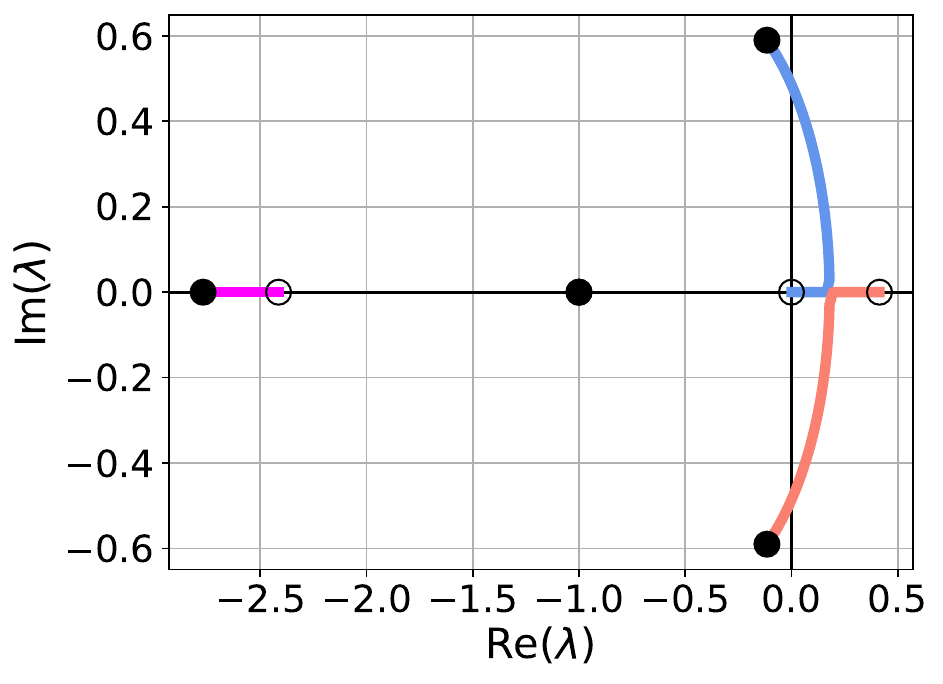}}}$    \\
         \hline\\
              Osc. Core (I,d) & {\scriptsize $\vcenter{\hbox{$\begin{pmatrix}
        -1 & 1 & 1 & 1\\
         1 & -1 & 0& 1\\
         1 & 1 & -1& 0\\
         0 & -1 & -1 & -1 
        \end{pmatrix}$}}$} &{\scriptsize $\begin{cases}
\ce{X} + ... &\rightarrow \quad\ce{Y} + \ce{Z}+ ...\\
\ce{Y} + \ce{W} + ... &\rightarrow \quad  \ce{X} + \ce{Z} +...\\
\ce{Z}+\ce{W}+...&\rightarrow \quad \ce{X}+...\\
\ce{W}+... &\rightarrow \quad \ce{X}+\ce{Y}+...
\end{cases}$ }& $\vcenter{\hbox{\includegraphics[scale=0.3]{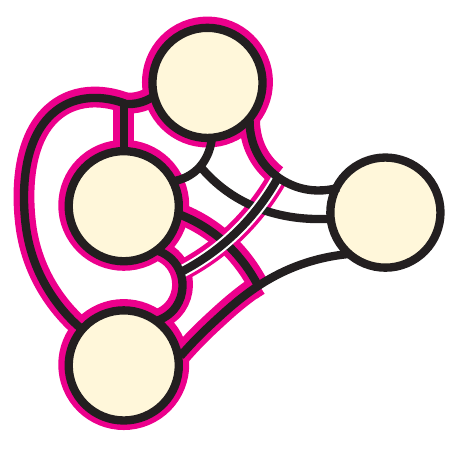}}}$ &$\vcenter{\hbox{\includegraphics[scale=0.2]{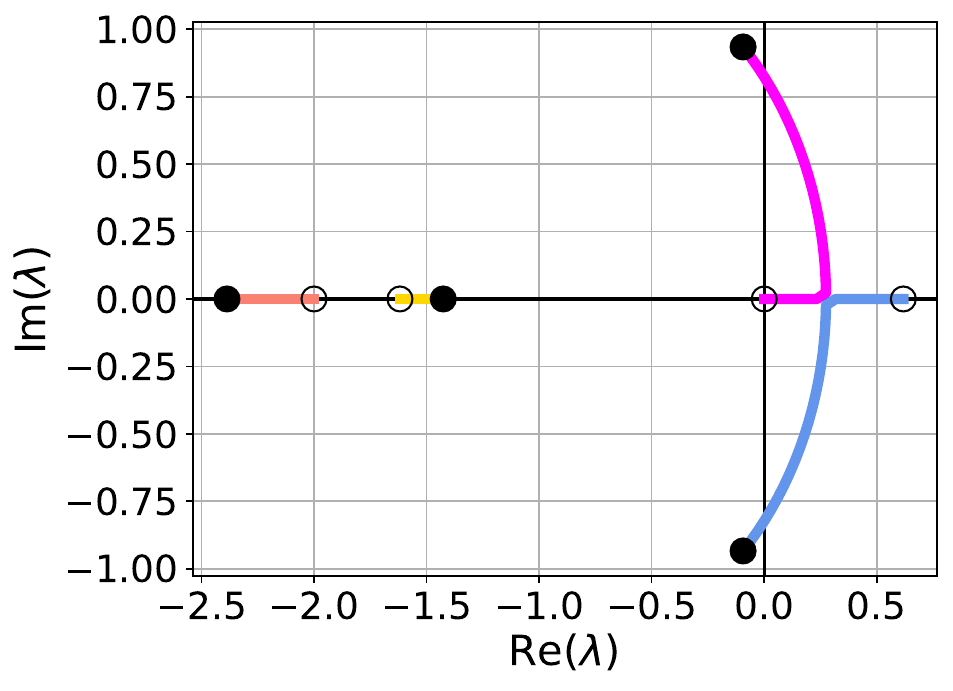}}}$   \\
        \hline \\
               Osc. Core (I,e) & {\scriptsize  $\vcenter{\hbox{$\begin{pmatrix}
        -1 & 1 & 1 & 1\\
         1 & -1 & 1& 1\\
         1 & 1 & -1& 0\\
         -1 & -1 & -1 & -1 
        \end{pmatrix}$}}$}  & {\scriptsize $\begin{cases}
\ce{X} + \ce{W}+ ... &\rightarrow \quad\ce{Y} + \ce{Z} +...\\
\ce{Y} + \ce{W} + ... &\rightarrow \quad  \ce{X} + \ce{Z} +...\\
\ce{Z}+\ce{W}+...&\rightarrow \quad \ce{X}+\ce{Y}+...\\
\ce{W}+... &\rightarrow \quad \ce{X}+\ce{Y}+...
\end{cases}$}& $\vcenter{\hbox{\includegraphics[scale=0.3]{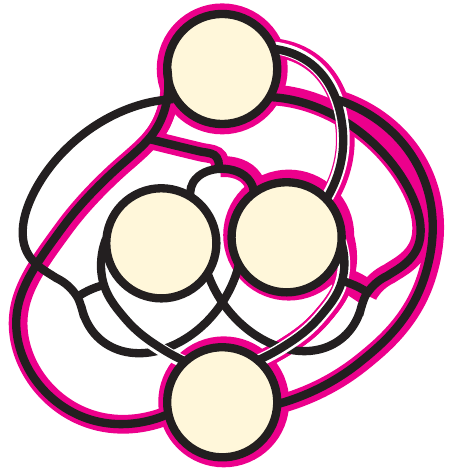}}}$ &$\vcenter{\hbox{\includegraphics[scale=0.2]{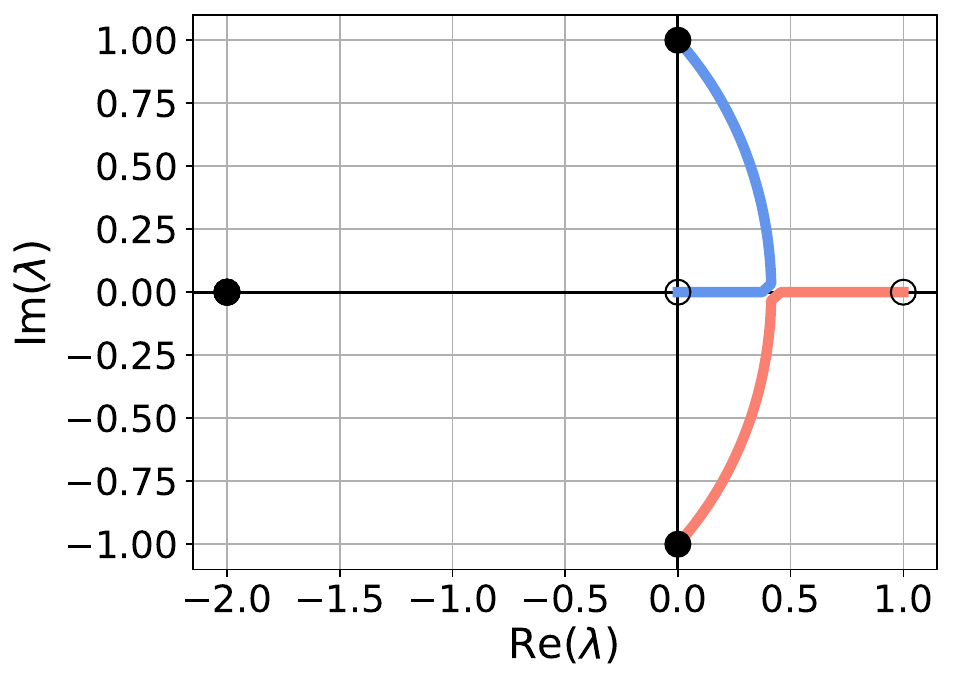}}}$  \\
    \end{tabular}
\caption{Oscillatory Cores (I,a) to (I,e). An autocatalytic core is highlighted in magenta in each CRN, corresponding to a) Type I, b) Type II, c) Type III, d) Type IV, e) Type V, following the classification in Ref. \cite{blokhuis20}. %We remark that oscillatory cores (I,b) to (I,e) contain additional autocatalytic cores (see appendix [todo]). 
Eigenvalues are plotted for $\mathbb{S} D(\beta)$, with $D(\beta):=\operatorname{diag}(1,...,1,\beta)$ and $0<\beta<1$. Endpoints $\beta=0,1$ are marked with circles ($\beta=0$ hollow, $\beta=1$ bold). $\beta=0$ corresponds to elimination of the final column of $\mathbb{S}$, i.e. the reaction not highlighted in magenta.}
\label{tab:posmain}
\end{table}

In turn, an \textbf{Oscillatory core of class II} involves unstable-negative feedback and it is defined as follows:
\begin{definition}[Oscillatory Cores of Class II]\label{def:Ocores2main} An \textbf{Oscillatory Core of Class II} is a $k\times k$ CS-matrix $\mathbb{S}[\pmb{\kappa}]$ which is an unstable-negative feedback and  it has a $(k-1)\times (k-1)$ Hurwitz-stable principal submatrix $\mathbb{S}[\pmb{\kappa}-1]$. 
\end{definition}

We remark again that an unstable-negative feedback never possesses real-positive eigenvalues but only complex conjugate pairs of eigenvalues with positive real part.  Moreover, unstable-negative feedbacks are never autocatalytic \cite{VasStad23}, i.e., this class of oscillatory cores never involve autocatalysis. To illustrate, we presented five examples of oscillatory cores of class II by constructing a negative-feedback analogous of the five autocataytic cores from \cite{blokhuis20}. See Table~\ref{tab:negmain} for an overview and Fig.~\ref{fig:recipe_2} for one detailed case, where we started again from the autocatalytic core depicted above \eqref{eq:autcorII}, and changed the first row into negative to turn the positive feedback into negative:
\begin{equation}\label{eq:negII}
  \begin{pmatrix}
        -1 & 0 & -1\\
        1 & -1 & 0 \\
        1 & 1 & -1
    \end{pmatrix},
\end{equation}
corresponding to reactions $\ce{X}\rightarrow \ce{Y}+\ce{Z}$, $\ce{Y}\rightarrow \ce{Z}$, $\ce{X}+\ce{Z}\rightarrow $, which however is Hurwitz-stable. To obtain an unstable-negative feedback, we split the first reaction 
\begin{equation}
    \ce{X}\rightarrow \ce{Y}+\ce{Z}
\end{equation}
into sufficiently many intermediates, in this case six,
\begin{equation}
\ce{X}\rightarrow\ce{X}_1\rightarrow\ce{X}_2\rightarrow\ce{X}_3\rightarrow\ce{X}_4\rightarrow\ce{X}_5\rightarrow\ce{X}_6 \rightarrow\ce{Y}+\ce{Z},
\end{equation} leading to an unstable-negative feedback:$$
\mathbb{S}[\pmb{\kappa}]=
\begin{pmatrix}
-1 & 0 & 0 & 0 & 0 & 0 & 0 & 0 & -1 \\
1 & -1 & 0 & 0 & 0 & 0 & 0 & 0 & 0 \\
0 & 1 & -1 & 0 & 0 & 0 & 0 & 0 & 0 \\
0 & 0 & 1 & -1 & 0 & 0 & 0 & 0 & 0 \\
0 & 0 & 0 & 1 & -1 & 0 & 0 & 0 & 0 \\
0 & 0 & 0 & 0 & 1 & -1 & 0 & 0 & 0 \\
0 & 0 & 0 & 0 & 0 & 1 & -1 & 0 & 0 \\
0 & 0 & 0 & 0 & 0 & 0 & 1 & -1 & 0 \\
0 & 0 & 0 & 0 & 0 & 0 & 1 & 1 & -1
\end{pmatrix}, 
$$
with eigenvalues approximately
($\mathbf{+0.0094 \pm 0.39i}$, $-1.82$, $-1.79\pm 0.5i$, $-1.28 \pm 0.98i$,$-0.53\pm 0.95i)$. It is worth noting that the negative-feedback instability may appear only upon inclusion of a sufficient number of intermediate steps: we call this feature \textbf{`principle of length'} and we discuss it further in Sec.~\ref{sec:principlelength}.  
A complete analysis of the five oscillatory cores of class II we constructed is provided in the SM, Sec.~\ref{sec:recipeII}. The presence of an oscillatory core of class II defines what we refer to as \textbf{Recipe \ref{recipe:2main}} for oscillations:
\begin{customthm}{II}\label{recipe:2main}
Assume that the network possesses an Oscillatory Core of Class II. Then the system admits nonstationary periodic orbits.    
\end{customthm}

\begin{table}[h!]
    \centering
    \begin{tabular}{c|c|c|c}
    Oscillatory Core (II,a) &         $\begin{pmatrix}
        -1 & -2\\
         1 & -1\\
        \end{pmatrix}$ 
        $n=6$ 
        &  $\vcenter{\hbox{\includegraphics[scale=0.3]{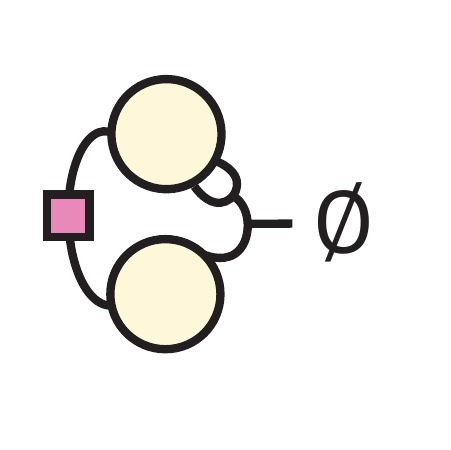}}}$
        &$\vcenter{\hbox{\includegraphics[scale=0.3]{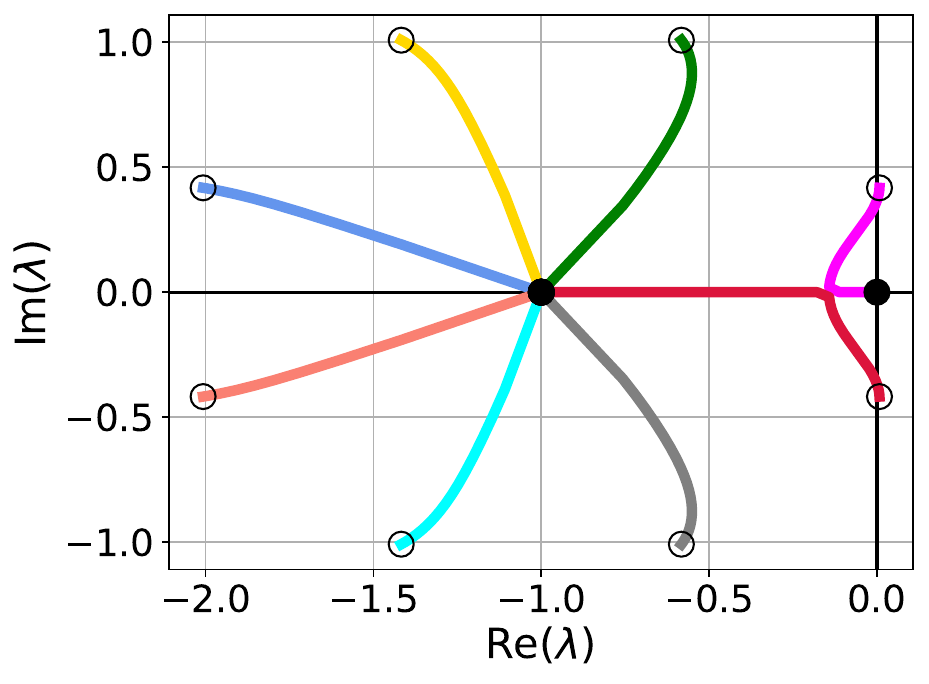}}}$ \\
        \hline      
    \\
      Oscillatory Core (II,b) &  $\begin{pmatrix}
        -1 & 0 & -1\\
         1 & -1 & 0\\
         1 & 1 & -1
        \end{pmatrix}$
         $n=6$ &  $\vcenter{\hbox{\includegraphics[scale=0.3]{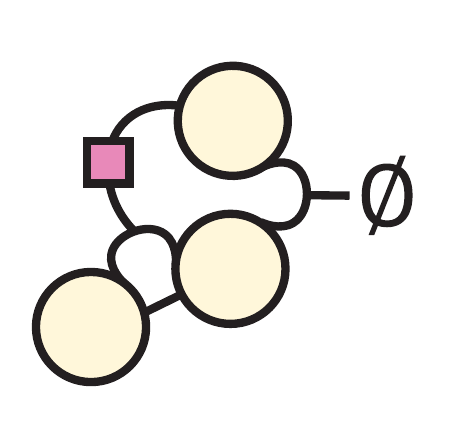}}}$ 
         
 &$\vcenter{\hbox{\includegraphics[scale=0.3]{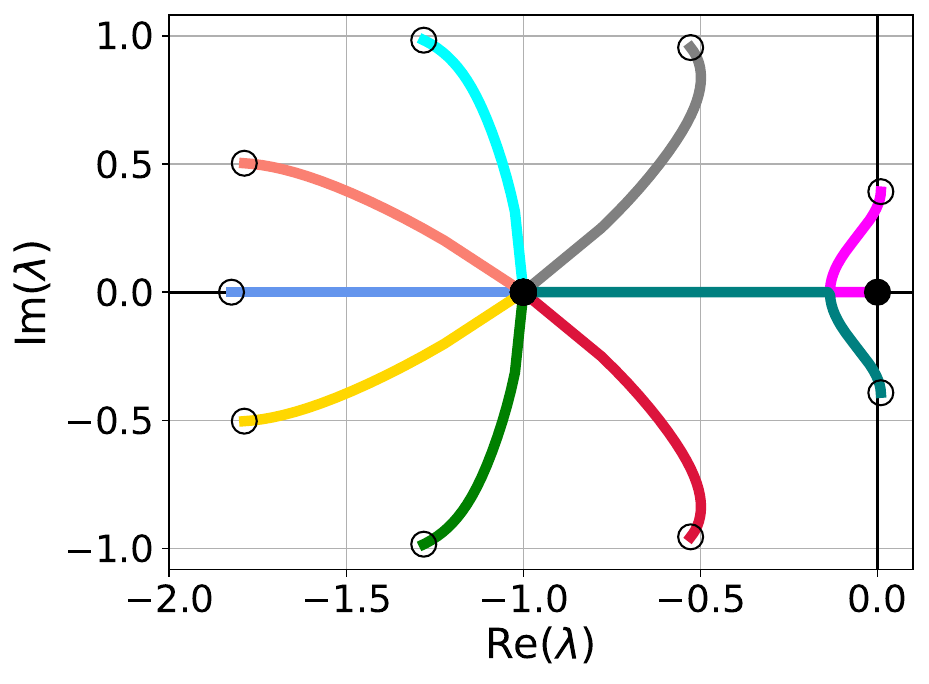}}}$ \\
        \hline      \\
       Oscillatory Core (II,c)& $\begin{pmatrix}
        -1 & -1 & -1\\
         1 & -1 & 0\\
         1 & 0 & -1
        \end{pmatrix}$ 
        $n=6$
         &  $\vcenter{\hbox{\includegraphics[scale=0.3]{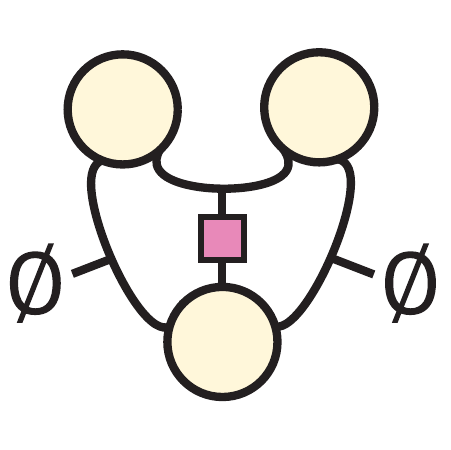}}}$   
          &$\vcenter{\hbox{\includegraphics[scale=0.3]{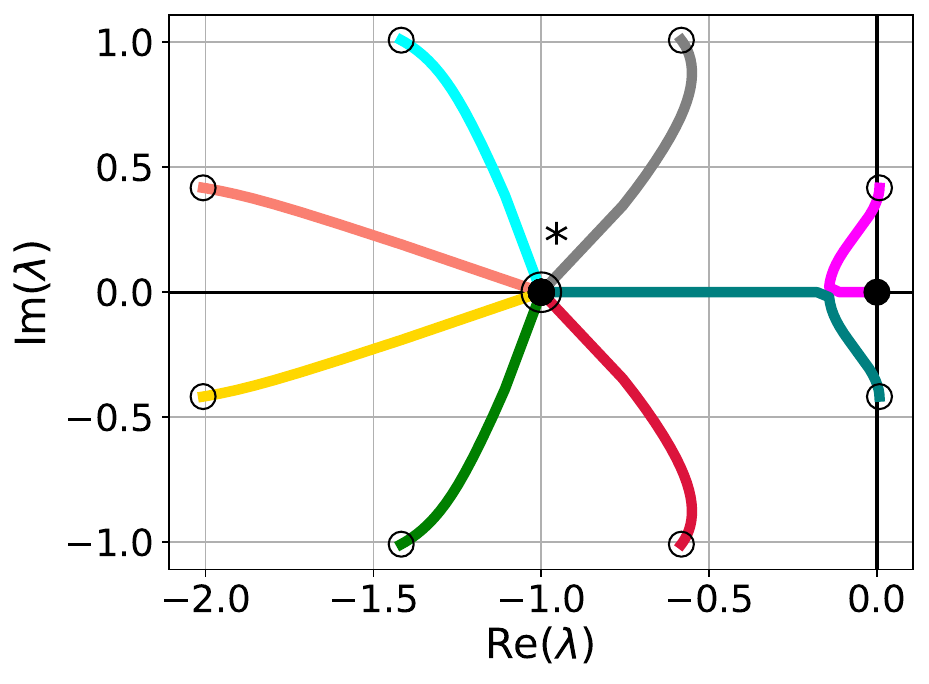}}}$ \\
        \hline      \\
         Oscillatory Core (II,d)&  $\begin{pmatrix}
        -1 & -1 & -1\\
         1 & -1 & 0\\
         1 & 1 & -1
        \end{pmatrix}$ 
       $n=3$
        &  $\vcenter{\hbox{\includegraphics[scale=0.3]{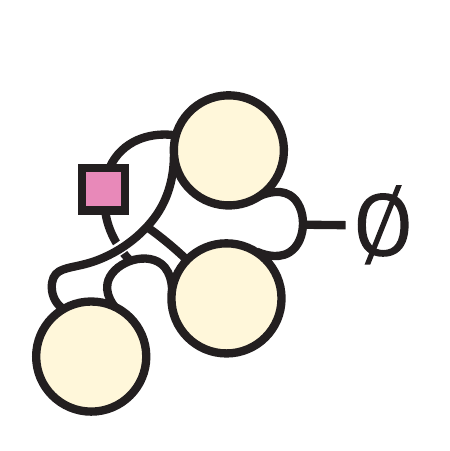}}}$  
         &$\vcenter{\hbox{\includegraphics[scale=0.3]{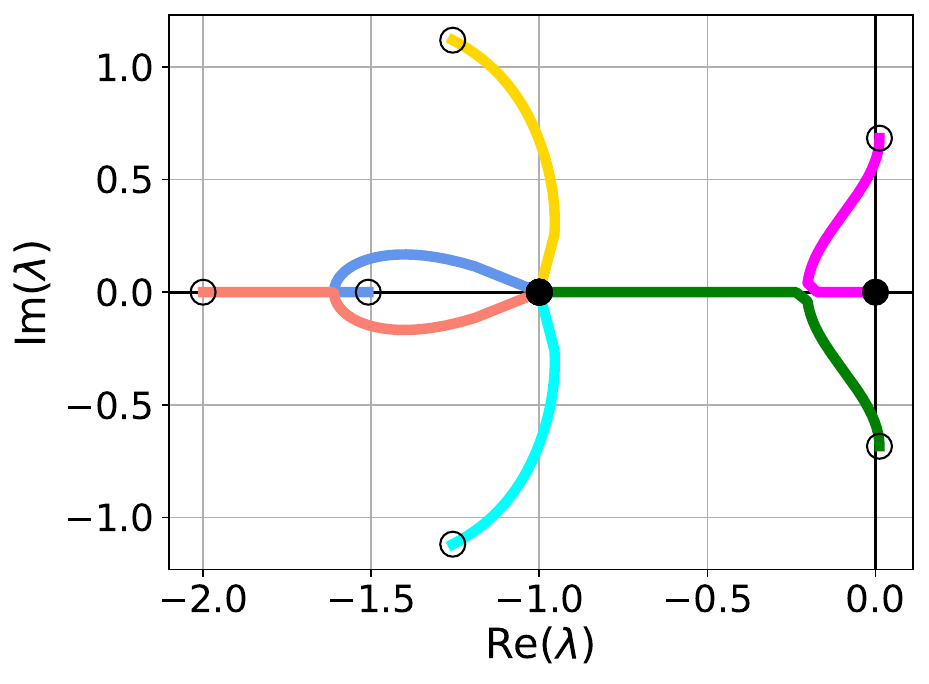}}}$ \\
        \hline    \\
       Oscillatory Core (II,e)&  $\begin{pmatrix}
        -1 & -1 & -1\\
         1 & -1 & 1\\
         1 & 1 & -1
        \end{pmatrix}$ 
     $n=1$  &  $\vcenter{\hbox{\includegraphics[scale=0.3]{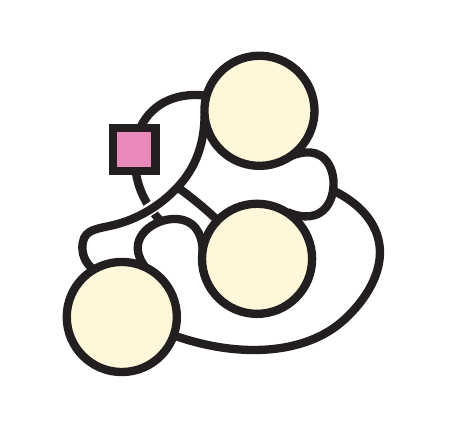}}}$   &$\vcenter{\hbox{\includegraphics[scale=0.3]{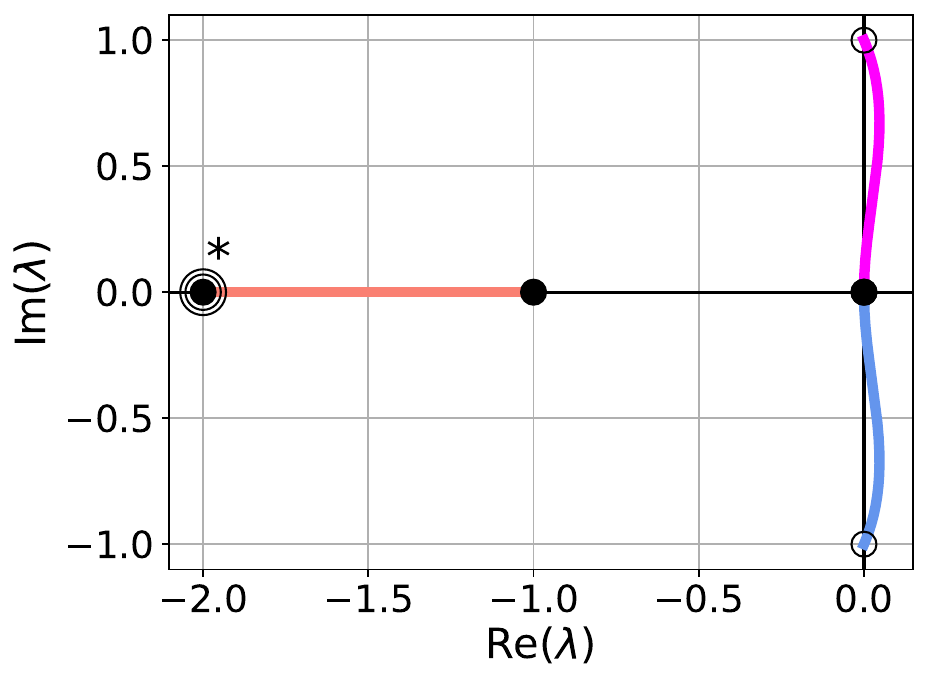}}}$
    \end{tabular}
\caption{Oscillatory cores (II,a) to (II,e). A pink square in the hypergraph denotes a succession of monomolecular steps, which according to the principle of length need to exceed some minimum number of $n$ steps to admit oscillations. This critical number $n$ is given for each core. Eigenvalues are plotted for the oscillatory cores, which are - as in Table I - multiplied by a diagonal unit matrix with one variable element $\beta$, corresponding to a unimolecular reaction.
An eigenvalue that remains fixed is indicated with an asterisk for Type IIc (-1) and Type IIe (-2) corresponding to an eigenvector $(0,...0,1,-1)^T$: the  reactions associated to such an eigenvector are unaltered by multiplication with our diagonal matrix. }
\label{tab:negmain}
\end{table}

The advantage of considering oscillatory cores as in Recipe \ref{recipe:1main} and Recipe \ref{recipe:2main}
is that the underlying structural mechanism triggering oscillations can be
encapsulated within a single CS-matrix $\mathbb{S}[\pmb{\kappa}]$.
%In our opinion, such
This level of abstraction is beneficial in understanding qualitative
features of oscillators. The bifurcation process itself can be visualized
in both cases by simply multiplying $\mathbb{S}[\pmb{\kappa}]$ by a
diagonal matrix of the form $\operatorname{diag}(1,1,1,1,1,\beta)$, where
the entry $\beta$ is chosen appropriately based on the stability properties
of the principal submatrices of $\mathbb{S}[\pmb{\kappa}] $, see Fig.~\ref{fig:recipe_1} and Fig.~\ref{fig:recipe_2}.  While for $\beta=1$ the product
reduces to $\mathbb{S}[\pmb{\kappa}]$ itself, choosing $\beta$ sufficiently
small induces a change in the stability of the product, while its
invertibility is preserved for all choices of $\beta>0$.  By
  Lemma~\ref{lem:Dhopfmain}, any change of stability implies purely
imaginary eigenvalues, and thus oscillations by
  Thm.~\ref{thm:mainmain}.

\begin{figure}%[tbhp]
\centering
\includegraphics{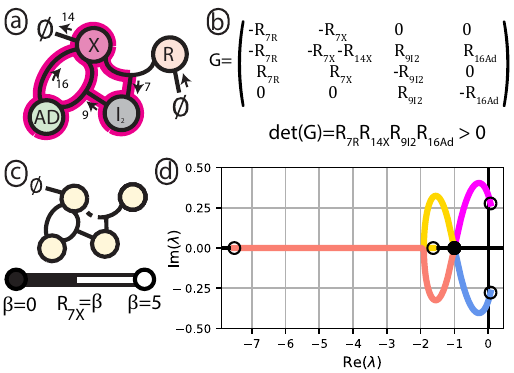}
\caption{An example of Recipe \ref{recipe:0main}. a) A subnetwork of the Activation-Inhibition  model from \cite{Nguyen18} with autocatalysis highlighted in magenta. b) The Jacobian $G$ is  invertible for all parameter choices. c) Keeping the other reactivities fixed at 1, we vary $R_{7X}=\beta\in(0,5]$ d) the trajectories of the eigenvalues of $G$ illustrate the Hopf bifurcation.}
\label{fig:recipe_0}
\end{figure}

\subsection{Recipe \ref{recipe:0main}: oscillations without employing oscillatory cores}\label{sec:recipe0main} 
A broader scan of the literature (see Sec.~\ref{sec:literature} of the
SM) on chemical oscillators reveals that in many cases the mechanism
triggering oscillations cannot be captured so compactly by a single
CS-matrix. Instead, it emerges from a more intricate interplay among
multiple components of the network. To account for these situations, we
also formulate a more general, less matrix-centered condition, see
Fig.~\ref{fig:cores}c and Fig.~\ref{fig:recipe_0}. We state it here an informally and intuitively. For a rigorous statement with proof we refer to Sec.~\ref{sec:recipe0main} below.

\begin{customthm}{0}\label{recipe:0main}
  A transition between regions of parameter space with stable and unstable
  steady states, under the assumption that the Jacobian remains invertible
  throughout.
\end{customthm}

As a comparison, the capacity for multistationarity has long been
associated with the possibility of Jacobian singularity \cite{BaPa16}. In
contrast, for monostationary networks -- where this capacity is absent --
Recipe \ref{recipe:0main} effectively equates the capacity to switch between stability and
instability with the potential to exhibit periodic behavior.

\paragraph{A minimalist template.}
The more abstract statement of Recipe \ref{recipe:0main} can be nevertheless applied to
readily construct new classes of oscillators. A possible minimalist
template (Fig.~\ref{fig:cores}c) for the construction of oscillators
consists of (i) an autocatalytic core, (ii) replenishment of an off-core
reactant consumed by an autocatalyst, (iii) degradation of the same
autocatalyst species, offering a surprisingly simple template for building
oscillators. See Sec.~\ref{sec:recipe0basicconstr} in the SM for a
detailed analysis of this construction.

Under mass-action, however, such template in Fig.~\ref{fig:cores}c is not
quite enough due to inherent parametric dependencies. Yet, these
constraints can be lifted by a slight further decoration of the reaction
network, for instance through an independent second degradation mechanism,
or by replacing first order degradation by a catalytic cycle
(Fig.~\ref{fig:recipe0_MA_main}). This result highlights how parameter-rich
kinetics provides a first facile means to study the oscillatory question in
general without extensive algebraic pre-processing. Subsequently, it might
be possible as well to identify what is needed to accommodate the result
under the constraints of mass-action.
%See \cite{Vassena2025} for a treatment of similar problem under mass
%action kinetics.

\begin{figure} %[tbhp]
\centering
\includegraphics{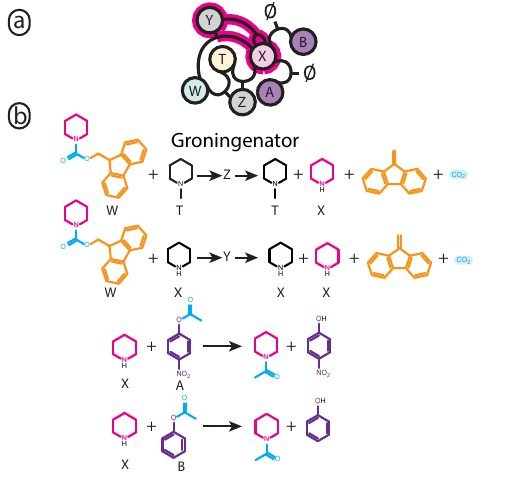}
\caption{b) Proposed reactions with molecular structures for the Groningenator \cite{Harmsel2023}. a) A subnetwork, with autocatalytisis highlighted in magenta. A capacity to oscillate can be demonstrated through Recipe \ref{recipe:0main}. Under mass-action, both degradation steps are necessary (See the SM sec~\ref{sec:activatorinhibitor}).}
\label{fig:groningenator_example}
\end{figure}

\paragraph{Recipe \ref{recipe:0main} in the literature.}
Fig.~\ref{fig:recipe_0} illustrates an oscillatory mechanism via Recipe \ref{recipe:0main},
detected in an Activation-Inhibition model \cite{Nguyen18}. In this example,
the symbolic Jacobian remains invertible for any choice of positive
reactivities. Notably, tuning the reactivity $R_{7x}$ -- which is involved
in an autocatalytic core -- does not affect the value of the determinant,
yet it can amplify the role of autocatalysis, rendering it dominant. As a
result, the system can transition to instability through the emergence of
purely imaginary eigenvalues, leading to periodic orbits. See also
Sec.~\ref{sec:activatorinhibitor} in the SM.

Fig \ref{fig:groningenator_example}b shows a set of proposed irreversible
reactions describing a recent small-molecule oscillator developed in
Groningen \cite{Harmsel2023}, which we refer to as the
\emph{Groningenator}. Upon omission of irreversibly produced side products
we obtain the reaction network in Fig \ref{fig:groningenator_example}a. The
Groningenator provides a further illustration of Recipe \ref{recipe:0main}. Again, it is
autocatalysis that ensures instability.  Degradation can counterbalance
autocatalysis, but under mass-action too simple degradation does not enable
oscillations. Here, both degradation pathways are essential: removing
either one of them (i.e., removing $\ce{A}$ or $\ce{B}$) removes the
capacity for oscillations, see also the SM, Sec.~\ref{sec:groning}.

Other literature examples have been analyzed in the SM under these lenses:
a Brussellator \cite{Prigogine1968} (Sec.~\ref{sec:bruss} in the SM), an
Oregonator \cite{Field1974} (Sec.~\ref{sec:oreg} in the SM), and a
general motif ubiquitous in mathematical epidemiology and Lotka--Volterra
models (Sec.~\ref{ex:ME} in the SM).

\begin{figure}%[tbhp]
\centering
\includegraphics{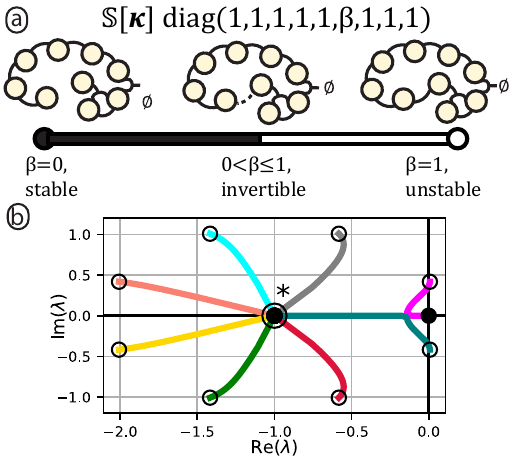}
\caption{Recipe \ref{recipe:2main} illustrated for the motif $\ce{X}\rightarrow\ce{X}_1\rightarrow...\rightarrow \ce{X}_6\rightarrow \ce{Y}+\ce{Z}$; $\ce{Y}\rightarrow \ce{Z}$; $\ce{X}+\ce{Z}\rightarrow \;$. Any consistent reaction network that includes such motif has the capacity for periodic oscillations. The bifurcation process follows a parameter $\beta\in[0,1]$. At $\beta=1$ the associated CS-matrix is Hurwitz-unstable, while at $\beta\approx0$ is Hurwitz-stable. The product matrix is invertible throughout: change of stability happens at purely imaginary eigenvalues. a) The bifurcation process is intuitively illustrated in the network as a removal of the reactivity of one reaction. b) The trajectory of the eigenvalues is displayed, a purely imaginary crossing can be seen.}
\label{fig:recipe_2}
\end{figure}

%\section{Net-stoichiometry and hidden catalysis}

\section{The principle of length}\label{sec:principlelength} 

Our theory implies that certain periodic oscillations may emerge only (!) upon the inclusion of a sufficient number of intermediate steps. Specifically, there exist non-oscillatory reaction networks containing a reaction
\begin{equation}
\ce{X} \quad \rightarrow \quad \ce{Y},
\end{equation}
which become oscillatory when this reaction is split into sufficiently many intermediate steps
\begin{equation}\label{eq:intermediates}
\ce{X} \quad \rightarrow \quad \ce{X}_1 \quad \rightarrow \quad \cdots \quad \rightarrow \quad \ce{X}_n \quad \rightarrow \quad \ce{Y},
\end{equation} 
making it then possible to choose kinetic constants that induce oscillations. We refer to this phenomenon as the \textbf{principle of length}. The principle challenges a commonly held axiom in working chemistry: that the presence or absence of intermediate steps does not fundamentally alter a system’s behavior.

We found that oscillatory cores of class II frequently follow the
\emph{principle of length}: they are structured as a negative feedback with
a variable number of intermediate steps of similar duration, and
oscillations arise only when the number of steps exceeds a certain critical
threshold. The underlying mathematical intuition is that instability in
such unstable-negative feedbacks depends not on the sign of the determinant
of the associated CS-matrix, but on the actual computation of the
eigenvalues, which is sensitive to the number of inserted intermediates. In
contrast, the instability of unstable-positive feedbacks (as in the
oscillatory cores of class I) is typically associated with a real positive
eigenvalue that can be detected through a determinant condition, which is
insensitive to the insertion of intermediate steps as
\eqref{eq:intermediates}.  Finally, it is worth noting that the critical
number of steps may depend on whether the network is endowed with
parameter-rich kinetics or mass-action: e.g.\ the oscillatory core in
Fig.~\ref{fig:recipe_2} requires at least $n_{PR}=6$ parameter-rich steps
but at least $n_{MA}=9$ mass-action steps \cite{Vassena2025}, due to the
additional presence of steady-state constraints.

On a theoretical basis, the principle of length has been previously
observed in the context of feedback loops \cite{MPSmith90,
  HofbauerHypercycle91, F19}, where a minimal loop length has been shown to
be a necessary condition for the emergence of oscillations. In the other
direction, a body of work by Banaji and co-authors on \emph{inheritance}
has focused on understanding which \emph{enlargements} of a network
preserve certain dynamical behaviors, such as multistationarity,
oscillations, and bifurcations; see \cite{ba23splitting} and references
therein. In particular, under mild technical conditions,
\cite{ba23splitting} shows that adding intermediates `preserves' the
original dynamical behavior. Yet, inheritance results are local: they
ensure that a larger network retains a feature of a smaller one only for
specific choices of parameters -- namely, when the added step is fast
enough to be essentially invisible (e.g. $\beta \approx 0$ in
Fig.~\ref{fig:recipe_2}). Thus, without contradiction, oscillations may
emerge for other parameter choices (e.g., $\beta \approx 1$ in
Fig.~\ref{fig:recipe_2}).

In experimental chemistry, to our knowledge, autocatalysis-free oscillators of this type have not been reported. However, a multiplicity of steps is not foreign to the literature: reactions are described on a variety of scales, going from overall reactions to detailed mechanisms involving many steps and intermediates. These steps tend to include vastly different timescales enabling a good approximation in terms of one or few rate-limiting steps. In this sense, chemical phenomenology appears robust to the introduction of many steps. This robustness no longer applies when we have many (slow) steps of similar timescale. Repeated steps of similar timescale can e.g. be realized in template-based polymerization, where they introduce a form of memory and delay \cite{lin2025biomasstransferautocatalyticreaction,blokhuis2020errthresh} and precise timing. This has been used to prepare synthetic molecular clocks \cite{Johnson-Buck2017}. In cell biology,  delayed negative-feedback due to introns is found to be essential for functional oscillations in the segmentation clock \cite{Takashima2011}.

%In chemistry, reactions are described on a variety of scales, going from overall reactions to detailed reaction mechanisms involving many steps and intermediates. For these descriptions to phenomenologically agree with each other, one might be tempted to posit a 'magnification conjecture':  adding extra (unimolecular) steps and intermediates to a description is inconsequential. Type II oscillators disprove this conjecture as they follow what we call the principle of length:
%type II oscillators become capable of oscillations when the number of steps (of similar duration) they contain surpasses a critical number. 

%This finding may appear nonintuitive in light of a perceived harmony among different coarse-grained levels of descriptions in chemistry. It should be stressed, however, that detailed descriptions in chemistry tend to include vastly different timescales. The principle of length involves several steps of similar timescales, which de facto introduces a form of memory and delay\cite{lin2025biomasstransferautocatalyticreaction,blokhuis2020errthresh} and has e.g. been used to make molecular clocks\cite{Johnson-Buck2017}. 

%It is tempting to think that periodic oscillations require catalysis: a combination of reactions eventually restores the same system composition, implying the presence of at least one cycle. How  

\section{Roles of catalysis}\label{sec:catalysis}

%\textcolor{purple}{Todo: write about role of catalysis, making the following points
%1. PM (using MM) can overcome Mass-Action restrictions by mimicking catalysis in a certain fast-exchange limit for reactions that were previously not catalyzed in that way.
%2. To lift certain PM oscillators to Mass-Action ones, we found that a catalytic degradation could do the trick.
%3. An entire class of oscillators is defined by autocatalysis.
%4. Within parameter-rich CRNs, catalytic stoichiometry can make the difference between presence or absence of oscillations.
%Catalysis is not necessary: OC Ia-c demonstrate that oscillators exist that make no use of catalysis. 
%5. Even then, catalysis has a role to play: catalysis can lower the threshold length required for oscillations.}

Catalysis appears to play an ubiquitous role in chemical oscillators, and our results identify a variety of mechanisms by which catalysis promotes oscillations: 
\begin{enumerate}[label=\Roman*.]
  \item \textbf{Lifting constraints}: 
A natural and valid choice of parameter-rich kinetics is to use Michaelis--Menten rate laws, which effectively - though implicitly - treat all reactions as catalyzed. Conversely, transforming parameter-rich models into mass-action ones can be done by expressing certain reactions through catalyzed elementary steps. See Fig.~\ref{fig:recipe0_MA_main}, where a minimalist template for oscillations via Recipe \ref{recipe:0main} in parameter-rich kinetics is transformed into a mass-action model by replacing a linear, uncatalyzed degradation reaction with a catalyzed one - thereby enabling oscillations.
\item \textbf{Providing instability}: When a catalyst promotes a reaction in which it is itself a (net) product, instability may arise due to autocatalysis or other feedback-driven mechanisms. As a result, periodic oscillations can occur even in systems with seemingly innocuous net stoichiometry. See Sec.~\ref{ex:hidcat} in the SM for an example: a simple monomolecular loop $\ce{X}_1\rightarrow\ce{X}_2\rightarrow\ce{X}_3\rightarrow\ce{X}_4 \rightarrow \ce{X}_5 \rightarrow \ce{X}_1$ becomes oscillatory upon the addition of `hidden' catalysis, i.e. invisible to net stoichiometry.
\item \textbf{Length}: As discussed in Sec.~\ref{sec:principlelength} for the example in Fig.~\ref{fig:recipe_2}, the critical number $n$ of intermediate steps $\rightarrow\ce{X}_1\rightarrow \ce{X}_2\rightarrow ...\rightarrow \ce{X}_n\rightarrow$ required to sustain oscillations can depend on whether the kinetics is parameter-rich or mass-action. In particular, implementing catalytic steps using, for example, a Michaelis--Menten scheme can reduce the threshold length needed for oscillations.
\end{enumerate}
On the other hand, the existence of mass-action examples based on Recipe \ref{recipe:2main} - such as the one in Fig.~\ref{fig:recipe_2} with $n = 9$ intermediate steps \cite{Vassena2025} - shows as well that oscillations can also arise without any catalysis.

\section{Outlook: building new oscillators}\label{sec:buildingosci}

Through a variety of techniques, we have demonstrated different mechanisms by which oscillations can emerge and underlying structural requirements. These insights furnish design considerations for new chemical oscillators. 

Firstly, there is untapped oscillatory potential in supramolecular chemistry, because some structural ingredients from either recipe are inherent to it. The (slow) formation of oligomers may grant access to regimes where a principle of length may readily operate. Furthermore, the formation of supramolecular species can proceed autocatalytically, supramolecular phases (e.g. micelles) may exhibit catalysis of their own.  A small number of supramolecular oscillators has to date been reported, e.g. involving tubulin fibers \cite{Horio1986,Hess2017}, fibers and colloids of perylene diimide \cite{Leira-Iglesias2018}, dipeptide fibers with SDS micelles  \cite{Torigoe2024}, disulfide surfactant micelles  \cite{Howlett2022} and Palladium nanoparticles \cite{Donlon2014,Lodge2025}. 

Secondly, we find that the extension of an oscillator under parameter-richness to a mass-action oscillator requires the creation of enough free variables to lift constraints. This can e.g. be achieved by using catalyzed reactions instead, or having multiple independent pathways achieving a similar transformation (for instance the pair of degradation reactions in the Groningenator described in \cite{Harmsel2023}). 

Thirdly, the conditions by which a system is maintained out of equilibrium can strongly influence the capacity for oscillations. Oscillations in closed systems (or stretched in time under plug flow \cite{Hermans2015}) must eventually vanish as thermodynamic forces annul, but conversely, not all means of maintaining thermodynamic forces necessarily preserve oscillations. For instance, if an oscillating species that enters through a slow influx is instead fixed in concentration (chemostatted), this can remove a previously present oscillation. %; see the SM, Sec.~\ref{sec:chemostatnoosc}. In Sec.~\ref{section:twocs} and Fig.~\ref{fig:twoslowcs} of the SM, we consider an example where a species is exchanged with a reservoir at a tunable rate, and find that the faster the exchange timescale is (relative to other timescales), the more suppressed the oscillations become. 
In designing oscillators using reservoirs (e.g., hydrophobic reactant in an oil layer \cite{Howlett2022}, oxygen in air \cite{Roelofs1983}), the timescale of exchange can thus become an important parameter. An oscillator employing benzaldehyde, cobalt, bromide, and oxygen under gas flow \cite{Jensen1983} has been adapted to slow chemostat conditions \cite{Roelofs1983}.

\begin{figure}[!tbhp]
\centering
\includegraphics{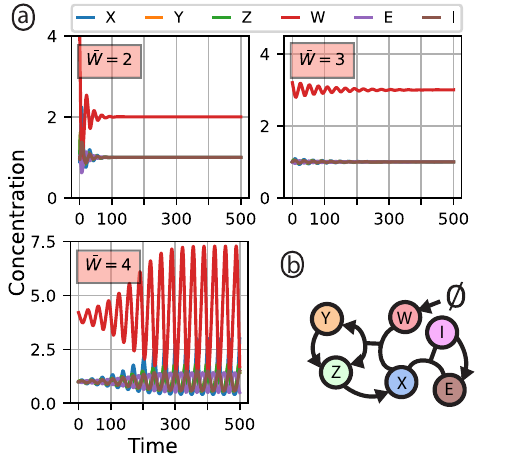}
\caption{Illustration of a simple oscillator via Recipe \ref{recipe:0main} for a mass-action network, following the minimalist template presented in Sec.~\ref{sec:recipe0main}.  a) Simulation of ODEs (see the SM, Sec.~\ref{sec:recipe0basicconstr}) with initial conditions $[\ce{X,Y,Z,W,E,I}] = (1,1,1,\bar{W}+0.2,1,1)$. b) The oscillator consists of an autocatalytic core (autocatalysts $\ce{X,Y,Z}$), a limiting reagent $\ce{W}$ supplied externally, and a catalyzed degradation reaction. 
For mass-action kinetics, fixing the concentration of $\ce{W}$ or having a linear, uncatalyzed degradation of $\ce{X}$ removes the capacity to oscillate.}
\label{fig:recipe0_MA_main}
\end{figure}

\section{Discussion}\label{sec:discussion}

We presented a stoichiometric framework for studying oscillations in chemical reaction networks. This structural analysis is enabled by parameter-rich kinetics, an approach that disentangles parametric dependencies from purely stoichiometric considerations. It leads to the structural notion of an oscillatory core: a network structure with the capacity to oscillate that contains no smaller substructure with this property. Many molecular properties (e.g., charge, acidity, aromaticity) can be rationalized in terms of functional groups: small substructures of the full molecular structure. For reaction networks, the same principle appears to hold, with several dynamical properties attributable to minimal modules (cores) \cite{blokhuis20,VasStad23}.
In particular, when \emph{unstable} cores become dominant, they induce a loss of steady-state stability, giving rise to a dichotomy between the onset of multistationarity and the emergence of periodic oscillations. For the latter case, oscillatory cores provide sufficient stoichiometric conditions.

Oscillatory cores come in two classes, mechanistically characterized by either positive feedback or negative feedback. The importance of autocatalysis has long been appreciated, and reported oscillators in synthetic chemistry thus predominantly rely on it. In contrast, oscillatory cores of class II -- introduced here -- never involve autocatalysis, but instead do require length: a minimum number of slow reaction steps is required for the capacity to oscillate (principle of length). Oscillators of this class have not yet been realized experimentally\footnote{We have not found synthetic experimental realizations of this class of oscillators in the chemistry literature.}, although theoretical examples are known \cite{hyver1978}. 

Throughout we identified several ways in which catalysis contributes to the emergence of periodic oscillations: it can induce the necessary instability, relax parametric dependencies, or reduce the length requirement. As an overall illustration, we constructed a minimalist template centered on an autocatalytic core, extended by the addition of an off-core reactant consumed by the autocatalyst together with a degradation mechanism for the same autocatalyst species. At the same time, we found that oscillators exist that involve no catalysis at all (i.e., neither allocatalysis \cite{blokhuis20} nor autocatalysis), for which synthetic counterparts remain to be prepared.

Our stoichiometric recipes and parameter-rich methodology broaden the toolkit for studying oscillations. We used these tools to provide qualitative and quantitative insights into  mechanisms, nonequilibrium thermodynamics and network structure underlying oscillations. Our approach furthermore predicts the existence of new types of chemical oscillators yet to be synthesized. The formalism lays the groundwork for investigating deeper questions about chemical oscillators, offering tools to start tackling them.

%While not conclusive for design purposes, the stoichiometric recipes we presented - i.e. the same mechanisms - nonetheless provide qualitative insight into the oscillatory behavior ofmass-action systems and may serve as guiding principles for future oscillator design.

\section{Proof of the main results}\label{sec:proofs}

Most results on periodic orbits in parametric systems rely on proving the occurrence of a \emph{local} Hopf bifurcation \cite{GuHo84}. While alternative methods exist, they often require restrictive assumptions that limit their general applicability. For example, the Poincaré--Bendixson criterion is limited to two-dimensional systems or to networks with very specific monotone feedback architectures \cite{MPSmith90}. A notable case of applicability of this elegant geometric argument is the hypercycle \cite{HofbauerHypercycle91}. More elaborate bifurcation scenarios, such as Takens--Bogdanov bifurcations, also imply the existence of a local Hopf bifurcation but simultaneously require multistationarity through a zero-eigenvalue bifurcation. As such, they fail to capture many oscillatory networks that do not exhibit multiple steady states - such as those arising from Recipe \ref{recipe:2main}. In summary, Takens--Bogdanov bifurcations are sufficient but not necessary for a local Hopf bifurcation. Other frameworks, such as hybrid bifurcations \cite{APNHJ23}, demand even more articulate structural conditions, which are rarely encountered in typical reaction networks. 

However, even detecting a local Hopf bifurcation presents major computational obstacles. It demands the identification of a steady state whose Jacobian has a pair of purely imaginary eigenvalues, typically via a Routh--Hurwitz criterion, a task already infeasible for networks with more than a handful of species. Moreover, asserting the existence of periodic orbits further requires checking non-resonance and transversality conditions. 

To overcome these limitations, we have applied here results from \emph{global} Hopf bifurcation theory, as developed by Chow, Mallet--Paret, and Yorke \cite{ChowMPYorkeGH78}, and further refined by Fiedler \cite{Fiedler85PhD}, who provided a powerful yet accessible criterion for establishing the existence of periodic orbits. In essence, Fiedler's result (see Thm.~\ref{thm:fiedlermain} below) asserts that any parametrized family of steady states - along which the Jacobian remains invertible but undergoes a stability change - necessarily implies the existence of periodic solutions: one needs only to verify the overall Jacobian invertibility and identify two steady-states with different stability. A particularly effective aspect of this result is that it bypasses the need to compute any bifurcation point explicitly. In turn, however, this method cannot provide information about the stability of the resulting periodic orbits. We argue that this limitation is not critical, since experimental observations of damping periodic oscillations can be interpreted and modeled in different ways. On one hand, they may correspond to stable periodic orbits in the idealized model that appear damped in experiments due to noise, thermodynamic effects, or open-system conditions. On the other hand, these oscillations could also arise from the unstable manifold of an inherently unstable periodic orbit within the model itself.

\subsection{The preliminary nondegeneracy condition} The main result of this contribution, Thm.~\ref{thm:mainmain}, applies to networks that are \emph{nondegenerate} in the following sense:
\begin{definition}[Nondegenerate networks]\label{def:nondegnetmain}
Let $\mathbf{\Gamma}$ be a reaction network with $|M|\times|E|$
stoichiometric matrix $\mathbb{S}$, and let $n=\operatorname{dim}\operatorname{ker}\mathbb{S}^T$. Then
$\mathbf{\Gamma}$ is nondegenerate if there exists a choice for the
reactivity matrix $R$ such that the symbolic Jacobian $G=\mathbb{S}R$
possesses a nonzero $(|M|-n)$-principal minor.
\end{definition}
Nondegenerate networks are those for which the trivial eigenvalues zero of the Jacobian all descend from the linear conservation laws identified as left-kernel vectors of the stoichiometric matrix $\mathbb{S}$. In particular, in absence of conservation laws, a nondegenerate network admits an invertible symbolic Jacobian determinant $G=\mathbb{S}R$ for some choice of $R$. In terms of CS-matrices, we have the following Lemma that characterizes nondegeneracy.
\begin{lemma}\label{lem:nondegcsmain}
A network $\mathbf{\Gamma}$ with stoichiometric matrix $\mathbb{S}$ such that $\operatorname{dim}\operatorname{ker}\mathbb{S}^T=n\ge0$ is nondegenerate if and only if there exists an invertible $(|M|-n)\times (|M|-n)$ CS-matrix.
\end{lemma}
\proof
The proof follows from the fact that each coefficient $c_k$ of the characteristic polynomial $g(\lambda)$ of the symbolic Jacobian $G=\mathbb{S}R$, 
\begin{equation}
   g(\lambda):=\operatorname{det}(\mathbb{S}R-\lambda \operatorname{Id})=\sum_{k=0}^{|M|}(-1)^kc_k\lambda^{|M|-k}
\end{equation}
can be expanded along Child Selections:
\begin{equation}\label{eq:expcsmain}
c_k=\sum_{\pmb{\kappa}}\operatorname{det}\mathbb{S}[\pmb{\kappa}]\prod_{\mathsf{m}\in\kappa}R_{J(m)m},
\end{equation}
where the sum runs on all $k$-CS triples $\pmb{\kappa}=(\kappa, E_\kappa, J)$. See Section 4 in \cite{VasStad23} for a detailed analysis of expansion \eqref{eq:expcsmain}. As $c_{|M|-n}$ can be seen as a multilinear homogenous polynomial of order $|M|-n$ in the positive variables $R_{jm}>0$, it follows that $c_{|M|-n}\equiv0$ if and only if all coefficients of the monomials, i.e. all $(|M|-n)\times (|M|-n)$ CS-matrices are singular, which proves the lemma.
\endproof

\subsection{Proof of Theorem \ref{thm:mainmain}}

The starting point for proving our main results is a theorem on
\emph{Global Hopf Bifurcation} by Fiedler which we recall here:
\begin{theorem}[(4.7) from \cite{Fiedler85PhD}]\label{thm:fiedlermain}
  Consider the ODE $\dot{\mathbf{x}}=f(\mathbf{x},\mu)$ with $\mathbf{x} \in \mathbb{R}^N$ and a
  parametric vector field $f(\mathbf{x},\mu)$ depending on a scalar parameter
  $\mu\in [a,b] \subseteq \mathbb{R}$ and assume the following conditions:
  \begin{enumerate}
  \item $f$ is analytic in $\mathbf{x}$ and $\mu$;
  \item A family of steady-states is parametrized over $\mu$, i.e. there
    exists $\bar{\mathbf{x}}(\mu)$ s.t.
    $$0=f(\bar{\mathbf{x}}(\mu),\mu)\quad\text{for}\quad \mu \in [a,b];$$
  \item The Jacobian of $\bar{\mathbf{x}}(\mu)$ is invertible for any $\mu$:
    $$\operatorname{det}f_x(\mathbf{x},\mu)|_{(\bar{\mathbf{x}}(\mu),\mu)}\neq 0 \quad \text{for}\quad \mu\in[a,b];$$
  \item Net-change of stability of the Jacobian:
    $$\operatorname{inertia}f_x(\mathbf{x},\mu)|_{(\bar{\mathbf{x}}(a),a)}\neq \operatorname{inertia}f_x(\mathbf{x},\mu)|_{(\bar{\mathbf{x}}(b),b)}.$$
  \end{enumerate}
  Then there exists $\mu^*$ such that $\dot{\mathbf{x}}=f(\mathbf{x},\mu^*)$  admits
  nonstationary periodic solutions.
\end{theorem}
Thm.~\ref{thm:fiedlermain} is stated in \cite{Fiedler85PhD} in wider
generalities for parabolic systems, and it has further conclusions that
more directly relate to the theory of global Hopf bifurcation.  For
simplicity of presentation, these are omitted in our tailored version.

The central step for the proof is stating a reduction lemma that
guarantees, via a perturbation argument, that Thm.~\ref{thm:fiedlermain}
applies to a nondegenerate network endowed with parameter-rich kinetics if
conditions 3 and 4 applies to a principal submatrix of its Jacobian. It is
the dynamics analogous to the linear algebra Lemma~\ref{lem:Dhopfmain},
with the further crucial delicacy that to apply Thm.~\ref{thm:fiedlermain} we
need to exclude further zero-eigenvalues. To do so, we rely on Lemma
\ref{lem:nondegcsmain}.

\begin{lemma}
  \label{lem:dynreductionmain}
  Let $\pmb{\Gamma}$ be a nondegenerate reaction network and $G[\kappa]$ a
  $k$-principal submatrix of its symbolic Jacobian $G=\mathbb{S}R$. Assume
  there exists an analytic path $\gamma(\mu)$ in parameter space,
  $\mu\in[0,1],$ such that $G[\kappa](\mu)$ is invertible for any $\mu$ and
  \begin{equation}
    \operatorname{inertia}G[\kappa](0)\neq \operatorname{inertia}G[\kappa](1).
  \end{equation}
  Then the associated dynamical system \eqref{eq:maineqmain} admits
  nonstationary periodic orbit.
\end{lemma}
\begin{proof}
  We aim to apply Thm.~\ref{thm:fiedlermain}. We first focus on conditions 3
  and 4, namely we will construct a parametrization of the symbolic
  Jacobian matrix $G(\mu)=\mathbb{S}R(\mu)$, $\mu\in[0,1]$ so that $G(\mu)$
  is invertible for any $\mu$ and
  \begin{equation}\label{eq:changeofinertiaproofmain}
    \operatorname{inertia}G(0)\neq \operatorname{inertia}G(1).
  \end{equation}
  We proceed incrementally by analyzing three cases of increasing
  generality and technical difficulty.

  If $|\kappa|=|M|$ there is nothing to prove, since $G[\kappa]$ is then
  just the full Jacobian $G$ and thus the assumptions of Lemma
  \ref{lem:dynreductionmain} become identical to conditions 3 and 4 in
  Thm.~\ref{thm:fiedlermain}.

  If the network possesses $n>0$ linearly independent conservation laws,
  i.e., $n=\operatorname{dim}\operatorname{ker}\mathbb{S}^T$, then the
  Jacobian $G$ cannot be invertible. The first simpler case is when
  $|\kappa|=|M|-n$. In this case, we shall consider the $|M|-n$ dimensional
  reduced dynamical system, with $n$ species removed, and whose Jacobian
  determinant corresponds indeed to the coefficient $c_{|M|-n}$ of the
  characteristic polynomial of $G$. Let $\{R_{jm}\}_{\mathsf{m}\in\kappa}$
  indicate the $d_\kappa$ nonzero parameters $R_{jm}$ that appear in
  $G[\kappa]$. Obviously, those are the $R_{jm}$ where $\mathsf{m}\in
  \kappa$ and $j$ is a reaction where $\mathsf{m}$ participates as a
  reactant. Without loss of generality, we can consider the analytic path
  $\gamma(\mu)$ to be defined only for $\{R_{jm}\}_{\mathsf{m}\in\kappa}$,
  i.e. on $\mathbb{R}^{d_k}_{>0}$.  Let now $\{R_{jm}\}$ indicate all
  nonzero parameters $R_{jm}$ appearing in $R$, ordered so that the
  $d_\kappa$ parameters associated to $\{R_{jm}\}_{\mathsf{m}\in\kappa}$
  appear first.  We embed the above path $\gamma(\mu)$, defined on
  $\{R_{jm}\}_{\mathsf{m}\in\kappa}$, i.e. on $\mathbb{R}^{d_k}_{>0}$, to
  $R$, i.e. $\mathbb{R}^d_{>0}$, by defining a path
  $\tilde{\gamma}(\mu,\varepsilon):[0,1]\times \mathbb{R}_{>0}\mapsto
  \mathbb{R}^d_{>0}$ as follows:
  \begin{equation}
    \tilde{\gamma}(t,\varepsilon)_h=\begin{cases}
    \gamma(t)_h \quad &\text{if $h\in {1,...,d_\kappa}$};\\
    \varepsilon\quad\quad \; &\text{otherwise}\\
    \end{cases}
  \end{equation}
  For $\varepsilon=0$, $c_{|M|-n}=\operatorname{det}G[\kappa]$, and thus
  the Jacobian of the reduced system is invertible. By continuity of the
  eigenvalues with respect to the entry of the matrix, the assumptions of
  the Lemma guarantees there is a parametrization $\tilde{\gamma}(\mu)$ for
  the Jacobian of the reduced system so that
  \eqref{eq:changeofinertiaproofmain} holds along an invertible path.

  The final general case is when $|\kappa|<|M|-n$. Since the network is
  nondegenerate, via Lemma \ref{lem:nondegcsmain}, there exists at least one invertible $(|M|-n) \times
  (|M|-n)$ CS-matrix. We pick one of such CS-matrices, and call it
  $\mathbb{S}[\bar{\pmb{\kappa}}_{|M|-n}]$ with associated $(|M|-n)$-CS
  $\bar{\pmb{\kappa}}_{|M|-n}=(\bar{\kappa}_{|M|-n},
  \bar{E}_{\bar{\kappa}_{|M|-n}}, \bar{J})$. As above, we embed the path
  $\gamma(\mu)$, defined on $\{R_{jm}\}_{\mathsf{m}\in\kappa}$, i.e. on
  $\mathbb{R}^{d_k}_{>0}$, to $R$, i.e. $\mathbb{R}^d_{>0}$, by defining a
  path $\tilde{\tilde{\gamma}}(\mu,\varepsilon,\eta):[0,1]\times (0,1)
  \times \mathbb{R}_{>0}\mapsto \mathbb{R}^d_{>0}$ as follows
  \begin{equation}
    \tilde{\tilde{\gamma}}(t,\varepsilon,\eta)_h=
    \begin{cases}
      \gamma(t)_h \quad&\text{if $h\in \{1,...,d_\kappa\}$};\\
      \varepsilon\quad\quad\quad
      &\text{if $h\not\in \{1,...,d_\kappa\}$ but $h$ corresponds to $R_{jm}=R_{\bar{J}(m)m}$, with $m\in \bar{\kappa}_{|M|-n}$};\\
      \varepsilon^\eta\quad\quad\quad
      &\text{otherwise.}\\
    \end{cases}
  \end{equation}
  Clearly, such a rescaling guarantees that for sufficiently large $\eta$
  and any $\varepsilon<1$ we have.
  \begin{equation}
    c_{|M|-n}\approx
    \operatorname{det}\bar{S}[\pmb{\kappa}_{|M|-n}]\prod_{\mathsf{m}\in
      \kappa_{|M|-n}}R_{J(m)m}\neq 0.
  \end{equation}
  Thus invertibility of the Jacobian of the reduced system is guaranteed
  for any $\varepsilon<1$. In turn, by continuity of eigenvalues, for
  $\varepsilon$ small enough the assumptions of this Lemma guarantees that
  \eqref{eq:changeofinertiaproofmain} holds, along a path with invertible
  Jacobian.

  To conclude, we just leverage Def.~\ref{def:prichmain} of parameter-rich
  kinetic models. Indeed, we can fix any positive value for the steady
  state $\bar{\mathbf{x}}>0$, and based on the path $\gamma(\mu)$ (or
  $\tilde{\gamma}(\mu)$, $\tilde{\tilde{\gamma}}(\mu)$ respectively), we
  find a parametrization of a family of a steady-states along
  $\gamma(t)$. Since $\mathbf{r}(\mathbf{x})$ and $\gamma(\mu)$ are analytic in
  their domain of definition, we can apply Thm.~\ref{thm:fiedlermain} and
  conclude the existence of nonstationary periodic orbits.
\end{proof}

%% The key technical tool for establishing the existence of nonstationary
%% periodic orbits is global Hopf bifurcation \cite{Fiedler85PhD}. The
%% presentation of this argument from nonlinear dynamics is deferred to
%% Section \ref{sec:proof}, where the full proof of Thm.~\ref{thm:main} is
%% also presented. However,

We next show that the Jacobian $G$ of a reaction network $\pmb{\Gamma}$ can
be parametrized as a product $\tilde{G}D$, where $D$ is a parametric
positive diagonal matrix:

\begin{lemma}\label{lem:mainGDmain}
  Let $G=\mathbb{S}R$ be the symbolic Jacobian of a reaction network, with
  $\mathbb{S}$ stoichiometric matrix, and $R$ parametric reactivity
  matrix. Then, there exists a rescale of $R$ such that
  $$G=\tilde{G}D,$$ where $D$ is a parametric positive diagonal matrix, and
  $\tilde{G}$ is fixed.
\end{lemma}
\begin{proof}
  Consider the following rescale for the matrix $R$:
  \begin{equation}
    R_{jm}=\tilde{R}_{jm}\;\rho_m,
  \end{equation}
  where $\tilde{R}_{jm}>0$ is a fixed positive value and
  $\rho_m\in(0,+\infty)$ is a positive parameter. Define $\tilde{R}$ as the
  reactivity matrix with entries defined by $\tilde{R}_{jm}$,
  $\tilde{G}:=S\tilde{R}$ and
  $D:=\operatorname{diag}(\rho_1,...\rho_{|M|})$. By construction, we get
  that $G=\tilde{G}D$.
\end{proof}

Since there is a parameter choice $\tilde{R}$ such that
$\tilde{G}=S\tilde{R}$ is $D$-Hopf, Lemma~\ref{lem:mainGDmain} allows us
express a family of symbolic Jacobian $G$ in the form $G=\tilde{G}D$,
where $D$ is a positive parametric matrix. As $\tilde{G}$ is $D$-Hopf,
there is an invertible principal submatrix $\tilde{G}[\kappa]$ of $G$ and two $k\times k$ positive diagonal matrices $D_1=\operatorname{diag}(\mathbf{d}_1)$ and $D_2=\operatorname{diag}(\mathbf{d}_2)$ so that
\begin{equation}
  \operatorname{inertia}\tilde{G}[\kappa]D_1 \neq \operatorname{inertia}\tilde{G}[\kappa]D_2.
\end{equation}
Considering any analytic path $\gamma(\mu)$
  in parameter space connecting $\mathbf{d}_1$ to $\mathbf{d}_2$ settles us
  in the generality of Lemma~\ref{lem:dynreductionmain}. This concludes the proof Thm.~\ref{thm:mainmain}.

\begin{remark}[Mass-Action kinetics I]\label{rm:lemmamamain}
Lemma \ref{lem:mainGDmain} applies not only to parameter-rich kinetics but an analogous result holds also for standard mass action kinetics, see \cite{Vassena2025} for a detailed analysis of the mass action case. For all kinetics mentioned in this paper (standard and generalized mass action, Michaelis--Menten, and Hill) the matrix $D$ in Lemma \ref{lem:mainGDmain} can be interpreted as the inverse of the steady-state values,
$$\rho_m=\dfrac{1}{\bar{x}_m}.$$
We exemplify this for monomolecular Michaelis--Menten kinetics:
\begin{equation}\label{eq:mmmonomain}
r(x)=a\frac{x}{1+b x}
\end{equation}
as any other case follows in total analogy. The derivative of $r(x)$ in \eqref{eq:mmmonomain}, evaluated at a steady-state $\bar{x}$ is
\begin{equation}
r'(\bar{x})=a\frac{1}{(1+b \bar{x})^2}=a\frac{1}{(1+b \bar{x})^2}\dfrac{\bar{x}}{\bar{x}}=\dfrac{r(\bar{x})}{1+b\bar{x}} \dfrac{1}{\bar{x}}.
\end{equation}
In particular, the derivative can be parametrized in terms of the steady-state flux $r(\bar{x})$, a free positive parameter $\tilde{b}=b\bar{x}\in(0,+\infty)$ and the value $\bar{x}$. Note that $\tilde{b}=0$ recovers the mass action case. Since all partial derivatives $r'_{jm}$ with respect to a variable $x_m$ appear in the $m^{th}$ column of the Jacobian matrix, the connection between $D$ and $\operatorname{diag}(1/\bar{x}_1,...,1/\bar{x}_{|M|})$ follows.
\end{remark}

\begin{remark}[Mass-Action kinetics II] In view of Remark~\ref{rm:lemmamamain}, Theorem~\ref{thm:mainmain} also applies to mass-action kinetics. However, the nondegeneracy assumption requires a dedicated discussion in this setting, since Lemma~\ref{lem:nondegcsmain} holds only for parameter-rich systems. In particular, for mass-action systems, the notion of degeneracy may depend on the steady-state constraints.
\end{remark}

\subsection{Proof of Cor.~\ref{cor:dhopfoscillationsmain}}\label{sec:proofscor}

The main Theorem~\ref{thm:mainmain} expresses a condition for
oscillations in terms of the Jacobian. We are primarily
interested, however, in conditions that directly refer to the
stoichiometric matrix $\mathbb{S}$. In the case of parameter-rich
kinetics this is indeed possible, as Corollary \ref{cor:dhopfoscillationsmain} states. Here we present its proof.

\begin{proof}[Proof of Corollary \ref{cor:dhopfoscillationsmain}]
Let $\pmb{\kappa}=(\kappa, E_\kappa, J)$ the associated $k$-CS. Consider
  the following rescale in the reactivity matrix $R$:
  \begin{equation}
    R_{jm}(\varepsilon)=
    \begin{cases}
      1 \quad \text{if $\mathsf{m}\in \kappa, j\in E_\kappa$ and $j=J(\mathsf{m})$;}\\
      \varepsilon \quad \text{otherwise}.
    \end{cases}
\end{equation}
At $\varepsilon=0$, the Jacobian $G(\varepsilon)=\mathbb{S}R(\varepsilon)$
is $D$-Hopf, since it can be written as a lower block-triangular matrix
with zero blocks on the diagonal, except for one block that is
$\mathbb{S}[\pmb{\kappa}]$. By continuity, the results holds for
$\varepsilon>0$, small enough. In fact, the principal submatrix of
$G[\kappa]$ identified by the species in $\kappa$ is still $D$-Hopf for
$\varepsilon>0$, small enough. Thus,
the full Jacobian $G$ is $D$-Hopf.
\end{proof}

In particular, the presence of even a single oscillatory core, no matter
which size, is sufficient for the whole system to exhibit nonstationary
periodic solutions. This in particular proves both Recipe \ref{recipe:1main} and Recipe \ref{recipe:2main}. However, it is important to note that this is not a
necessary condition, as Thm.~\ref{thm:mainmain} assumptions are more general
and do not require an oscillatory core, as we will further argue below.

%To summarize, we state here below three recipes (practical corollaries),
%Recipe \ref{recipe:0main}, Recipe \ref{recipe:1main}, and Recipe \ref{recipe:2main},
%which guarantee the insurgence of nonstationary periodic solutions in
%nondegenerate reaction networks that contains such structure.
%\subsubsection{Recipes \ref{recipe:1main} and \ref{recipe:2main}: Oscillatory cores}
%\addtocontents{toc}{\protect\setcounter{tocdepth}{-1}}

%The following Recipe \ref{recipe:1main} and Recipe \ref{recipe:2main}
%are based on Oscillatory Cores of Class I and II, respectively, as defined
%in Definitions \ref{def:Ocores1} and \ref{def:Ocores2}, respectively. Thus,
%Recipe \ref{recipe:1main} is based on unstable-positive feedback and Recipe
%\ref{recipe:2main} on unstable-negative feedback, only. The proofs of both
%recipes are a direct consequence of Corollary \ref{cor:dhopfoscillations}.

%\begin{customthm}{I}\label{recipe:1}
% Assume that the network possesses an Oscillatory Core of Class I. Then the system admits nonstationary periodic orbits.   
%\end{customthm}

%\begin{customthm}{II}\label{recipe:2}
%Assume that the network possesses an Oscillatory Core of Class II. Then the system admits nonstationary periodic orbits.    
%\end{customthm}

\subsection{Proof of Recipe~\ref{recipe:0main}}

We state here formally and prove Recipe \ref{recipe:0main}.

\begin{customthm}{0}\label{recipe:0}
  Let $\bar{R}$ be a partially specified reactivity matrix, where some
  entries $\bar{R}_{jm} \in [0, +\infty)$ are fixed, and the remaining
    entries are left as free parameters. Let $\bar{G} = \mathbb{S} \bar{R}$
    denote the corresponding symbolic Jacobian, and suppose there exists a
    principal submatrix $\bar{G}[\kappa]$ of $\bar{G}$ that remains
    invertible for all admissible choices of the free parameters in
    $\bar{R}$. If there exist two such parameter choices, say $\bar{R}_1$
    and $\bar{R}_2$, such that $\bar{G}[\kappa]$ is Hurwitz-stable for
    $\bar{R}_1$ and Hurwitz-unstable for $\bar{R}_2$, then the system
    admits nonstationary periodic solutions.
\end{customthm}
\begin{proof}
  Recipe \ref{recipe:0main} is in essence a corollary of
  Lemma~\ref{lem:dynreductionmain}. It assumes the invertibility of $G[\kappa]$
  for all parameter choice in $\bar{R}$, once few parameters have been
  preliminarily fixed. It assumes further the existence of two choices of
  parameters $\bar{R}_1$ and $\bar{R}_2$ so that a principal submatrix
  $G[\kappa]$ of the Jacobian is Hurwitz-stable at $\bar{R}_1$ and
  Hurwitz-unstable at $\bar{R}_2$. We just consider any analytic path
  $\gamma(\mu)$ connecting these $\bar{R}_1$ and $\bar{R}_2$ in parameter
  space and we apply Lemma \ref{lem:dynreductionmain}. Finally, by considering
  a sufficiently small positive choice for the parameters in
  $\bar{R}_{jm}$, which had been possible previously fixed to zero, the
  same conclusions hold, by continuity of the eigenvalues with respect to
  the entries.
\end{proof}

\begin{remark}
  We emphasize that, according to Recipe~\ref{recipe:0main}, it is admissible
  to fix certain parameters $R_{jm}$ to zero - even though this may seem to
  contradict the assumption of monotone kinetic models. This contradiction
  is only superficial: as shown in the proof, the conclusions
  still apply to systems governed by parameter-rich monotone kinetic
  models by a standard perturbation arguments. In practice, however, setting certain $R_{jm}$ to zero can
  simplify the search for parameter configurations that satisfy the
  hypotheses of Recipe~\ref{recipe:0main}, thereby facilitating the detection
  of oscillatory behavior.
\end{remark}

The instability in Recipe \ref{recipe:0main} may be generated from any
(D-)unstable cores and, in particular, both from an unstable-positive
feedback, and a unstable-negative feedback. It is remarkable that
Recipe \ref{recipe:0main} does not necessarily require the presence of an
oscillatory core. 
To exemplify, consider the case where a consistent and nondegenerate network $\mathbf{\Gamma}=(M,E)$ possesses only one $|M|$-CS $\bar{\pmb{\kappa}} = (M, \bar{E}_M, \bar{J})$, and the associated CS-matrix $\mathbb{S}[\bar{\pmb{\kappa}}]$ is $D$-stable. In this case, the determinant of the symbolic Jacobian $G$ is a single monomial:
\begin{equation}
\operatorname{det}G = \operatorname{det}\mathbb{S}[\pmb{\kappa}] \prod_{m \in M} R_{\bar{J}(m)m},
\end{equation}
and thus invertible for any choice of parameters. Furthermore, since $D$-stability implies Hurwitz-stability, such a situation ensures that the network admits a locally stable steady-state, while precluding the possibility that any principal minor of $\mathbb{S}[\pmb{\kappa}]$ forms an oscillatory core. However, the mere hypothetical presence of any other $D$-unstable core in the network -independently of any further conditions, and in particular regardless of the presence or absence of oscillatory cores - guarantees that the assumptions of Recipe~\ref{recipe:0main} are satisfied. In fact, by focusing on the whole symbolic Jacobian $G=\mathbb{S}R$ we get that
\begin{enumerate}
    \item The network admits stability, i.e. there exists choice $R_1$ with $G_1=\mathbb{S}R_1$ Hurwitz-stable;
    \item The network admits instability, i.e. there exists choice $R_2$ with $G_2=\mathbb{S}R_2$ Hurwitz-unstable;
    \item The symbolic Jacobian $G$ is invertible for all parameter choices $R$.
\end{enumerate}

Hence, the existence of periodic orbits follows from Recipe \ref{recipe:0main}, independently of the existence of oscillatory cores. Further examples and elaborations are discussed in Section~\ref{sec:recipe0} of the SM.

\bibliography{preprint.bib}

\begin{thebibliography}{10}

\bibitem{Morgan:1916}
John~Stanley Morgan.
\newblock The periodic evolution of carbon monoxide.
\newblock {\em J. Chem. Soc., Trans.}, 109:274--283, 1916.

\bibitem{Bray:1921}
William~C. Bray.
\newblock A periodic reaction in homogeneous solution and its relation to
  catalysis.
\newblock {\em J. Am. Chem. Soc.}, 43:1262--1267, 1921.

\bibitem{Noyes:1972}
R.~M. Noyes, R.~J Field, and R.~Kor{\"o}s.
\newblock Oscillations in chemical systems. {I.} {D}etailed mechanism in a
  system showing temporal oscillations.
\newblock {\em J. Am. Chem. Soc.}, 94:1394--1395, 1972.

\bibitem{Zhabotinsky:64}
A.~M. Zhabotinsky.
\newblock Periodical oxidation of malonic acid in solution (a study of the
  {Belousov} reaction kinetics).
\newblock {\em Biofizika}, 9:306--311, 1964.
\newblock (russian).

\bibitem{Kovacs2007}
K.~Kovacs, R.~E. McIlwaine, S.K. Scott, and A.F. Taylor.
\newblock An organic-based {pH} oscillator.
\newblock {\em J. Phys. Chem. A.}, 111:549--551, 2007.

\bibitem{Semenov2015}
S.~N. Semenov, A.~Wong, R.~van~der Made, Sjoerd G.~J. Postma, Joost Groen,
  Hendrik W.~H. van Roekel, Tom F.~A. de~Greef, and Wilhelm T.~S. Huck.
\newblock Rational design of functional and tunable oscillating enzymatic
  networks.
\newblock {\em Nature Chem}, 7:160--165, 2015.

\bibitem{Semenov2016}
S.~N. Semenov, L.~J. Kraft, Mengxia Ainla, A.~Zhao, Mostafa Baghbanzadeh,
  Victoria~E Campbell, Kyungtae Kang, Jerome M~Fox Fox, and George~M
  Whitesides.
\newblock Autocatalytic, bistable, oscillatory networks of biologically
  relevant organic reactions.
\newblock {\em Nature}, 537(7622):656--660, 2016.

\bibitem{Semenov2021}
Sergey~N. Semenov.
\newblock {\em De novo Design of Chemical Reaction Networks and Oscillators and
  Their Relation to Emergent Properties}, chapter~4, pages 91--130.
\newblock John Wiley \& Sons, Ltd, 2021.

\bibitem{Fei19}
Martin Feinberg.
\newblock {\em Foundations of Chemical Reaction Network Theory}.
\newblock Springer, 2019.

\bibitem{blokhuis20}
Alex Blokhuis, David Lacoste, and Philippe Nghe.
\newblock Universal motifs and the diversity of autocatalytic systems.
\newblock {\em Proceedings of the National Academy of Sciences},
  117(41):25230--25236, 2020.

\bibitem{VasStad23}
Nicola Vassena and Peter~F Stadler.
\newblock Unstable cores are the source of instability in chemical reaction
  networks.
\newblock {\em Proceedings of the Royal Society A}, 480(2285):20230694, 2024.

\bibitem{Gosh2024}
Souvik Ghosh, Mathieu~G. Baltussen, Nikita~M. Ivanov, Rianne Haije, Miglė
  Jakštaitė, Tao Zhou, and Wilhelm T.~S. Huck.
\newblock Exploring emergent properties in enzymatic reaction networks: Design
  and control of dynamic functional systems.
\newblock {\em Chemical Reviews}, 124(5):2553--2582, 2024.

\bibitem{Ivan25}
Nikita~M. Ivanov.
\newblock {\em Emergent Phenomena in Artificial Enzymatic Reaction Networks}.
\newblock Radboud Repository, 2025.

\bibitem{blokhuis2025datadimensionchemistry}
Alex Blokhuis, Martijn van Kuppeveld, Daan van~de Weem, and Robert Pollice.
\newblock On data and dimension in chemistry -- irreversibility, concealment
  and emergent conservation laws.
\newblock {\em ArXiv preprint: 2306.09553}, 2025.

\bibitem{MM13}
L.~Michaelis and M.~L. Menten.
\newblock Die {K}inetik der {I}nvertinwirkung.
\newblock {\em Biochem. Z.}, 49:333--369, 1913.

\bibitem{Hill10}
Archibald~Vivian Hill.
\newblock The possible effects of the aggregation of the molecules of
  haemoglobin on its dissociation curves.
\newblock {\em The Journal of Physiology}, 40:4--7, 1910.

\bibitem{Muller:12}
Stefan M{\"u}ller and Georg Regensburger.
\newblock Generalized mass action systems: Complex balancing equilibria and
  sign vectors of the stoichiometric and kinetic-order subspaces.
\newblock {\em SIAM Journal on Applied Mathematics}, 72(6):1926--1947, 2012.

\bibitem{Ang07}
David Angeli, Patrick De~Leenheer, and Eduardo~D Sontag.
\newblock A {P}etri net approach to the study of persistence in chemical
  reaction networks.
\newblock {\em Mathematical biosciences}, 210(2):598--618, 2007.

\bibitem{Hsubook}
Sze-Bi Hsu and Kuo-Chang Chen.
\newblock {\em Ordinary differential equations with applications}, volume~23.
\newblock World scientific, 2022.

\bibitem{GuHo84}
John Guckenheimer and Philip Holmes.
\newblock {\em Nonlinear oscillations, dynamical systems and bifurcations of
  vector fields}.
\newblock Springer, 1984.

\bibitem{Fiedler85PhD}
Bernold Fiedler.
\newblock An index for global {H}opf bifurcation in parabolic systems.
\newblock {\em Journal f{\"u}r die reine und angewandte Mathematik}, 358:1--36,
  1985.

\bibitem{ostrowski62}
Alexander Ostrowski and Hans Schneider.
\newblock Some theorems on the inertia of general matrices.
\newblock {\em Journal of Mathematical analysis and applications}, 4(1):72--84,
  1962.

\bibitem{Tyson2002}
John~J. Tyson.
\newblock {\em Biochemical Oscillations}, chapter~9, pages 230--260.
\newblock Springer, 2002.

\bibitem{Vassena2025}
Nicola Vassena.
\newblock Mass action systems: Two criteria for hopf bifurcation without
  hurwitz.
\newblock {\em SIAM Journal on Applied Mathematics}, 85(3):1046--1066, 2025.

\bibitem{Nguyen18}
Basile Nguyen, Udo Seifert, and Andre~C Barato.
\newblock Phase transition in thermodynamically consistent biochemical
  oscillators.
\newblock {\em The Journal of chemical physics}, 149(4), 2018.

\bibitem{BaPa16}
Murad Banaji and Casian Pantea.
\newblock Some results on injectivity and multistationarity in chemical
  reaction networks.
\newblock {\em SIAM Journal on Applied Dynamical Systems}, 15(2):807--869,
  2016.

\bibitem{Harmsel2023}
Matthijs Ter~Harmsel, Oliver~R Maguire, Sofiya~A Runikhina, Albert~SY Wong,
  Wilhelm~TS Huck, and Syuzanna~R Harutyunyan.
\newblock A catalytically active oscillator made from small organic molecules.
\newblock {\em Nature}, 623:87--93, 2023.

\bibitem{Prigogine1968}
Ilya Prigogine and Ren{\'e} Lefever.
\newblock Symmetry breaking instabilities in dissipative systems. ii.
\newblock {\em The Journal of Chemical Physics}, 48:1695–--1700, 1968.

\bibitem{Field1974}
Richard~J. Field and Richard~M. Noyes.
\newblock Oscillations in chemical systems. iv. limit cycle behavior in a model
  of a real chemical reaction.
\newblock {\em The Journal of Chemical Physics}, 60(5):1877--1884, 03 1974.

\bibitem{MPSmith90}
John Mallet-Paret and Hal Smith.
\newblock The poincar{\'e}-bendixson theorem for monotone cyclic feedback
  systems.
\newblock {\em Journal of Dynamics and Differential Equations}, 2(4):367--421,
  1990.

\bibitem{HofbauerHypercycle91}
J~Hofbauer, J~Mallet-Paret, and HL~Smith.
\newblock Stable periodic solutions for the hypercycle system.
\newblock {\em Journal of Dynamics and Differential Equations}, 3:423--436,
  1991.

\bibitem{F19}
Bernold Fiedler.
\newblock Global {H}opf bifurcation in networks with fast feedback cycles.
\newblock {\em Discrete and Continuous Dynamical Systems - Series S},
  0(1937-1632\_2019\_0\_144), 2020.

\bibitem{ba23splitting}
Murad Banaji.
\newblock Splitting reactions preserves nondegenerate behaviors in chemical
  reaction networks.
\newblock {\em SIAM Journal on Applied Mathematics}, 83(2):748--769, 2023.

\bibitem{lin2025biomasstransferautocatalyticreaction}
Wei-Hsiang Lin.
\newblock Biomass transfer on autocatalytic reaction network: a delay
  differential equation formulation.
\newblock {\em ArXiv preprint: 2210.09470}, 2025.

\bibitem{blokhuis2020errthresh}
Alex Blokhuis, Philippe Nghe, Luca Peliti, and David Lacoste.
\newblock The generality of transient compartmentalization and its associated
  error thresholds.
\newblock {\em Journal of Theoretical Biology}, 487:110110, 2020.

\bibitem{Johnson-Buck2017}
Alexander Johnson-Buck and William~M. Shih.
\newblock Single-molecule clocks controlled by serial chemical reactions.
\newblock {\em Nano Letters}, 17(12):7940--7944, 2017.

\bibitem{Takashima2011}
Yoshiki Takashima, Toshiyuki Ohtsuka, Aitor González, Hitoshi Miyachi, and
  Ryoichiro Kageyama.
\newblock Intronic delay is essential for oscillatory expression in the
  segmentation clock.
\newblock {\em Proceedings of the National Academy of Sciences},
  108(8):3300--3305, 2011.

\bibitem{Horio1986}
Tetsuya Horio and Hirokazu. Hotani.
\newblock Visualization of the dynamic instability of individual microtubules
  by dark-field microscopy.
\newblock {\em Nature}, 321:605--607, 1986.

\bibitem{Hess2017}
H.~Hess and Jennifer~L. Ross.
\newblock Non-equilibrium assembly of microtubules: from molecules to
  autonomous chemical robots.
\newblock {\em Chem. Soc. Rev.}, 46:5570--5587, 2017.

\bibitem{Leira-Iglesias2018}
J.~Leira-Iglesias, A.~Tassoni, and T.~et~al. Adachi.
\newblock Oscillations, travelling fronts and patterns in a supramolecular
  system.
\newblock {\em Nature Nanotech}, 13:1021–1027, 2018.

\bibitem{Torigoe2024}
Shogo Torigoe, Kazutoshi Nagao, Ryou Kubota, and Itaru Hamachi.
\newblock Emergence of dynamic instability by hybridizing synthetic
  self-assembled dipeptide fibers with surfactant micelles.
\newblock {\em Journal of the American Chemical Society}, 146(9):5799--5805,
  2024.

\bibitem{Howlett2022}
M.G. Howlett, A.H.J. Engwerda, and R.J.H. et~al. Scanes.
\newblock An autonomously oscillating supramolecular self-replicator.
\newblock {\em Nature Chemistry}, 14:805--810, 2022.

\bibitem{Donlon2014}
Lynn Donlon and Katarina Novakovic.
\newblock Oscillatory carbonylation using alkyne-functionalised poly(ethylene
  glycol).
\newblock {\em Chem. Commun.}, 50:15506--15508, 2014.

\bibitem{Lodge2025}
Rhys~W. Lodge, William~J. Cull, Andreas Weilhard, Stephen~P. Argent, Jesum
  Alves~Fernandes, and Andrei~N. Khlobystov.
\newblock A nanoscale chemical oscillator: reversible formation of palladium
  nanoparticles in ionic liquid.
\newblock {\em Nanoscale}, 17:10105--10116, 2025.

\bibitem{Hermans2015}
T.M. Hermans, P.S. Stewart, and B.A. Grzybowski.
\newblock p{H} oscillator stretched in space but frozen in time.
\newblock {\em J. Phys. Chem. Lett.}, 6:760--766, 2015.

\bibitem{Roelofs1983}
Mark~G. Roelofs, E.~Wasserman, James~H. Jensen, and Allan~E. Nader.
\newblock Mechanism of an oscillating organic reaction: oxidation of
  benzaldehyde with oxygen catalyzed by cobalt/bromine.
\newblock {\em Journal of the American Chemical Society}, 105(20):6329--6330,
  1983.

\bibitem{Jensen1983}
James~H. Jensen.
\newblock New type of oscillating reaction: air oxidation of benzaldehyde.
\newblock {\em Journal of the American Chemical Society}, 105(9):2639--2641,
  1983.

\bibitem{hyver1978}
Claude Hyver.
\newblock Schema reactionnel, catalyse et oscillations chimiques.
\newblock {\em Comptes rendus de l'Academie des Sciences}, 286, 1978.

\bibitem{APNHJ23}
Alejandro L{\'o}pez-Nieto, Phillipo Lappicy, Nicola Vassena, Hannes Stuke, and
  Jia-Yuan Dai.
\newblock Hybrid bifurcations: periodicity from eliminating a line of
  equilibria.
\newblock {\em Mathematische Annalen}, 391(4):6373--6399, 2025.

\bibitem{ChowMPYorkeGH78}
Shui-Nee Chow, John Mallet-Paret, and James~A Yorke.
\newblock Global {H}opf bifurcation from a multiple eigenvalue.
\newblock {\em Nonlinear Analysis: Theory, Methods \& Applications},
  2(6):753--763, 1978.

\bibitem{MuFlSt:22}
Stefan M{\"u}ller, Christoph Flamm, and Peter~F. Stadler.
\newblock What makes a reaction network ``chemical''?
\newblock {\em J. Cheminformatics}, 14:63, 2022.
\newblock arXiv 2201.01646.

\bibitem{MA64}
P~Waage and CM~Guldberg.
\newblock Studier over affiniteten.
\newblock {\em Forhandlinger i Videnskabs-selskabet i Christiania}, 1:35--45,
  1864.

\bibitem{Gio15}
Giorgio Giorgi and Cesare Zuccotti.
\newblock An overview on {D}-stable matrices.
\newblock Technical Report 097, Department of Economics and Management DEM
  Working Paper Series, 2015.

\bibitem{Ku19}
Olga~Y. Kushel.
\newblock Unifying matrix stability concepts with a view to applications.
\newblock {\em SIAM Review}, 61(4):643--729, 2019.

\bibitem{FisherFuller58}
Michael~E Fisher and A~T Fuller.
\newblock On the stabilization of matrices and the convergence of linear
  iterative processes.
\newblock {\em Mathematical Proceedings of the Cambridge Philosophical
  Society}, 54(4):417--425, 1958.

\bibitem{Fisher72simple}
Franklin~M Fisher.
\newblock A simple proof of the {F}isher--{F}uller theorem.
\newblock {\em Mathematical Proceedings of the Cambridge Philosophical
  Society}, 71(3):523--525, 1972.

\bibitem{VasHunt}
Nicola Vassena.
\newblock Symbolic hunt of instabilities and bifurcations in reaction networks.
\newblock {\em Discrete and Continuous Dynamical Systems - Series B},
  30(6):2183--2208, 2023.

\bibitem{Bullo20}
Francesco Bullo.
\newblock {\em Lectures on network systems}, volume~1.
\newblock Kindle Direct Publishing, 2020.

\bibitem{V23}
Nicola Vassena.
\newblock Structural conditions for saddle-node bifurcations in chemical
  reaction networks.
\newblock {\em SIAM Journal on Applied Dynamical Systems}, 22(3):1639--1672,
  2023.

\bibitem{VAA24ME}
Nicola Vassena, Florin Avram, and Rim Adenane.
\newblock Finding bifurcations in mathematical epidemiology via reaction
  network methods.
\newblock {\em Chaos}, 34:123163, 2024.

\bibitem{Murray:2007}
James~D Murray.
\newblock {\em Mathematical biology}.
\newblock Springer Berlin, Heidelberg, 1993.

\bibitem{ClarkeSNA}
Bruce~L Clarke.
\newblock Stoichiometric network analysis.
\newblock {\em Cell biophysics}, 12:237--253, 1988.

\bibitem{polettini_irreversible_2014}
Matteo Polettini and Massimiliano Esposito.
\newblock Irreversible thermodynamics of open chemical networks. {I}.
  {Emergent} cycles and broken conservation laws.
\newblock {\em The Journal of Chemical Physics}, 141(2):024117, July 2014.

\bibitem{Ba18}
Murad Banaji.
\newblock Inheritance of oscillation in chemical reaction networks.
\newblock {\em Applied Mathematics and Computation}, 325:191--209, 2018.

\end{thebibliography}
\bibliographystyle{unsrt}

\newpage

\begin{center}
   {\Huge \textbf{Supplementary Material}}
\end{center}

\newcommand{\beginsupplement}{%
  \setcounter{table}{0}
  \renewcommand{\thetable}{S\arabic{table}}%
  \setcounter{figure}{0}
  \renewcommand{\thefigure}{S\arabic{figure}}%
  \setcounter{section}{0}
  \renewcommand{\thesection}{S\arabic{section}}%
}

\beginsupplement

The SM is organized as follows.

In Part \hyperlink{part1}{I}, we give a self-contained,
rigorous description of the mathematical results underlying the idea of
oscillatory cores and the three recipes: \ref{recipe:0main}, \ref{recipe:1main}, \ref{recipe:2main}. To this end we first recall basic concepts of the theory
of reaction networks (Sec.~\ref{sec:reactionnetworks}). In
Sec.~\ref{sec:linearalgebra}, we present a linear algebra overview of the
concept of $D$-stability and $P^-$ matrices. This section stands
independent and does not rely on any network concept. We recall the network
tools of Child Selections in Sec.~\ref{sec:CS}. We then proceed in
Sec.~\ref{sec:bridge} to bridge Sec.~\ref{sec:linearalgebra} with reaction
networks: we interpret the linear algebra considerations in terms of the
stoichiometry of the network. 

In Part \hyperlink{part2}{II}, we discuss in detail several examples. A first list of oscillatory cores of both classes is presented in Sec.~\ref{sec:oscillatorycores}. Examples involving Recipe \ref{recipe:0main} are presented in Sec.~\ref{sec:recipe0}. A network with innocent net-stoichiometry and oscillations induced by hidden catalysts is presented in Sec.~\ref{ex:hidcat}.

Part \hyperlink{part3}{III} exemplifies how to leverage our results in the context of explicit numerical simulations.\\

\vspace{2cm}

\textbf{\Large \hypertarget{part1}{PART I}: Mathematical Results}
\addtocontents{toc}{\protect\setcounter{tocdepth}{-1}}

\section{Reaction Networks}\label{sec:reactionnetworks}
\addtocontents{toc}{\protect\setcounter{tocdepth}{-1}}

A reaction network $\pmb{\Gamma}=(M,E)$ consists of a set $M$ of \emph{species} and a
set $E$ of \emph{reactions} that transform subsets of species into each other. In \emph{chemical} reaction networks in particular, it is important to keep track
of the copy number of reactant and product species. Thus a reaction $j\in
E$ is formalized as an ordered association between non-negative linear combinations of the species
\begin{equation}\label{reaction}
  \sum_{\mathsf{m}\in M}
  s^j_{m} \mathsf{m}
  \quad\underset{j}{\rightarrow}\quad
  \sum_{\mathsf{m}\in M} \tilde{s}^j_{m} \mathsf{m}.
\end{equation}
Species appearing on the left-hand side of the reactions are called \emph{reactants}, and species appearing on the right-hand side of the reaction are called \emph{products}. The nonzero \emph{stoichiometric coefficients} of the reactants are indicated as $s^j_m>0$ and of the products $\tilde{s}^j_m>0$, respectively. We say that the network is \emph{nonambiguous} if explicit
catalysts are absent, i.e., no species appear as both reactant and product in the
same reaction. In this case, the reaction network is completely determined by the
stochiometric matrix with entries
\begin{equation}
  \mathbb{S}_{mj} \coloneqq \tilde{s}^j_{m} - s^j_{m}\,,
\end{equation}
which describes the net production or consumption. We do not assume nonambiguity, unless otherwise explicitly stated. Formally, we assume all reactions to be irreversible, i.e. we give a fixed order to a reactions, according to the sign of the arrow in \eqref{reaction}. Of course, reversible processes 
\begin{equation}
    \ce{X} \quad \rightleftharpoons \quad \ce{Y}
\end{equation}
can be naturally treated simply by considering two irreversible, opposite reactions:
\begin{equation}
    \ce{X} \quad \rightarrow \quad \ce{Y}  \quad \text{and} \quad \ce{Y} \quad \rightarrow \quad \ce{X}.
\end{equation}

We do not, a priori, assume that the network represents a closed system of well-formed chemical reactions, as we consider general reaction networks that may be open, involve strongly irreversible reactions, or be subject to external driving. Thus, we do not assume that the stoichiometric coefficients
are integers, and also make no assumptions on the existence of reaction invariants such as conservation of mass, nor do we require that energy is conserved. The consequences of such constraints on the stoichiometry of the network are discussed in
\cite{MuFlSt:22}. 

The concentration of species $\mathsf{m}$ is denoted by $x_m\ge0$. We
assume that the dynamics of the network is completely determined by the
reactions $j\in E$ and thus can be written in the form 
\begin{equation}\label{eq:dynamics}
  \dot{\mathbf{x}} = f(\mathbf{x}) = \mathbb{S} \pmb{r}(\mathbf{x})
\end{equation}
where $r_j(\mathbf{x})$ is the rate of reaction $j$. For each initial condition $\mathbf{x}_0\in \mathbb{R}^{|M|}_{>0}$ the sets
\begin{equation}
\mathbf{x}_0 + \operatorname{Im} \mathbb{S},
\end{equation}
invariant to the flow of \eqref{eq:dynamics}, are called \emph{stoichiometric compatibility classes}. Fixed points $\bar{\mathbf{x}}$ of \eqref{eq:dynamics},
\begin{equation}
  0 = \mathbb{S} \pmb{r}(\bar{\mathbf{x}})
  \label{eq:fixedpoint}
\end{equation}
are called \emph{steady-states}. Throughout, we consider only networks whose
stoichiometry allows for the existence of positive steady-states,
i.e., we assume that the stoichiometric matrix $\mathbb{S}$ possesses a
positive right kernel vector $\pmb{v}>0$, i.e.,
\begin{equation}
    \mathbb{S} \pmb{v}=0.
\end{equation}
Networks satisfying such a property are called \emph{consistent} \cite{Ang07}. In turn, \emph{left} kernel vectors $w$ of the stoichiometric matrix $\mathbb{S}$ identify \emph{conservation laws}, since
\begin{equation}
\dfrac{d (w \mathbf{x}(t))}{dt}=w\dot{\mathbf{x}}=w\mathbb{S}\mathbf{r}(\mathbf{x})=0.
\end{equation}

We further assume that a reaction $j$
produces its products with a rate $r_j(\mathbf{x})$ that depends only on the
concentrations of the reactants of reaction $j$. The rate $r_j(\mathbf{x})$ vanishes
if some reactants of $j$ are missing, and increases with the reactant
concentrations. We refer to such features as \emph{monotone kinetic model}, in the sense formalized by the following definition:
\begin{definition}[Monotone kinetic model]\label{def:kineticmodel}
  Let $\pmb{\Gamma}=(M,E)$ be a chemical reaction network. An analytic
  function $\pmb{r}:\mathbb{R}^M_{\ge0}\to\mathbb{R}^E$ is a \emph{monotone
  kinetic model} for $\pmb{\Gamma}$ if
  \begin{enumerate}
  \item $r_j(\mathbf{x})\ge0$ for all $\mathbf{x}$, 
  \item $r_j(\mathbf{x})>0$ implies $x_m>0$ for all $m$ with $s^j_{m}>0$,
  \item $s^j_{m}=0$ implies $\partial r_j/\partial x_m\equiv 0$,
  \item if $\mathbf{x}>0$ and $s^j_{m}>0$ then $\partial r_j/\partial x_m>0$.
  \end{enumerate}
\end{definition}
The analyticity assumption is introduced solely for technical reasons to apply Fiedler's Theorem on Global Hopf Bifurcation (see Theorem~\ref{thm:fiedlermain} in the main text). However, this assumption is not restrictive in the context of kinetics, as most kinetic models are naturally expressed as rational functions defined on the positive orthant. The reaction rates in general depend on a collection $\mathbf{p}>0$ of positive
parameters, $\pmb{r}(\mathbf{x})=\pmb{r}(\mathbf{x};\mathbf{p})$. For simplicity of notation, at a given steady-state vector $\bar{\mathbf{x}}\in\mathbb{R}^M_{\ge 0}$ we write
\begin{equation} 
  r'_{jm}=r'_{jm}(\bar{\mathbf{x}},\mathbf{p}) \coloneqq  \frac{\partial r_j(\mathbf{x},\mathbf{p})}{\partial x_m}\bigg\vert_{\mathbf{x}=\bar{\mathbf{x}}}\,,
\end{equation}
for the partial derivatives of the reaction rates. In this paper, we adopt the general framework of \emph{parameter-rich kinetic models} introduced in \cite{VasStad23}, for which we first provide a preliminary definition.
\begin{definition}[Reactivity matrix]
\label{def:R}
Consider a network $\mathbf{\Gamma}$. We call any $|E|\times |M|$ matrix $R$ with nonnegative
  entries satisfying $R_{jm}>0$ if $s_{m}^{j}>0$ and $R_{jm}=0$ if
  $s_{m}^{j}=0$ a \emph{reactivity matrix}
for the network $\mathbf{\Gamma}$. 
\end{definition}

For nonambiguous networks,  $\mathbb{S}_{mj}=s_m^j<0$ for all reactions $j$ with reactant $m$. In this case, the reactivity matrix $R$ has a nonzero entry $R_{jm}>0$ if and only if the stoichiometric matrix has a negative entry $\mathbb{S}_{mj}<0$. 
\begin{definition}[Parameter-rich kinetics]\label{def:parameterrich}
A monotone kinetic rate model $\pmb{r}(\mathbf{x};\mathbf{p})$ is \emph{parameter-rich} if, for every
  positive steady-state $\bar{\mathbf{x}}>0$ and every reactivity matrix $R$, there are parameters $\bar{\mathbf{p}}=\mathbf{p}(\bar{\mathbf{x}},R)$ such that
\begin{equation}R_{jm}=r'_{jm}(\bar{\mathbf{x}},\bar{\mathbf{p}}).
\end{equation}
  \label{def:p-rich}
\end{definition}

As shown in \cite{VasStad23}, Michaelis--Menten kinetics \cite{MM13}, 
\begin{equation}\label{MMeq}
  r_j(\mathbf{x},a_j,b_j) : =a_j\prod_{m\in M} \Bigg( \frac{x_m}{(1+b^j_m
    x_m)}\Bigg)^{s^j_m},\quad a_j,b^j_m>0;
\end{equation}
Hill kinetics \cite{Hill10},
\begin{equation}\label{eq:hill}
  r_j(\mathbf{x},a_j,b_j,c_j) : =a_j\prod_{m\in M} \Bigg( \frac{x_m^{c^j_m}}{(1+b^j_m
    x_m^{c^j_m})}\Bigg)^{s^j_m}\quad a_j,b^j_m, c^j_m>0;
\end{equation}
and generalized mass action kinetics \cite{Muller:12}, 
\begin{equation}\label{eq:gma}
  r_j(\mathbf{x},a_j,b_j) := a_j\prod_{m\in M} x_m^{b_{m}^{j}}, \quad \quad \quad
  \text{$b^j_m\ge 0$ with $b^j_m \neq 0$ if and only if $s^j_m\neq 0$;}
\end{equation}
all are parameter-rich. In contrast, standard mass-action kinetics \cite{MA64} is not parameter-rich:
\begin{equation}\label{MAeq}
  r_j(\mathbf{x},a_j) := a_j\prod_{m\in M} x_m^{s_{m}^{j}}. 
\end{equation}
See the dedicated paragraph below for an extended discussion on the relation between elementary mass action and parameter-rich kinetics.

We recall that a square matrix $A$ is \emph{Hurwitz-stable} if all of its eigenvalues have negative real part, while it is \emph{Hurwitz-unstable} if at least one of its eigenvalues has positive real part. We recall that the dynamical stability of a steady-state $\bar{\mathbf{x}}$ of a dynamical system
$$\dot{\mathbf{x}}=f(\mathbf{x})$$
relates to the Hurwitz-stability of the Jacobian matrix \cite{Hsubook} $$G(\bar{\mathbf{x}}):=\dfrac{\partial f(\mathbf{x})}{\partial \mathbf{x}}\bigg |_{\mathbf{x}=\bar{\mathbf{x}}}.$$
In particular, if the Jacobian matrix evaluated at $\bar{\mathbf{x}}$ is Hurwitz-stable (resp. Hurwitz-unstable), then the steady-state $\bar{\mathbf{x}}$ is asymptotically stable (resp. asymptotically unstable). However, in the case of reaction network systems, \eqref{eq:dynamics}, the rank of the Jacobian 
is bounded above by
\begin{equation}
\operatorname{rank}G(\bar{\mathbf{x}})\le\operatorname{rank}\mathbb{S}.
\end{equation}
In the presence of linear conservation laws, thus, the Jacobian matrix always possesses a trivial eigenvalue zero with algebraic multiplicity at least
\begin{equation}
n=\operatorname{dim}\operatorname{kernel}(\mathbb{S}^T).
\end{equation} 
Via a standard reduction procedure, the conservation laws can always be used to reduce the system to a fixed stoichiometric compatibility class by eliminating 
$n$ variables \cite{Fei19}. Throughout this paper, Hurwitz stability for Jacobians of reaction networks with conservation laws is understood within each stoichiometric compatibility class, disregarding the 
$n$ trivial zero eigenvalues that arise from the conservation laws. We are now ready for the next definition.
\begin{definition}
  A network $\pmb{\Gamma}=(M,E)$ with a parametrized kinetic model $r(\mathbf{x};\mathbf{p})$
  \emph{admits instability (resp. stability)} if there exists a choice $\bar{\mathbf{p}}$ of
  parameters such that there is a positive steady-state $\bar{\mathbf{x}}$ of
  $\dot{\mathbf{x}}=\mathbb{S} r(\mathbf{x};\bar{\mathbf{p}})$ with a Hurwitz-unstable (resp. Hurwitz-stable) Jacobian $G(\bar{\mathbf{x}})$.
\label{def:admitinstabil}
\end{definition}

%Obviously, 
One can think of the reactivity matrix $R$ defined in Def.~\ref{def:R} as a symbolic matrix where each entry $R_{jm}\in(0,+\infty)$ is a positive parameter. For a reaction network endowed with parameter-rich kinetics, the Jacobian of \eqref{eq:dynamics} at any steady-state $\bar{\mathbf{x}}$ can be written as
\begin{equation}\label{eq:SR}
  G(\bar{\mathbf{x}}) = G:=\mathbb{S} R,
\end{equation}
and thus can be thought as well as a symbolic matrix. We refer to the matrix $G=\mathbb{S}R$ as the \emph{symbolic Jacobian} of the network. Note that the right-hand side of \eqref{eq:SR} does not explicitly depend on the steady-state value $\bar{\mathbf{x}}>0$. In this sense, the parameter-rich framework enables the study of local stability of steady states to be reformulated as a purely algebraic analysis of the symbolic Jacobian $\mathbb{S}R$. Since Def.~\ref{def:p-rich} applies to any positive steady-state $\bar{\mathbf{x}}>0$, its value can be chosen arbitrarily. The reaction constants can then always be determined through inversion formulas, whose explicit form naturally depends on the specific kinetic model. Crucially, Def.~\ref{def:p-rich} itself ensures the feasibility of this procedure.

\paragraph{Parameter-rich results and mass-action kinetics.}
Our results, in particular, will not decisively establish the existence of periodic oscillations for mass-action systems. The advantage of adopting the parameter-rich framework is that conditions for oscillations can be interpreted more directly from the stoichiometry, in contrast to the more algebraically involved criteria required for mass-action systems. For a dedicated analysis of this related, yet distinct, case, see \cite{Vassena2025}.

It is worth highlighting a twofold subtlety in the relationship between mass-action systems and parameter-rich kinetics. On the one hand, many non-elementary parameter-rich kinetics - such as Michaelis--Menten - can be derived as singular limits of enlarged mass-action mechanisms. For instance, the reaction
\begin{equation}
    \ce{X_1} \quad \rightarrow \quad \ce{X_2},
\end{equation}
with Michaelis--Menten kinetics, can be viewed as the singular limit of the enzymatic mass-action scheme
\begin{equation}
\ce{X_1} + \ce{E}\quad  \rightleftharpoons \quad \ce{I} \quad \rightarrow \quad \ce{X_2} + \ce{E},\end{equation}
in the regime where the enzyme $\ce{E}$ and intermediate $\ce{I}$ are present at concentrations negligible compared to those of the substrate $\ce{X_1}$ and product $\ce{X_2}$. In this sense, periodic solutions observed in parameter-rich systems suggest the potential for analogous behavior in the corresponding expanded mass-action systems. We will exploit this idea to provide oscillatory mass action systems in Sec.~\ref{sec:recipe0basicconstr}.

On the other hand, from a more mathematical perspective, the polynomial mass-action nonlinearity~\eqref{MAeq} arises as a limiting case of Michaelis--Menten kinetics, specifically when $b_j = 0$. In this sense, mass-action kinetics can be regarded as a special instance - rather than a foundation - of parameter-rich models. This has two key implications: First, any result that holds for \emph{all} parameter choices in a parameter-rich setting will necessarily extend to mass-action kinetics; in particular, exclusion results remain valid as, e.g., ruling out a stability or bifurcation scenario. Second, even when the results are not conclusive for mass-action models, the parameter-rich framework may offer insightful guidance identifying parameter regimes likely to support oscillations. We will use this strategy to interpret oscillatory mass-action networks reported in the literature.

\paragraph{Nondegenerate networks.}
Throughout this paper, we consider only \emph{nondegenerate networks}
satisfying a nondegeneracy condition for the symbolic Jacobian
$G=\mathbb{S}R$.
\begin{definition}[Nondegenerate networks]\label{def:nondegnet}
Let $\mathbf{\Gamma}$ be a reaction network with $|M|\times|E|$
stoichiometric matrix $\mathbb{S}$, and let $n\ge0$ indicate the number of
linearly independent conservation laws, i.e.,
$n=\operatorname{dim}\operatorname{ker}\mathbb{S}^T$. Then
$\mathbf{\Gamma}$ is nondegenerate if there exists a choice for the
reactivity matrix $R$ such that the symbolic Jacobian $G=\mathbb{S}R$
possesses a nonzero $(|M|-n)$-principal minor.
\end{definition}
Nondegenerate networks are those for which the trivial eigenvalues zero of the Jacobian all descend from the conservation laws. In particular, in absence of conservation laws, a nondegenerate network admits an invertible Jacobian determinant $G=\mathbb{S}R$ for some choice of $R$.  We will provide in Sec.~\ref{sec:bridge} a simple stoichiometric characterization of nondegenerate networks, based on the language of Child Selections.

In the continuation of this paper we will consider nondegenerate networks endowed with parameter-rich kinetics that admit both stability and instability. That is, the network admits two choices of $R_1$ and $R_2$ such that $G_1=\mathbb{S}R_1$ is Hurwitz-stable and $G_2=\mathbb{S}R_2$ is Hurwitz-unstable. Clearly, any continuous path $\gamma$ in matrix space that connects $R_1$ to $R_2$ necessarily hits at least one \emph{bifurcation point} $R^*$ such that $\mathbb{S}R^*$ possesses an eigenvalue with zero real part. We will investigate conditions for which such stability-loss happens at purely imaginary eigenvalues, triggering thus nonstationary periodic oscillations via Hopf bifurcation.

In the next Sec.~\ref{sec:linearalgebra}, we discuss abstract conditions for purely imaginary eigenvalues in a class of matrices.

\section{Linear Algebra}\label{sec:linearalgebra}
\addtocontents{toc}{\protect\setcounter{tocdepth}{-1}}
This section revisits the standard linear algebra notions of $D$-stability and $P^-$ matrices, with an eye on our objective. For a more comprehensive discussion, see \cite{Gio15, Ku19}. We introduce some necessary notation: let $\kappa$ denote any selection of $k \leq n$ indices from $\{1,\dots,n\}$. $A[\kappa]$ represents the principal submatrix of $A$, obtained by retaining the columns and rows with common indices given by those in $\kappa$; its determinant is referred to as \emph{a $k$-principal minor} of $A$.  We now recall the concept of the \emph{inertia} of a matrix.

\begin{definition}[Inertia of a matrix]
The inertia of an $n \times n$ square matrix $A$ is a nonnegative triple 
    $$\operatorname{inertia}(A) \coloneqq (\sigma^-_A,\sigma^+_A,\sigma^0_A),$$
where $\sigma^-_A$, $\sigma^+_A$, and $\sigma^0_A$ represent the number of eigenvalues of $A$ with negative real part, positive real part, and zero real part, respectively. The eigenvalues are counted according to their algebraic multiplicities so that $\sigma^+_A + \sigma^-_A + \sigma^0_A = n$.
\end{definition}

We state a new definition of $D$-Hopf matrix as follows.

\begin{definition}[D-Hopf matrix]\label{def:dhopf}
A $n\times n$ matrix $A$ is called \emph{D-Hopf} if there exists an invertible $k$-principal submatrix $A[\kappa]$ of $A$, and two positive $k\times k$ diagonal matrices $D_1$ and $D_2$ such that 
\begin{equation}\label{eq:changeofinertia}
\operatorname{inertia}A[\kappa]D_1\neq\operatorname{inertia}A[\kappa]D_2.
\end{equation}
\end{definition}

In the literature there exist a few related concepts, which we gather together in the following definition.

\begin{definition}[Stability and related properties]\label{def:stabilities}
    A $n\times n$ matrix $A$ is said to be
    \begin{enumerate}
    \item \emph{Hurwitz-stable} if its inertia is $$(n,0,0),$$ i.e., if $A$ possesses only eigenvalues with negative real part.
    \item \emph{Hurwitz-unstable} if its inertia is $$(\cdot,\sigma^+_A,\cdot)\quad\text{with $\sigma^+_A>0,$}$$ i.e. if $A$ possesses at least one eigenvalue with positive real part.
    \item \emph{Hyperbolic} if its inertia is 
    $$(\cdot,\cdot,0),$$
    i.e. if $A$ does not possess eigenvalues with zero real part.
    \item \emph{$D$-stable} if $AD$ is Hurwitz-stable for all positive diagonal matrices $D$.
    \item \emph{$D$-unstable} if there exists a positive diagonal matrix $\bar{D}$ such that $A\bar{D}$ is Hurwitz-unstable.
    \item \emph{$D$-hyperbolic} if $AD$ is hyperbolic for all positive diagonal matrices $D$.
    \item \emph{$D$-nonhyperbolic} if there exists a positive diagonal matrix $\bar{D}$ such that $A\bar{D}$ is not hyperbolic.
    \end{enumerate}
\end{definition}

%Note that, in more applied literature, e.g. in \emph{control theory}, (in)stability is refereed as Hurwitz (in)stability. 

\begin{remark}
A $D$-Hopf matrix $A$ is always $D$-unstable, since either $AD_1$ or $AD_2$ must be Hurwitz-unstable.
\end{remark}

According to Def.~\ref{def:stabilities}, a $D$-Hopf matrix possesses an invertible principal submatrix that is not $D$-hyperbolic. Setting $D = \operatorname{Id}$ shows however that $D$-stability and $D$-hyperbolicity imply Hurwitz-stability and hyperbolicity, respectively. In particular, if a matrix $A$ is not hyperbolic in the usual sense, it is also $D$-nonhyperbolic, though this does not necessarily make it $D$-Hopf. Thus, the key point is that Def.~\ref{def:dhopf} depends on principal submatrices. It is straightforward, then, to underline that a matrix $A$ that possesses any principal submatrix that is a $D$-Hopf matrix is itself $D$-Hopf. Because of the importance of this observation, we state it in the following proposition.
\begin{prop}\label{prop:reduction}
   Let $A$ be a $n\times n$ matrix. If any principal submatrix $A[\kappa]$ of $A$ is $D$-Hopf, then $A$ is $D$-Hopf. 
\end{prop}
\begin{proof}
The proof follows from Def.~\ref{def:dhopf}. If $A[\kappa]$ is $D$-Hopf, then there exists a principal submatrix of $A[\kappa]$, $A[\kappa']$ that is invertible and satisfies \eqref{eq:changeofinertia}. Since $A[\kappa']$ is also a principal submatrix of $A$, the statement follows.
\end{proof}

Proposition \ref{prop:reduction} guarantees that the $D$-Hopf property is preserved by considering $D$-Hopf matrices as principal submatrices of matrices of any size. This \emph{reduction property} will be crucial in identifying oscillatory motifs in reaction networks of any size.

We will prove insurgence of periodic oscillations via Hopf bifurcation, which in turn requires the existence of a steady-state with purely imaginary eigenvalues, the next proposition supports our choice of name.

\begin{lemma}\label{lem:Dhopfpurely}
If a matrix $A$ is D-Hopf, there exists a positive diagonal matrix $D^*$ such that $AD^*$ has purely imaginary eigenvalues.
\end{lemma}
\proof Let $A[\kappa]$ be the invertible $k$-principal submatrix of $A$ such that \eqref{eq:changeofinertia} holds, i.e., there exists two positive $k\times k$ diagonal matrices
$D_1=\operatorname{diag}(\mathbf{d}_1)$,
$D_2=\operatorname{diag}(\mathbf{d}_2)$, with $\mathbf{d}_1,\mathbf{d}_2\in
\mathbb{R}^k_{>0}$, such that 
$$\operatorname{inertia} A[\kappa]D_1\neq \operatorname{inertia} A[\kappa]D_2$$
Consequently, we can consider any path
$\gamma(\mu):[0,1]\mapsto\mathbb{R}^k_{>0}$ with $\gamma(0)=\mathbf{d}_1$ and
$\gamma(1)=\mathbf{d}_2$, and the associated matrix eigenvalue path
$\Lambda(\mu)$ in $\mathbb{R}^k$ defined
as $$\Lambda(\mu) \coloneqq \operatorname{eigenvalues}(A\operatorname{diag}(\gamma(\mu))).$$
Since between $\mu=0$ and $\mu=1$, the real part of at least one of the
eigenvalues changes sign, the intermediate value theorem identifies at least
one $\mu^*$ such that $A[\kappa]D^*$ is not hyperbolic for the positive diagonal matrix $D^*\coloneqq \operatorname{diag}(\gamma(\mu^*))$. 
Since multiplication with a full-rank matrix preserves rank, and hence 
$$\operatorname{rank}A[\kappa]\operatorname{diag}(\gamma(\mu))=\operatorname{rank}A[\kappa]=k,$$
$A[\kappa]D^*$ is invertible, i.e., $A[\kappa]$ has no real zero-eigenvalues but possesses purely imaginary eigenvalues.

To lift the statement to $A$, assume without loss of generality that $\kappa=\{1,\dots,k\}$. We argue by continuity. We embed the above path $\gamma(\mu)$ to $\mathbb{R}^n$ by defining a path $\tilde{\gamma}(\mu,\varepsilon):[0,1]\times \mathbb{R}_{\ge0}\mapsto \mathbb{R}^n_>0$ as follows:
$$\tilde{\gamma}(\mu,\varepsilon)_m=\begin{cases}
    \gamma(\mu)_m   \quad\;\text{if $m=1,...,k$;}\\
        \varepsilon, \quad\quad \quad \text{otherwise.}
\end{cases}$$
Clearly, for $\varepsilon=0$ we find - as above - at least one stability change at purely-imaginary eigenvalues. Note that the imaginary part of such eigenvalues is necessarily at a positive distance from zero. Thus, by continuity of the eigenvalues of a matrix with respect to the entries, the statement holds also for a $\bar{\varepsilon}>0$ small enough, but positive. Therefore, we have proved the existence of a $\bar{\mu}^*$ such that for $\bar{D}^*:=\operatorname{diag}(\tilde{\gamma}(\bar{\mu}^*,\bar{\varepsilon}))$, $A\bar{D}^*$ has purely imaginary eigenvalues. \endproof

In general, proving that a matrix $A$ is $D$-Hopf is computationally expensive. A standard approach would involve identifying a hyperbolic principal submatrix of $A$ that is not $D$-hyperbolic using a Routh--Hurwitz computation, which becomes unfeasible for matrices beyond small sizes.
This approach also has the drawback of obscuring any structural intuition about \emph{why} a matrix is $D$-Hopf. To establish manageable and insightful sufficient conditions for a matrix to be $D$-Hopf, we leverage the concept of $P^-$ matrices.

\begin{definition}[$P^-$ and $P^-_0$ matrices]\label{def:Pmatrix}
A matrix $A$ of size $n \times n$ is called a $P^-$ matrix if all of its $k$-principal minors have the sign $(-1)^k$. It is called a $P^-_0$ matrix if all of its \emph{nonzero} $k$-principal minors have the sign $(-1)^k$.
\end{definition}

The set of $P^-_0$ matrices is the closure of the open set of $P^-$ matrices. We now interpret two standard results to connect $P^-_0$ and $D$-Hopf matrices. 

\begin{prop}\label{prop:Pf}
If any principal submatrix $A[\kappa]$ of $A$ is Hurwitz-stable but not a $P^-_0$ matrix, then $A$ is $D$-Hopf.
\end{prop}
\proof
If $A[\kappa]$ is not a $P^-_0$ matrix, then there is a principal submatrix $A[\kappa']$ of $A[\kappa]$ such that $$\operatorname{sign}\operatorname{det}A[\kappa']=(-1)^{k'-1}.$$ 
Consequently, $A[\kappa']$ has an odd number of real positive eigenvalues, making it Hurwitz-unstable. Again, without loss of generality consider $\kappa'=\{1,...,k'\}$.
In similar fashion as in Lemma~\ref{lem:Dhopfpurely} we can then choose $D(\varepsilon)$ positive diagonal matrix defined as follows
\begin{equation}
    D_{mm}(\varepsilon)=
    \begin{cases}
        1 \quad\;\text{if $m=1,...,k'$}\\
        \varepsilon \quad  \text{otherwise.}
    \end{cases}
\end{equation}
Considering $\varepsilon=0$ shows that $A[\kappa]D(0)$ is Hurwitz-unstable as $A[\kappa']$. By continuity of the eigenvalues there exists $\varepsilon>0$ small enough such that $A[\kappa]D(\varepsilon)$ is Hurwitz-unstable. As $A[\kappa]$ is Hurwitz-stable, we have then shown
$$\operatorname{inertia}A[\kappa]\operatorname{Id}\neq  \operatorname{inertia} A[\kappa]D(\varepsilon),$$
so $A$ is $D$-Hopf.
\endproof

The second result follows from a renowned theorem by Michael E. Fisher and A. T. Fuller (1958) \cite{FisherFuller58}, also discussed by Franklin M. Fisher in \cite{Fisher72simple}. An additional definition is needed to fully state the theorem

\begin{definition}[Fisher\&Fuller $P^-_{FF}$ matrices]
A $n\times n$ $P^-_0$ matrix $A$ is termed a $P^-_{FF}$ matrix, or a \emph{Fisher\&Fuller matrix}, if there exists a sequence of nested invertible principal matrices 
$$(A[\kappa_1], A[\kappa_2], \dots, A[\kappa_{n-1}], A[\kappa_{n}]),$$
with each order $|\kappa_i| = i$ for $i = 1, \dots, n$, such that $A[\kappa_{i-1}]$ is a principal submatrix of $A[\kappa_i]$.
\end{definition}

Note that a $P^-_{FF}$ matrix $A$ is itself invertible. It is straightforward to note the following inclusions:
$$\text{$P^-$ matrices} \quad \subset \quad \text{$P^-_{FF}$ matrices} \quad \subset \quad \text{$P^-_0$ matrices}.$$

A note of caution: Fisher \cite{Fisher72simple} refers to $P^-_{FF}$ matrices as \emph{Hicksian}. The term also appears in the literature \cite{Gio15} for $P^-$ matrices.

\begin{theorem}[Theorem 1' in \cite{Fisher72simple}]\label{thm:FF}
If $A$ is a $P^-_{FF}$ matrix, there exists a positive diagonal matrix $D$ such that $AD$ has all eigenvalues that are real, negative, and simple.
\end{theorem}

Now, we can state a corollary for $D$-Hopf matrices.

\begin{cor}\label{cor:2}
If $A$ has a principal minor $A[\kappa]$ that is a Hurwitz-unstable $P^-_{FF}$ matrix. Then $A$ is $D$-Hopf.
\end{cor}
\proof
Since $A[\kappa]$ is Hurwitz-unstable, $A[\kappa]\operatorname{Id}$ has an eigenvalue with positive real part. Since $A[\kappa]$ is a $P^-_{FF}$ matrix, Thm.~\ref{thm:FF} implies the existence of a positive diagonal matrix  $D_2$ such that $A[\kappa]D_2$ has all eigenvalues real, negative, and simple. Thus $A$ is $D$-Hopf.
\endproof

A variation on Thm.~\ref{thm:FF} and requires only the Hurwitz-stability of one of the principal submatrices of $A[\kappa]$.

\begin{prop}\label{prop:variation}
Let $A$ be a $n\times n$ matrix such that
\begin{equation}
\operatorname{sign}\operatorname{det}(A)=(-1)^n,
\end{equation}
and further assume that $A$ has a Hurwitz-stable $(n-1)\times(n-1)$ principal submatrix. Then there exists a positive diagonal matrix $D$ such that $AD$ is Hurwitz-stable.
\end{prop}

\proof
Let $\tilde{A}$ be the Hurwitx-stable $(n-1)\times(n-1)$ principal submatrix of $A$. Without loss of generality, assume that $\tilde{A}$ is a leading principal submatrix of $A$. Consider the positive diagonal  matrix $D(\varepsilon)$ defined as follows:
\begin{equation}
\begin{cases}
    D_{ii}=1\quad \text{for $i=1,...,(n-1)$};\\
    D_{nn}=\varepsilon.
\end{cases}
\end{equation}
Clearly, for the limit case $\varepsilon=0,$ the product $AD(0)$ has $(n-1)$ eigenvalues with negative real part and one real zero eigenvalue. By continuity of the eigenvalues with respect to the entries, for $\varepsilon$ small enough $AD(\varepsilon)$ has still $(n-1)$ eigenvalues with negative real part. Due to the fact that
\begin{equation}
\operatorname{sign}\operatorname{det}AD(\varepsilon)=\operatorname{sign}\operatorname{det}A=(-1)^n
\end{equation}
for any $\varepsilon$, the last eigenvalue cannot be real-positive, so it must either real-negative or a complex eigenvalue conjugate to one eigenvalue with negative-real part. In both cases, $AD(\varepsilon)$ is Hurwitz-stable.
\endproof

\begin{cor}\label{cor:3}
If $A$ has a $k\times k$ principal submatrix $A[\kappa]$ that is an invertible Hurwitz-unstable $P^-_0$ matrix with a $(k-1)\times(k-1)$ Hurwitz-stable principal submatrix. Then $A$ is $D$-Hopf.
\end{cor}
\proof
The proof follows analogously as in the proof of Cor.~\ref{cor:2} from Prop.~\ref{prop:variation}.
\endproof

\section{Back to networks: Child-Selections}\label{sec:CS}
\addtocontents{toc}{\protect\setcounter{tocdepth}{-1}}

We first recall concepts from \cite{VasHunt} and \cite{VasStad23}. 

\begin{definition}[Child-Selections]
  Let $\mathbb{S}$ be the stoichiometric matrix of the reaction network
  $\pmb{\Gamma}=(M,E)$. 
  A \emph{$k$-Child-Selection triple}, or $k$-CS for short, is a triple
  $\pmb{\kappa}=(\kappa,E_{\kappa},J)$ such that $|\kappa|=|E_{\kappa}|=k$,
  $\kappa\subseteq M$, $E_{\kappa}\subseteq E$, and $J:\kappa\to
  E_{\kappa}$ is a bijection satisfying $s_m^{J(m)}>0$ for all
  $m\in\kappa$. We call $J$ a \emph{Child-Selection bijection}.
\end{definition}

%Note that since $J$ is a bijection between two ordered sets, we can
%naturally consider the \emph{signature} (or parity) $\sgn{J}$ of the map
%$J$, where $J$ is seen as a permutation of a set of cardinality $k$.
A Child-Selection $\pmb{\kappa}$ in particular defines a 
square matrix  $\mathbb{S}[\pmb{\kappa}]$ with entries
\begin{equation}
  \mathbb{S}[\pmb{\kappa}]_{ml} \coloneqq \mathbb{S}[\kappa,E_{\kappa}]_{m,J(l)} =
  \tilde s_{m}^{J(l)} - s_{m}^{J(l)},
\end{equation}
where the permutation of the columns of $\mathbb{S}[\pmb{\kappa}]$ from
$\mathbb{S}[\kappa,E_\kappa]$ is described by the Child-Selection bijection $J$. We refer to such a matrix as a \emph{Child-Selection matrix} (CS-matrix).

Nondegenerate networks, defined in Def.~\ref{def:nondegnet}, admits a direct characterization in terms of CS-matrices.
\begin{lemma}\label{lem:nondegcs}
A network $\mathbf{\Gamma}$ with stoichiometric matrix $\mathbb{S}$ such that $\operatorname{dim}\operatorname{ker}\mathbb{S}^T=n\ge0$ is nondegenerate if and only if there exists an invertible $(|M|-n)\times (|M|-n)$ CS-matrix.
\end{lemma}
\proof
The proof follows from the fact that each coefficient $c_k$ of the characteristic polynomial $g(\lambda)$ of the symbolic Jacobian $G=\mathbb{S}R$, 
\begin{equation}
   g(\lambda):=\operatorname{det}(\mathbb{S}R-\lambda \operatorname{Id})=\sum_{k=0}^{|M|}(-1)^kc_k\lambda^{|M|-k}
\end{equation}
can be expanded along Child Selections:
\begin{equation}\label{eq:expcs}
c_k=\sum_{\pmb{\kappa}}\operatorname{det}\mathbb{S}[\pmb{\kappa}]\prod_{\mathsf{m}\in\kappa}R_{J(m)m},
\end{equation}
where the sum runs on all $k$-CS triples $\pmb{\kappa}=(\kappa, E_\kappa, J)$. See Section 4 in \cite{VasStad23} for a detailed analysis of expansion \eqref{eq:expcs}. As $c_{|M|-n}$ can be seen as a multilinear homogenous polynomial of order $|M|-n$ in the positive variables $R_{jm}>0$, it follows that $c_{|M|-n}\equiv0$ if and only if all coefficients of the monomials, i.e. all $(|M|-n)\times (|M|-n)$ CS-matrices are singular, which proves the lemma.
\endproof

Furthermore, Child-Selections and CS-matrices provide simple sufficient conditions for a network to admit stability or instability, as the next two Propositions state, whose proofs can be found in the references.
\begin{prop}[Corollary 5.11 of \cite{VasHunt}]\label{prop:stability}
Consider a network $\mathbf{\Gamma}=(M,E)$ with stoichiometric matrix $\mathbb{S}$ such that $\operatorname{dim}\operatorname{ker}\mathbb{S}^T=n\ge0$. Assume there exists a $(|M|-n)$-CS $\pmb{\kappa}$ such that its associated $(|M|-n)\times(|M|-n)$ CS-matrix is Hurwitz-stable. Then the network admits stability.
\end{prop}

\begin{prop}[Corollary 5.1 of \cite{VasStad23}, Corollary 5.12 of \cite{VasHunt}]\label{pro:csunstable}
Consider a network $\mathbf{\Gamma}=(M,E)$. Assume there exists a $k$-CS $\pmb{\kappa}$ such that its associated $k\times k$ CS-matrix is Hurwitz-unstable. Then the network admits instability.
\end{prop}

In nonambiguous networks, explicit catalysts that participate in a reaction both as reactant and product are excluded. In this case, therefore, a CS-matrix $\mathbb{S}[\pmb{\kappa}]$ is always a square matrix with strictly negative
diagonal entries. Note that the converse is true in general, independently from the network being nonambiguous:
\begin{lemma}
  For any $k\times k$ submatrix $\tilde{\mathbb{S}}$ of $\mathbb{S}$ with columns reordered
  so that $\tilde{\mathbb{S}}$ has negative diagonal, there exists a $k$-CS
  $\pmb{\kappa}$ such that $\tilde{\mathbb{S}}=\mathbb{S}[\pmb{\kappa}]$.
\end{lemma}
\begin{proof}
Fix $\kappa=\{\mathsf{m}_1,...,\mathsf{m}_k\}$ to be the species whose index appears as rows of $\tilde{\mathbb{S}}$ and, respectively, $E_\kappa=\{j_1,...,j_k\}$ the reactions whose index appears as columns of $\tilde{\mathbb{S}}$. The order of both sets given here follows the order in $\tilde{\mathbb{S}}$. Consider the bijection $J$  given by  
$$J(\mathsf{m}_i)=j_i, \text{ for $i=1,...,k$}.$$
$J$ is a Child-Selection bijection since $s^{j_i}_{\mathsf{m}_i}\ge s^{j_i}_{\mathsf{m}_i}-\tilde{s}^{j_i}_{\mathsf{m}_i}= -\tilde{\mathbb{S}}_{ii}>0$ by assumption. Thus $\pmb{\kappa}=(\kappa,E_\kappa, J)$ is a Child-Selection triple so that $\tilde{\mathbb{S}}=\mathbb{\mathbb{S}}[\pmb{\kappa}].$
\end{proof}

Let $\pmb{\kappa}=(\kappa,E_{\kappa},J)$ be a $k$-CS. Forming
Child-Selections can be concatenated in the sense that any restriction
$\pmb{\kappa'}$ of $\pmb{\kappa}$, to subsets $\kappa'\subset \kappa$ and
$E_{\kappa'}=J(\kappa')\subset E_{\kappa}$ is itself also a $k'$-CS
$\pmb{\kappa'}=(\kappa',E_{\kappa'},J)$. At a level of matrices, the
CS-matrix associated to $\pmb{\kappa'}$ appears as a principal submatrix of
the CS-matrix associated to $\pmb{\kappa}$. Therefore, it is natural to
defines minimal matrices with respect to some property in the following
sense.

\begin{definition}[Cores]\label{def:cores}
Let $\mathbb{P}$ be a matrix property. A $\mathbb{P}$-core is a
CS-matrix $\mathbb{S}[\pmb{\kappa}]$ with property $\mathbb{P}$ that
does not have a proper principal submatrix with property $\mathbb{P}$.
\end{definition}

A similar idea was used in \cite{blokhuis20} to define \emph{autocatalytic cores}, based on a stoichiometric definition of autocatalysis. In the same contribution, for nonambiguous networks, an exhaustive classication of five types of autocatalytic cores were given, which we can summarize essentially in the following matrices:
\bea
\label{autocat_cores}
    \begin{pmatrix}
        -1 & 2\\
        1 & -1
    \end{pmatrix} \ \ \vcenter{\hbox{\includegraphics[scale=0.3]{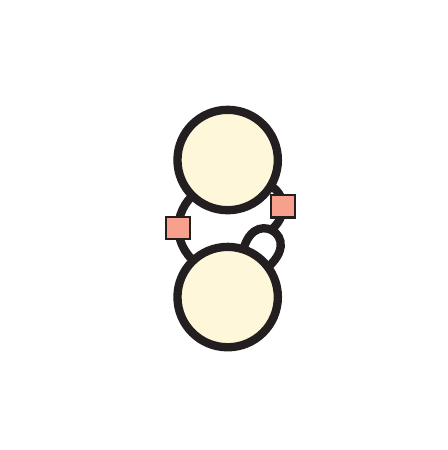}}} \quad 
    \begin{pmatrix}
        -1 & 0 & 1\\
        1 & -1 & 0 \\
        1 & 1 & -1
    \end{pmatrix} \ \ \vcenter{\hbox{\includegraphics[scale=0.3]{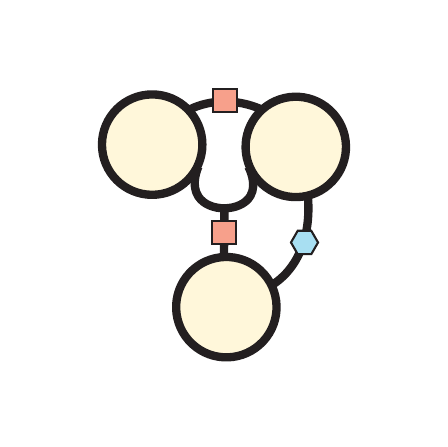}}}
    \quad
    \begin{pmatrix}
        -1 & 1 & 1\\
        1 & -1 & 0 \\
        1 & 0 & -1
    \end{pmatrix} \ \ \vcenter{\hbox{\includegraphics[scale=0.3]{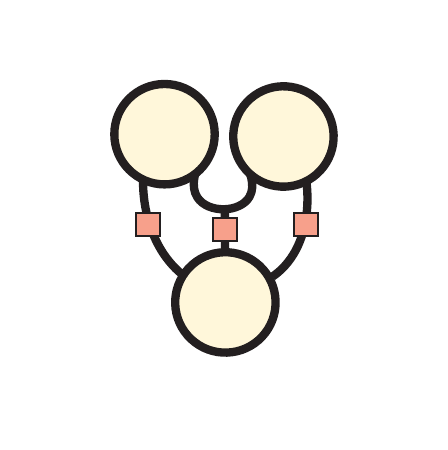}}} \nonumber \\
    \quad 
    \begin{pmatrix}
        -1 & 1 & 1\\
        1 & -1 & 0 \\
        1 & 1 & -1
    \end{pmatrix} \vcenter{\hbox{\includegraphics[scale=0.3]{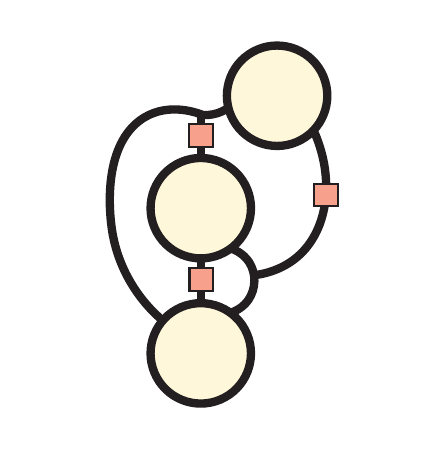}}}
    \quad
    \begin{pmatrix}
        -1 & 1 & 1\\
        1 & -1 & 1 \\
        1 & 1 & -1
    \end{pmatrix} \ \vcenter{\hbox{\includegraphics[scale=0.3]{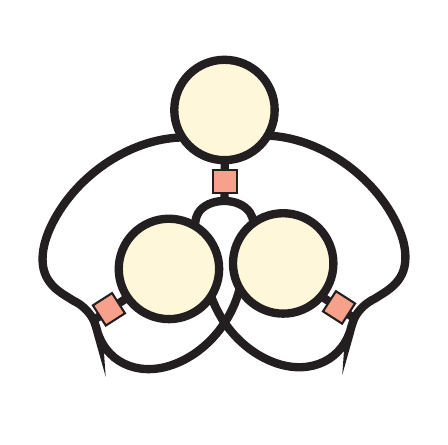}}}
\eea

In \cite{VasStad23}, the authors focused on Hurwitz-instability and in the Child-Selection language defined \emph{unstable cores} as follows.

\begin{definition}
  An unstable core is a Hurwitz-unstable Child Selection matrix that does not
  possess any Hurwitz-unstable principal submatrix.
\end{definition}

Due to Prop.~\ref{pro:csunstable}, the presence in the stoichiometry of any unstable core, no matter of which size, is sufficient for the network to admit instability. Unstable cores naturally come in two fashions, depending on the sign of the determinant.

\begin{definition}
  An unstable core $\mathbb{S}[\pmb{\kappa}]$ is called \textbf{unstable-positive
    feedback} if
$$\operatorname{det}\mathbb{S}[\pmb{\kappa}]=(-1)^{k-1}\,.$$
It is called \textbf{unstable-negative feedback} if 
$$\operatorname{det}\mathbb{S}[\pmb{\kappa}]=(-1)^{k}\,.$$
\end{definition}

Autocatalytic cores, as defined in \cite{blokhuis20}, turn out to be a special class of unstable
cores. Recall that a Metzler matrix is a square matrix with non-negative
off-diagonal entries \cite{Bullo20}. 
\begin{prop} (Thm.~7.11 of \cite{VasStad23})\label{prop:thm711}
  The autocatalytic cores of $\mathbb{S}$ \textit{sensu} \cite{blokhuis20} coincide
  with the unstable-positive feedbacks $\mathbb{S}[\pmb{\kappa}]$ that are
  Metzler matrices. 
\end{prop}

Autocatalytic cores do not account for all cases of unstable-positive
feedbacks. In particular, there exists nonautocatalytic unstable-positive
feedback which have same characteristic polynomial as autocatalytic cores,
see the concept of \emph{twin-pair of matrices} in Sec.7 of
\cite{VasStad23}.  Hurwitz-instability of an unstable-positive feedback, on
its own, therefore does not differentiate between autocatalytic and
non-autocatalytic mechanisms.

Finally, we can generalize the definition of unstable cores, by exploiting
the concept of $D$-instability rather than Hurwitz-instability.
\begin{definition}
A $D$-unstable core is a $D$-unstable CS matrix $\mathbb{S}[\pmb{\kappa}]$ such that none of its principal submatrices are $D$-unstable.
\end{definition}
$D$-unstable cores are sufficient for a network to admit instability.
In \cite{VasStad23}, it is conjectured that $D$-unstable cores are also
necessary:
  %characterize the capacity of a network of admitting instability:
\begin{conjecture}[Conjecture 5.5 in \cite{VasStad23}] A network admits instability if and only if it possesses $D$-unstable cores.  
\end{conjecture}

\section{Bridge: $D$-Hopf matrices and oscillatory cores}\label{sec:bridge}
\addtocontents{toc}{\protect\setcounter{tocdepth}{-1}}

We aim at investigating condition for a CS-matrix to be $D$-Hopf. We define
\emph{oscillatory cores} as follows.
\begin{definition}[Oscillatory Cores] \label{def:Oscillatorycore}
A $D$-Hopf CS-matrix $\mathbb{S}[\pmb{\kappa}]$ is called \emph{oscillatory core} if none of its principal submatrix is $D$-Hopf.
\end{definition}

\begin{prop}\label{prop:DHopfcores}
Any CS-matrix that is $D$-Hopf contains an invertible principal submatrix that is an oscillatory core.
\end{prop}
\proof
Def.~\ref{def:dhopf} of $D$-Hopf matrices requires the existence of an invertible principal submatrix $A[\kappa]$ for which \eqref{eq:changeofinertia} holds. $A[\kappa]$ is in particular a $D$-Hopf matrix. Thus, any $D$-Hopf $A$ matrix contains as a principal submatrix an invertible $D$-Hopf matrix. By minimality, we get the statement.
\endproof

We observe a direct consequence of Prop.~\ref{prop:DHopfcores}.
\begin{obs}\label{obs:invertible}
    An oscillatory core is an invertible matrix.
\end{obs}

Using the concept of $D$-unstable cores, we state the following first Lemma:
\begin{lemma}\label{prop:stableDunstable}
    A $D$-unstable core $\mathbb{S}[\pmb{\kappa}]$, which is Hurwitz-stable, is an oscillatory core.
\end{lemma}
\proof
$\mathbb{S}[\pmb{\kappa}]$ is Hurwitz-stable thus invertible. We moreover have that there exists a positive diagonal matrix $D$ such that $\mathbb{S}[\pmb{\kappa}]D$ is Hurwitz-unstable. Thus,
$$\operatorname{inertia}\mathbb{S}[\pmb{\kappa}]\operatorname{Id}\neq \operatorname{inertia}\mathbb{S}[\pmb{\kappa}]D,$$
and $\mathbb{S}[\pmb{\kappa}]$ is $D$-Hopf. Assume now indirectly that there exist a principal submatrix $\mathbb{S}[\pmb{\kappa}']$ of $\mathbb{S}[\pmb{\kappa}]$, which is $D$-Hopf. In particular $\mathbb{S}[\pmb{\kappa}']$ is $D$-unstable, contradicting the minimality of $\mathbb{S}[\pmb{\kappa}]$ being a $D$-unstable core.
\endproof

We next discuss two different constructions for a CS-matrix to be an
oscillatory core based respectively on unstable-positive feedbacks and
unstable-negative feedbacks.

\paragraph{Unstable-positive feedbacks.} It is straightforward to see that unstable-positive feedbacks have a single real positive eigenvalue (Lemma 6.3 in \cite{VasStad23}). Multistationarity - the coexistence of multiple steady states under identical conditions - naturally arises in parameter space due to zero-eigenvalue bifurcations, where the Jacobian loses invertibility and the implicit function theorem no longer applies. Based on this, unstable-positive feedbacks are potential sources of multistationarity, e.g. via a saddle-node bifurcation \cite{V23}. However, we will see that they can also give rise to periodic oscillations. The next lemma and corollary are central steps towards this conclusion. 

\begin{lemma}\label{lem:upfpmatrix}
Let $\mathbb{S}[\pmb{\kappa}]$ be a CS-matrix such that it contains as a principal submatrix an unstable-positive feedback. Then $\mathbb{S}[\pmb{\kappa}]$ is not a $P^-_0$ matrix.
\end{lemma}
\begin{proof}
  It just follows from the definition of unstable-positive feedback, which
  is a $k\times k$ unstable core with determinant of sign $(-1)^{k-1}$,
  namely contradicting the definition of $P^-_0$ matrix.
\end{proof}

Lemma~\ref{lem:upfpmatrix} implies straightforward the corollary:
\begin{cor}\label{cor:Ocores1}
  Let $\mathbb{S}[\pmb{\kappa}]$ be a stable CS-matrix such that it contains as a
  principal submatrix an unstable-positive feedback. Then
  $\mathbb{S}[\pmb{\kappa}]$ is $D$-Hopf.
\end{cor}
\proof
The statement follows from combining Prop.~\ref{prop:Pf} with Lemma \ref{lem:upfpmatrix}.
\endproof

Note that since $\mathbb{S}[\pmb{\kappa}]$ in Cor.~\ref{cor:Ocores1} is Hurwitz-stable, then the unstable-positive feedback must be a \emph{strict} principal submatrix. Using Prop.~\ref{prop:DHopfcores}, from Cor.~\ref{cor:Ocores1} we define \emph{Oscillatory Cores of Class I} as follows.
\begin{definition}[Oscillatory Cores of Class I]\label{def:Ocores1} An \textbf{Oscillatory Core of Class I} is a CS-matrix that is minimal with the property of being Hurwitz-stable and possessing an unstable-positive feedback as a principal submatrix. 
\end{definition}

\paragraph{Unstable negative feedbacks.}
Complementarily to the observation on unstable-positive feedbacks, Lemma 6.3 in \cite{VasStad23} states that unstable-negative feedbacks have no real positive eigenvalues, but only complex-conjugate pairs of eigenvalues with positive real part. This implies that the instability of unstable-negative feedback is not a potential source of zero-eigenvalue bifurcations and consequent multistationarity. In contrast, unstable-negative feedbacks naturally induce periodic oscillations. We argue further in this direction.
\begin{prop}\label{prop:unfpmatrix}
  If $\mathbb{S}[\pmb{\kappa}]$ is an unstable-negative feedback then
  $\mathbb{S}[\pmb{\kappa}]$ is an Hurwitz-unstable $P^-_0$ matrix.
\end{prop}
\begin{proof}
  The definition of unstable cores implies that $\mathbb{S}[\pmb{\kappa}]$ is
  Hurwitz-unstable and does not have an Hurwitz-unstable principal submatrix. In particular the sign of all its $k'$-principal minor is $(-1)^{k'}$. Further, the sign of $\operatorname{det}\mathbb{S}[\pmb{\kappa}]$ is $(-1)^k$ as well, by definition of negative feedback. Thus $\mathbb{S}[\pmb{\kappa}]$
  is a $P^-_0$ matrix.
\end{proof}

We get then two corollaries, in analogy from Cor.~\ref{cor:2} and Cor.~\ref{cor:3}.

\begin{cor}\label{cor:Ocores2}
  Let $\mathbb{S}[\pmb{\kappa}]$ be an unstable-negative feedback and a
  $P^-_{FF}$ matrix, then $\mathbb{S}[\pmb{\kappa}]$ is $D$-Hopf.
\end{cor}
\proof
It follows from Cor.~\ref{cor:2}.
\endproof
\begin{cor}\label{cor:Ocores2b}
    Let $
\mathbb{S}[\pmb{\kappa}]$ be a $k\times k$ unstable-negative feedback such that one of its $(k-1)\times (k-1)$ principal submatrices is Hurwitz-stable, then $\mathbb{S}[\pmb{\kappa}]$ is $D$-Hopf.
\end{cor}
\proof
It follows from Cor.~\ref{cor:3}.
\endproof

Using again Prop.~\ref{prop:DHopfcores}, from Cor.~\ref{cor:Ocores2} we define \emph{Oscillatory Cores of Class II} as follows.
\begin{definition}[Oscillatory Cores of Class II]\label{def:Ocores2} An \textbf{Oscillatory Core of Class II} is a $k\times k$ CS-matrix which is an unstable-negative feedback and either is a $P^-_{FF}$ matrix or it has a $(k-1)\times (k-1)$ Hurwitz-stable principal submatrix. 
\end{definition}
\begin{remark}
 The two ``either...or'' options in Def.~\ref{def:Ocores2} do not imply each other: a $P^-_{FF}$ may have principal submatrices with purely imaginary eigenvalues.
\end{remark}
\begin{remark}
Note that in Def.~\ref{def:Ocores2} the minimality is guaranteed by the fact that an Oscillatory Core of Class II is an unstable-negative feedback. 
\end{remark}

\newpage
\textbf{\Large \hypertarget{part2}{PART II}: Examples}

\section{Oscillatory Cores}\label{sec:oscillatorycores}
\addtocontents{toc}{\protect\setcounter{tocdepth}{-1}}

We now explore in detail simple patterns for oscillations.

\subsection{Recipe I, positive-feedback and autocatalysis}\label{sec:recipeI}
\addtocontents{toc}{\protect\setcounter{tocdepth}{-1}}
Most oscillators reported in the literature involve autocatalysis. Out of this reason, for Recipe \ref{recipe:1main}, based on Oscillatory Cores of class I, we enlarge the five $n\times n$ autocatalytic cores \eqref{autocat_cores} from \cite{blokhuis20} to get a $(n+1) \times (n+1)$ Hurwitz-stable matrix with negative diagonal, see Table \ref{tab:posmain} in the main text for an overview of these examples.

\textbf{This choice of examples is arbitrary and it is intended only to provide a first selection.} It is not exhaustive for several reasons: Firstly, one could start with unstable-positive feedbacks that are not autocatalytic, e.g. by relying on the concept of \emph{twin-pair}  of matrices, see Sec.~\ref{sec:bridge}. Secondly, the difference in size between the unstable-positive feedback and the oscillatory core may be more than 1. Thirdly, there are multiple ways 
to enlarge an unstable-positive feedback to a Hurwitz-stable CS-matrix. Our informal guiding principle in design has been to be as frugal as possible, i.e., we used only bimolecular reactions and we minimized the nonzero entries in the oscillatory core.
 
\paragraph{Oscillatory core (I,a).} The first oscillatory core of class I in our selection is:
\begin{equation}\label{eq:ocIa}%centering
\begin{pmatrix}
    -1 & 2 & 1 \\
     1 & -1 & 1\\
     0 & -1 & -1
    \end{pmatrix}
,\ \vcenter{\hbox{\includegraphics[scale=0.3]{core11redraw.pdf}}} 
\end{equation}
Such oscillatory core is built on the first type of the autocatalytic cores from \cite{blokhuis20},
\begin{equation}
\begin{pmatrix}
    -1 & 2\\
    1 & -1
\end{pmatrix}, \ \vcenter{\hbox{\includegraphics[scale=0.3]{ACtypeI.pdf}}}
\end{equation}
highlighted in magenta in the figure, and enlarged to the $3\times3$ Hurwitz-stable matrix \eqref{eq:ocIa}, with eigenvalues approximately $(-2.32, -0.34 \pm 0.56i)$. Recipe \ref{recipe:1main} implies that any consistent network containing the following three reactions admits periodic oscillations for a certain choice of parameters if endowed with parameter-rich kinetics.
\begin{equation}
\begin{cases}
\ce{X} + ... &\rightarrow \quad\ce{Y} + ...\\
\ce{Y} + \ce{Z} + ... &\rightarrow \quad 2 \ce{X} + ...\\
\ce{Z}+...&\rightarrow \quad \ce{X}+\ce{Y}
\end{cases}
\end{equation}
Above, `...' indicates that the reactions may involve other species different from $\ce{X}, \ce{Y}, \ce{Z}$. Periodic oscillations appear when the autocatalytic core becomes dominant. An example of a mechanism involving such oscillatory core has been described by Semenov et al. \cite{Semenov2015}. 

\paragraph{Oscillatory core (I,b).} The second oscillatory core of class I in our selection is:
\begin{equation}\label{eq:ocIb}
\begin{pmatrix}
    -1 & 0 & 1 & 1 \\
     1 & -1 & 0 & 0\\
     1 & 1 & -1 & 0\\
     0 & 0 & -1 & -1
    \end{pmatrix}, \
\vcenter{\hbox{\includegraphics[scale=0.3]{core12.pdf}}}
\end{equation}
Such oscillatory core is built on the second type of the autocatalytic cores from \cite{blokhuis20},
\begin{equation}
\begin{pmatrix}
    -1 & 0 & 1\\
    1 & -1 & 0\\
    1 & 1 & -1
\end{pmatrix}, \ \vcenter{\hbox{\includegraphics[scale=0.3]{ACtypeII.pdf}}}
\end{equation}
highlighted in magenta in the figure, and enlarged to the $4\times4$ Hurwitz-stable matrix \eqref{eq:ocIb}, with eigenvalues approximately $(-0.13 \pm  0.5i;
  -1.8660 \pm 0.5i)
$. Recipe \ref{recipe:1main} implies that any consistent network containing the following four reactions admits periodic oscillations for a certain choice of parameters if endowed with parameter-rich kinetics.
\begin{equation}
\begin{cases}
\ce{X} + ... &\rightarrow \quad\ce{Y} + \ce{Z} + ...\\
\ce{Y}  + ... &\rightarrow \quad  \ce{Z} + ...\\
\ce{Z}+\ce{W}+...&\rightarrow \quad \ce{X}+...\\
\ce{W}+... &\rightarrow \quad \ce{X}+...
\end{cases}
\end{equation}
Above, `...' indicates that the reactions may involve other species different from $\ce{X}, \ce{Y}, \ce{Z}, \ce{W}$. Periodic oscillations appear when the autocatalytic core becomes dominant. An example of a mechanism involving such oscillatory core has been described by Semenov et al. \cite{Semenov2016}. 

\paragraph{Oscillatory core (I,c).} The third oscillatory core of class I in our selection is:
\begin{equation}\label{eq:ocIc}\begin{pmatrix}
    -1 & 1 & 1 & 2 \\
     1 & -1 & 0 & 0\\
     1 & 0 & -1 & 0\\
     0 & 0 & -1 & -1
    \end{pmatrix}, \ \vcenter{\hbox{\includegraphics[scale=0.3]{coreIcfirst_x.pdf}}}
\end{equation}
Such oscillatory core is built on the third type of the autocatalytic cores from \cite{blokhuis20},
\begin{equation}
\begin{pmatrix}
    -1 & 1 & 1\\
    1 & -1 & 0\\
    1 & 0 & -1
\end{pmatrix}, \ \ \vcenter{\hbox{\includegraphics[scale=0.3]{ACtypeIII.pdf}}}
\end{equation}
highlighted in magenta in the figure, and enlarged to the $4\times4$ Hurwitz-stable matrix \eqref{eq:ocIc}, with eigenvalues approximately $(-1, -2.77,
  -0.12 \pm 0.59i)$. Recipe \ref{recipe:1main} implies that any consistent network containing the following four reactions admits periodic oscillations for a certain choice of parameters if endowed with parameter-rich kinetics.
\begin{equation}
\begin{cases}
\ce{X} + ... &\rightarrow \quad\ce{Y} + \ce{Z} +...\\
\ce{Y}  + ... &\rightarrow \quad  \ce{X} + ...\\
\ce{Z}+\ce{W}+...&\rightarrow \quad \ce{X}+...\\
\ce{W}+... &\rightarrow \quad 2\ce{X}+...
\end{cases}
\end{equation}
Above, `...' indicates that the reactions may involve other species different from $\ce{X}, \ce{Y}, \ce{Z}, \ce{W}$. Periodic oscillations appear when the autocatalytic core becomes dominant. 

\paragraph{Oscillatory core (I,d).} The fourth oscillatory core of class I in our selection is:
\begin{equation}\label{eq:ocId}
\begin{pmatrix}
    -1 & 1 & 1 & 1 \\
     1 & -1 & 0 & 1\\
     1 & 1 & -1 & 0\\
     0 & -1 & -1 & -1
    \end{pmatrix}, \
\vcenter{\hbox{\includegraphics[scale=0.3]{core15.pdf}}}
\end{equation}
Such oscillatory core is built on the fourth type of the autocatalytic cores from \cite{blokhuis20},
\begin{equation}
\begin{pmatrix}
    -1 & 1 & 1\\
    1 & -1 & 0\\
    1 & 1 & -1
\end{pmatrix}, \  \ \vcenter{\hbox{\includegraphics[scale=0.3]{ACtypeIV.pdf}}}
\end{equation}
highlighted in magenta in the figure, and enlarged to the $4\times4$ Hurwitz-stable matrix \eqref{eq:ocId}, with eigenvalues approximately $(-2.39;-1.43;-0.09 \pm 0.93i)$. Recipe \ref{recipe:1main} implies that any consistent network containing the following four reactions admits periodic oscillations for a certain choice of parameters if endowed with parameter-rich kinetics.
\begin{equation}
\begin{cases}
\ce{X} + ... &\rightarrow \quad\ce{Y} + \ce{Z}+ ...\\
\ce{Y} + \ce{W} + ... &\rightarrow \quad  \ce{X} + \ce{Z} +...\\
\ce{Z}+\ce{W}+...&\rightarrow \quad \ce{X}+...\\
\ce{W}+... &\rightarrow \quad \ce{X}+\ce{Y}+...
\end{cases}
\end{equation}
Above, `...' indicates that the reactions may involve other species different from $\ce{X}, \ce{Y}, \ce{Z}, \ce{W}$. Periodic oscillations appear when the autocatalytic core becomes dominant.

\paragraph{Oscillatory core (I,e).} The fifth oscillatory core of class I in our selection is:
\begin{equation}\label{eq:ocIe}
\begin{pmatrix}
    -1 & 1 & 1 & 1 \\
     1 & -1 & 1 & 1\\
     1 & 1 & -1 & 0\\
     -1 & -1 & -1 & -1
    \end{pmatrix}\ \vcenter{\hbox{\includegraphics[scale=0.3]{coreIefirst_x.pdf}}}
\end{equation}
Such oscillatory core is built on the fifth type of the autocatalytic cores from \cite{blokhuis20},
\begin{equation}
\begin{pmatrix}
    -1 & 1 & 1\\
    1 & -1 & 1\\
    1 & 1 & -1
\end{pmatrix}, \ \ \ \vcenter{\hbox{\includegraphics[scale=0.3]{ACtypeV.pdf}}}
\end{equation}
highlighted in magenta in the figure, and enlarged to the $4\times4$ matrix \eqref{eq:ocIe}, with eigenvalues $(-2;-2, \pm i)$. As such matrix is not Hurwitz-stable matrix, it is not literally an Oscillatory Core of class I, following Def.~\ref{def:Ocores1}. However, with a slight abuse of notation, we retain it here for consistency with the other cores and for its relevance in identifying a nontrivial matrix with purely imaginary eigenvalues. Such stoichiometry still produces periodic orbits because \eqref{eq:ocIe} is a $D$-Hopf matrix. It is indeed invertible, and by choosing $D(\beta)=\operatorname{diag}(1,1,1,\beta)$, for $\beta>0$ we can see that  
\begin{equation}
    \begin{pmatrix}
       -1 & 1 & 1 & 1 \\
     1 & -1 & 1 & 1\\
     1 & 1 & -1 & 0\\
     -1 & -1 & -1 & -1 
    \end{pmatrix}
    \begin{pmatrix}
    1 & 0 & 0 & 0\\
     0 & 1 & 0 & 0\\
     0 & 0 & 1 & 0\\
     0 & 0 & 0 & \beta
     \end{pmatrix}
\end{equation}
is Hurwitz-stable if and only if $\beta>1$ and Hurwitz-unstable if and only if $\beta<1$. Note that $\beta$ is a parameter that tunes the dominance of the autocatalytic core. Corollary \ref{cor:dhopfoscillationsmain} implies then that any consistent network containing the following four reactions admits periodic oscillations for a certain choice of parameters if endowed with parameter-rich kinetics.
\begin{equation}
\begin{cases}
\ce{X} + \ce{W}+ ... &\rightarrow \quad\ce{Y} + \ce{Z} +...\\
\ce{Y} + \ce{W} + ... &\rightarrow \quad  \ce{X} + \ce{Z} +...\\
\ce{Z}+\ce{W}+...&\rightarrow \quad \ce{X}+\ce{Y}+...\\
\ce{W}+... &\rightarrow \quad \ce{X}+\ce{Y}+...
\end{cases}
\end{equation}
Above, `...' indicates that the reactions may involve other species different from $\ce{X}, \ce{Y}, \ce{Z}, \ce{W}$. Periodic oscillations appear when the autocatalytic core becomes dominant.

\subsection{Recipe II, negative-feedback and the principle of length}\label{sec:recipeII}
\addtocontents{toc}{\protect\setcounter{tocdepth}{0}}

For Recipe \ref{recipe:2main}, we chose five unstable-negative feedbacks whose structure  echoes the five autocatalytic cores \eqref{autocat_cores} from \cite{blokhuis20}. In particular, we swap the sign of the off-diagonal entries of the first row, which intuitively switches the sign of the feedback from positive to negative. The other entries are left unaltered. 
Again, for simplicity, we assume here that the networks are nonambiguous, and thus the stoichiometric matrix univocally characterizes the network. The crucial observation is that the Hurwitz-instability of a negative feedback may appear only upon inclusion of a sufficient number of intermediates. Let us exemplify this general principle in the first of our examples in Tab.~\ref{tab:negmain}:
\begin{equation}
\begin{pmatrix}
   -1 & -2\\
   1 & -1\\
\end{pmatrix},
\end{equation}
with stable eigenvalues $(-1\pm 1.41i)$. Such stoichiometry reflects chemical processes as
\begin{equation}
\begin{cases}
\ce{X} + ... &\rightarrow \quad \ce{Y} + ...\\
2\ce{X} + \ce{Y} + ... &\rightarrow \quad ...
\end{cases},
\end{equation}
where, again, `...' indicates any combination of species, possibly empty, in the network that do not include $\ce{X}$ and $\ce{Y}$. The instability appears only upon inclusion of a sufficient number of intermediates, e.g. in
\begin{equation}
\ce{X} \rightarrow \ce{X_1} \rightarrow \ce{X_2} \rightarrow \ce{X_3} \rightarrow \ce{X_4} \rightarrow \ce{X_5}  \rightarrow \ce{X_6} \rightarrow \ce{Y},
\end{equation}
which indeed give rise to an associated CS-matrix:
\begin{equation}\label{eq:o21}
\begin{pmatrix}
-1 & 0  & 0  & 0  & 0  & 0  & 0  & -2 \\
1  & -1 & 0  & 0  & 0  & 0  & 0  & 0  \\
0  & 1  & -1 & 0  & 0  & 0  & 0  & 0  \\
0  & 0  & 1  & -1 & 0  & 0  & 0  & 0  \\
0  & 0  & 0  & 1  & -1 & 0  & 0  & 0  \\
0  & 0  & 0  & 0  & 1  & -1 & 0  & 0  \\
0  & 0  & 0  & 0  & 0  & 1  & -1 & 0  \\
0  & 0  & 0  & 0  & 0  & 0  & 1  & -1
\end{pmatrix}, \ \vcenter{\hbox{\includegraphics[scale=0.3]{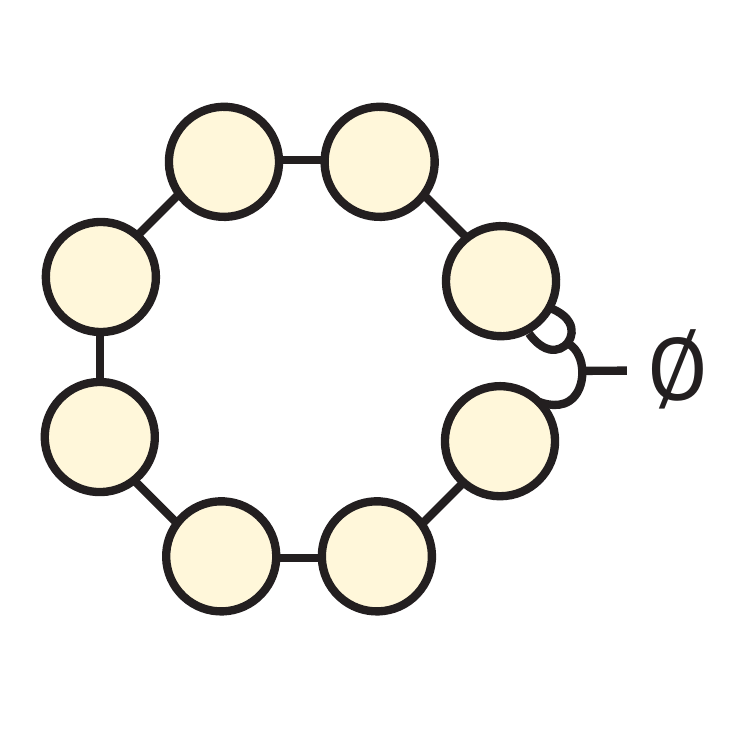}}}
\end{equation}
with eigenvalues $(\mathbf{+0,0075\pm0.4i}; -2\pm 0.4i;-1.4\pm i;0.6\pm i)$. The fact that \eqref{eq:o21} is a $P^-_{FF}$ matrix is also clear, as we can simply consider the leading principal submatrices, that are all lower triangular with negative diagonal. Thus, \eqref{eq:o21} is an oscillatory core of class II and thus Recipe \ref{recipe:2main} guarantees the possibility for oscillatory behavior in any consistent network containing a chemistry as 
\begin{equation}
\begin{cases}
\ce{X} + ... \rightarrow \ce{X_1} \rightarrow \ce{X_2} \rightarrow \ce{X_3} \rightarrow \ce{X_4} \rightarrow \ce{X_5}  \rightarrow \ce{X_6} \rightarrow \ce{Y}+...\\
2\ce{X} + \ce{Y} + ... \quad \rightarrow \quad ...
\end{cases},
\end{equation}
where now `...' refers to any combination of the species, possibly empty, not including $\ce{X},\ce{Y}$ and $\ce{X}_i$ for $i=1,...,6$. In principle, each intermediate $\ce{X}_i$ could also be substituted with $\ce{X}_i +$ `...', of course. The instability grows monotonically with respect to the number of intermediates: the longer the sequence of intermediates, the more unstable the network is. We refer to this principle as the 
$$\textbf{Principle of Length.}$$

{\paragraph{Geometrical intuition for the Principle of Length.} The eigenvalues can be computed explicitly for the stoichiometric matrix associated to the class of networks:
\begin{equation}
\begin{cases}
\ce{X} \quad\rightarrow\quad \ce{X_1} &\rightarrow \quad  ... \quad \rightarrow \quad \ce{X_n} \quad \rightarrow \quad \ce{Y},\\
m \ce{X}+\ce{Y} \quad &\rightarrow \quad...
\end{cases},
\end{equation}
for general $n,m>0$. In fact, the stoichiometric matrix may be written as:
\begin{equation}
    \mathbb{S}=
\begin{pmatrix}
-1 & 0  & ...  &  0  & -m \\
1  & -1 & ... & 0  & 0  \\
\vdots& & & &\vdots \\
0  & 0  & ... & -1 & 0  \\
0  & 0  &... & 1  & -1
\end{pmatrix}=\begin{pmatrix}
0 & 0  & ...  &  0  & -m \\
1  & 0 & ... & 0  & 0  \\
\vdots& & & &\vdots \\
0  & 0  & ... & 0 & 0  \\
0  & 0  &... & 1  & 0
\end{pmatrix}-\begin{pmatrix}
1 & 0  & ...  &  0  & 0 \\
0  & 1 & ... & 0  & 0  \\
\vdots& & & &\vdots \\
0  & 0  & ... & 1 & 0  \\
0  & 0  &... & 0  & 1
\end{pmatrix}.
\end{equation}
Leveraging the root-of-unit structure of the first matrix, the $n+2$ eigenvalues $\lambda_k$ are given by the explicit formula
\begin{equation}
 \lambda_k = m^{\frac{1}{n+2}} e^{\frac{(2k+1)\pi i}{n+2}}-1,
\end{equation}
thus complex pair of eigenvalues with positive real part appear for any $m>1$, if $n$ is sufficiently large. While the same intuition appears instructive also for the other treated examples, a full linear algebra generalization of this argument is beyond the scope of this paper.\\

%We underline that the principle of length applies only to Recipe \ref{recipe:2main}, which is based on unstable-negative feedback. As the instability of unstable-positive feedbacks is characterized by a determinant computation, the principle of length does not apply and the magnification conjecture still stands, in that case.

See table \ref{tab:negmain} for an overview of five oscillatory cores of class II. The depicted value of $n$ is the minimum number of intermediates that must be added in the reaction corresponding to the first column to achieve Hurwitz-instability.

\paragraph{Oscillatory Core (II,a).} The basic structure of the core,
$$\begin{pmatrix}
-1 & -2\\
1 & -1
\end{pmatrix}, \  \vcenter{\hbox{\includegraphics[scale=0.3]{Osc_coreIIa_z.pdf}}}$$
echoes an autocatalytic core of type I:
$$\begin{pmatrix}
-1 & 2\\
1 & -1
\end{pmatrix}.$$
We refer to the previous paragraph for the detailed analysis of this core.

\paragraph{Oscillatory Core (II,b).}
The basic structure of the core,
$$\begin{pmatrix}
-1 & 0 & -1\\
1 & -1 & 0\\
1 & 1 & -1
\end{pmatrix}, \ \vcenter{\hbox{\includegraphics[scale=0.3]{Osc_coreIIb_w.pdf}}} $$
with stable eigenvalues approximately   $(-1.68, -0.66 \pm 1.16i)$, echoes an autocatalytic core of type II:
$$\begin{pmatrix}
-1 & 0 & 1\\
1 & -1 & 0\\
1 & 1 & -1
\end{pmatrix},$$
where the nonzero entries in the first row have been swapped with sign. By adding $n=6$ intermediates to the reaction in the first column, the matrix becomes:
\begin{equation}\label{eq:o2b}
\begin{pmatrix}
-1 & 0 & 0 & 0 & 0 & 0 & 0 & 0 & -1 \\
1 & -1 & 0 & 0 & 0 & 0 & 0 & 0 & 0 \\
0 & 1 & -1 & 0 & 0 & 0 & 0 & 0 & 0 \\
0 & 0 & 1 & -1 & 0 & 0 & 0 & 0 & 0 \\
0 & 0 & 0 & 1 & -1 & 0 & 0 & 0 & 0 \\
0 & 0 & 0 & 0 & 1 & -1 & 0 & 0 & 0 \\
0 & 0 & 0 & 0 & 0 & 1 & -1 & 0 & 0 \\
0 & 0 & 0 & 0 & 0 & 0 & 1 & -1 & 0 \\
0 & 0 & 0 & 0 & 0 & 0 & 1 & 1 & -1
\end{pmatrix}, 
\end{equation}
with eigenvalues approximately:
\begin{equation}
(\mathbf{+0.0094 \pm 0.39i}, -1.82, -1.79\pm 0.5i, -1.28 \pm 0.98i,-0.53\pm 0.95i),
\end{equation}
and it is thus Hurwitz-unstable. The fact that it is $P^-_{FF}$ matrix is also clear, as we can simply consider the leading principal submatrices, that are all lower triangular with negative diagonal. Thus, \eqref{eq:o2b} is an oscillatory core of class II and Recipe \ref{recipe:2main} guarantees the possibility for periodic oscillations in any consistent network containing reactions as:
\begin{equation}
\begin{cases}
\ce{X}+...\rightarrow \ce{X_1} \rightarrow \ce{X_2} \rightarrow \ce{X_3} \rightarrow \ce{X_4} \rightarrow \ce{X_5}  \rightarrow \ce{X_6} \rightarrow \ce{Y}+\ce{Z}+...\\
\ce{Y}+...\quad \rightarrow \quad \ce{Z}+...\\
\ce{X}+\ce{Z}+...\quad \rightarrow \quad ...
\end{cases}
\end{equation}
where `...' refers to any combination of the species, possibly empty, not including $\ce{X},\ce{Y}$ and $\ce{X}_i$ for $i=1,...,6$. Each intermediate $\ce{X}_i$ could in principle be substituted with $\ce{X}_i +$ `...' .

\paragraph{Oscillatory Core (II,c).}
The basic structure of the core,
$$\begin{pmatrix}
-1 & -1 & -1\\
1 & -1 & 0\\
1 & 0 & -1
\end{pmatrix}, \ \vcenter{\hbox{\includegraphics[scale=0.3]{Osc_coreIIc_new.pdf}}}$$
with stable eigenvalues approximately   $(-1,  -1 \pm 1.41 i)$, echoes an autocatalytic core of type III:
$$\begin{pmatrix}
-1 & 1 & 1\\
1 & -1 & 0\\
1 & 0 & -1
\end{pmatrix},$$
where the nonzero entries in the first row have been swapped with sign. By adding $n=6$ intermediates to the reaction in the first column, the matrix becomes:
\begin{equation}\label{eq:o2c}
\begin{pmatrix}
-1 & 0 & 0 & 0 & 0 & 0 & 0 & -1 & -1 \\
1 & -1 & 0 & 0 & 0 & 0 & 0 & 0 & 0 \\
0 & 1 & -1 & 0 & 0 & 0 & 0 & 0 & 0 \\
0 & 0 & 1 & -1 & 0 & 0 & 0 & 0 & 0 \\
0 & 0 & 0 & 1 & -1 & 0 & 0 & 0 & 0 \\
0 & 0 & 0 & 0 & 1 & -1 & 0 & 0 & 0 \\
0 & 0 & 0 & 0 & 0 & 1 & -1 & 0 & 0 \\
0 & 0 & 0 & 0 & 0 & 0 & 1 & -1 & 0 \\
0 & 0 & 0 & 0 & 0 & 0 & 1 & 0 & -1
\end{pmatrix},
\end{equation}
with eigenvalues approximately:
\begin{equation}
(\mathbf{+  0.0075 \pm 0.42i}, -2 \pm 0.42i,-1.42\pm 1.01i, -0.58\pm 1.01i),
\end{equation}
and it is thus Hurwitz-unstable. The fact that it is $P^-_{FF}$ matrix is also clear, as we can simply consider lower triangular principal submatrices with negative diagonal. Thus, \eqref{eq:o2c} is an oscillatory core of class II and Recipe \ref{recipe:2main} guarantees the possibility for periodic oscillations in any consistent network containing reactions as:
\begin{equation}
\begin{cases}
\ce{X}+...\rightarrow \ce{X_1} \rightarrow \ce{X_2} \rightarrow \ce{X_3} \rightarrow \ce{X_4} \rightarrow \ce{X_5}  \rightarrow \ce{X_6} \rightarrow \ce{Y}+\ce{Z}+...\\
\ce{X}+\ce{Y}+...\quad \rightarrow \quad ...\\
\ce{X}+\ce{Z}+...\quad \rightarrow \quad ...
\end{cases}
\end{equation}
where `...' refers to any combination of the species, possibly empty, not including $\ce{X},\ce{Y}$ and $\ce{X}_i$ for $i=1,...,6$. Each intermediate $\ce{X}_i$ could in principle be substituted with $\ce{X}_i +$ `...' . The mechanism of this Oscillatory Core (II,c) was central for the model described by Hiver \cite{hyver1978}, as an example of a nonautocatalytic system that showed periodic oscillations. The analysis in that case was under mass action kinetics, see also \cite{Vassena2025}.

\paragraph{Oscillatory Core (II,d).}
The basic structure of the core,
$$\begin{pmatrix}
-1 & -1 & -1\\
1 & -1 & 0\\
1 & 1 & -1
\end{pmatrix},$$
with stable eigenvalues approximately   $(-1.45,  -0.77 \pm 1.47i)$, echoes an autocatalytic core of type IV:
$$\begin{pmatrix}
-1 & 1 & 1\\
1 & -1 & 0\\
1 & 1 & -1
\end{pmatrix}, \ \vcenter{\hbox{\includegraphics[scale=0.3]{Osc_coreIId.pdf}}}$$
where the nonzero entries in the first row have been swapped with sign. By adding $n=3$ intermediates to the reaction in the first column, the matrix becomes:
\begin{equation}\label{eq:o2d}
\begin{pmatrix}
-1 & 0 & 0 & 0 & -1 & -1 \\
1 & -1 & 0 & 0 & 0 & 0 \\
0 & 1 & -1 & 0 & 0 & 0 \\
0 & 0 & 1 & -1 & 0 & 0 \\
0 & 0 & 0 & 1 & -1 & 0 \\
0 & 0 & 0 & 1 & 1 & -1
\end{pmatrix},
\end{equation}
with eigenvalues approximately:
\begin{equation}
(\mathbf{+0.0118 \pm 0.6840i}, -2,-1.51
  -1.26 \pm 1.12i),
\end{equation}
and it is thus Hurwitz-unstable. The fact that it is $P^-_{FF}$ matrix is also clear, as we can simply consider lower triangular principal submatrices with negative diagonal. Thus, \eqref{eq:o2d} is an oscillatory core of class II and Recipe \ref{recipe:2main} guarantees the possibility for periodic oscillations in any consistent network containing reactions as:
\begin{equation}
\begin{cases}
\ce{X}+...\rightarrow \ce{X_1} \rightarrow \ce{X_2} \rightarrow \ce{X_3} \rightarrow \ce{Y}+\ce{Z}+...\\
\ce{X}+\ce{Y}+...\quad \rightarrow \quad \ce{Z}+...\\
\ce{X}+\ce{Z}+...\quad \rightarrow \quad ...
\end{cases}
\end{equation}
where `...' refers to any combination of the species, possibly empty, not including $\ce{X},\ce{Y}$ and $\ce{X}_i$ for $i=1,...,3$. Each intermediate $\ce{X}_i$ could in principle be substituted with $\ce{X}_i +$ `...' . 

\paragraph{Oscillatory Core (II,e).}
The basic structure of the core,
$$\begin{pmatrix}
-1 & -1 & -1\\
1 & -1 & 1\\
1 & 1 & -1
\end{pmatrix}, \ \vcenter{\hbox{\includegraphics[scale=0.3]{Osc_coreIIe.pdf}}}$$
with stable eigenvalues approximately   $(-2,  -0.5 \pm 1,32i)$, echoes an autocatalytic core of type V:
$$\begin{pmatrix}
-1 & 1 & 1\\
1 & -1 & 1\\
1 & 1 & -1
\end{pmatrix},$$
where the nonzero entries in the first row have been swapped with sign. By adding $n=1$ intermediates to the reaction in the first column, the matrix becomes:
\begin{equation}\label{eq:o2e}
\begin{pmatrix}
-1 & 0 & -1 & -1\\
1 & -1 & 0 & 0\\
0 & 1 & -1 & 1\\
0 &1 & 1 & -1
\end{pmatrix},
\end{equation}
with eigenvalues:
\begin{equation}
(-2,-2, \pm i),
\end{equation}
with purely-imaginary eigenvalues, similarly to the parallel case with positive-feedback Oscillatory Core (I,e). As such matrix is not Hurwitz-stable matrix, it is not literally an Oscillatory Core of class II, following Def.~\ref{def:Ocores2}. However, with a slight abuse of notation, we retain it here for consistency with the other cores and for its relevance in identifying a nontrivial matrix with purely imaginary eigenvalues. Such stoichiometry still produces periodic orbits because \eqref{eq:o2e} is a $D$-Hopf matrix. It is indeed invertible, and actually a $P^-_{FF}$ matrix, as all leading $k$-principal minors are of sign $(-1)^k$. Moreover, by choosing $D(\beta)=\operatorname{diag}(1,1,1,\beta)$, for $\beta>0$ we can see that  
\begin{equation}
  \begin{pmatrix}
-1 & 0 & -1 & -1\\
1 & -1 & 0 & 0\\
0 & 1 & -1 & 1\\
0 &1 & 1 & -1
\end{pmatrix}  
    \begin{pmatrix}
    1 & 0 & 0 & 0\\
     0 & 1 & 0 & 0\\
     0 & 0 & 1 & 0\\
     0 & 0 & 0 & \beta
     \end{pmatrix}
\end{equation}
is Hurwitz-unstable if and only if $\beta>1$ and Hurwitz-stable if and only if $\beta<1$. Corollary \ref{cor:dhopfoscillationsmain} concludes that any consistent network containing the following four reactions admits periodic oscillations for a certain choice of parameters if endowed with parameter-rich kinetics:
\begin{equation}
\begin{cases}
\ce{X}+... \rightarrow \quad \ce{X_1}+...\rightarrow \ce{Y}+\ce{Z}+...\\ 
\ce{X}+\ce{Y}+... \quad \rightarrow \quad \ce{Z}+...\\
\ce{X}+\ce{Z}+...\quad \rightarrow \quad \ce{Y}...
\end{cases}
\end{equation}
where `...' refers to any combination of the species, possibly empty, not including $\ce{X},\ce{X_1},\ce{Y},\ce{Z}$.

\section{Recipe 0}\label{sec:recipe0}
\addtocontents{toc}{\protect\setcounter{tocdepth}{-1}}

\subsection{An essential construction based on autocatalytic cores}\label{sec:recipe0basicconstr}
\addtocontents{toc}{\protect\setcounter{tocdepth}{-1}}

We outline a particularly simple construction of a chemical oscillator. We explain the construction in detail based on an autocatalytic core of type II:
\begin{equation}
    \begin{cases}
        \ce{X}\quad\underset{1}{\rightarrow}\quad \ce{Y}+\ce{Z}\\
        \ce{Y} \quad\underset{2}{\rightarrow} \quad\ce{Z}\\
         \ce{Z} \quad\underset{3}{\rightarrow} \quad\ce{X}\\
    \end{cases},
\end{equation}
with associated stoichiometric matrix:
\begin{equation}
    \mathbb{S}=\begin{pmatrix}
        -1 & 0 & 1\\
        1 & -1 & 0\\
        1 & 1 & -1
    \end{pmatrix}.
\end{equation}

The construction focuses on the species $\ce{X}$ and modifies the autocatalytic core by adding one species $\ce{W}$ and two reactions 4 and 5 as follows. The species $\ce{W}$ co-participates - together with $\ce{X}$ - as reactant in reaction 1. Reaction 4 is just an outflow  reaction from $\ce{X}$, and reaction 5 is just an inflow reaction to $\ce{W}$. We then obtain the following enlarged network.

\begin{equation}\label{eq"recipe0ex}
    \begin{cases}
        \ce{X} +\ce{W} \quad&\underset{1}{\rightarrow}\quad \ce{Y}+\ce{Z}\\
        \ce{Y} \quad&\underset{2}{\rightarrow} \quad\ce{Z}\\
         \ce{Z} \quad&\underset{3}{\rightarrow} \quad\ce{X}\\
         \ce{X} \quad&\underset{4}{\rightarrow}\\
         &\underset{5}{\rightarrow}\quad \ce{W}
    \end{cases},
\end{equation}
with associated stoichiometric matrix:
\begin{equation}
    \mathbb{S}=\begin{pmatrix}
        -1 & 0 & 1 & -1& 0\\
        1 & -1 & 0 & 0&0\\
        1 & 1 & -1 & 0&0\\
        -1 & 0 & 0 & 0 &1\\
    \end{pmatrix}.
\end{equation}
We note that such network is consistent as it admits a steady-state: $\mathbb{S}$ has a positive kernel vector $v$ of the form
\begin{equation}\label{eq:recipe0ss}
    v=(c,c,2c,c,c)^T, 
\end{equation}
for $c\in\mathbb{R}_{>0}.$ It is also relevant to note that the only reaction in which $\ce{W}$ participates as a reactant is reaction 1. As a consequence, such network, made of four species and four reactions, possesses a unique (!) $4$-CS $\pmb{\kappa}$ with CS-bijection defined as:
$$J(\ce{X},\ce{Y},\ce{Z},\ce{W})=(4,2,3,1),$$
and CS-matrix
\begin{equation}\label{eq:recipe0CS}
    \mathbb{S}[\pmb{\kappa}]=\begin{pmatrix}
        -1 & 0 & 1 & -1\\
        0 & -1 & 0 & 1\\
        0 & 1 & -1 & 1\\
        0 & 0 & 0 & -1\\
    \end{pmatrix},
\end{equation}
which is invertible, and thus the Jacobian of the system is invertible for all choices of parameters. Secondly, $\mathbb{S}[\pmb{\kappa}]$ is Hurwitz-stable with eigenvalues $(-1,-1,-1,-1)$ and thus via Prop.~\ref{prop:stability}, the network admits stability. Thirdly, since the network is autocatalytic as it contains an autocatalytic core, it admits instability via Prop.~\ref{prop:thm711}. Summarizing, we have the following three features:
\begin{enumerate}
    \item The Jacobian of the network is invertible for all choices of parameters.
    \item The network admits stability.
    \item The network admits instability.
\end{enumerate}
Recipe \ref{recipe:0main} applies and we conclude that the network admits periodic orbits for a certain choice of parameters. The same (sub)network is at the basis of a literature oscillator in an activator/inhibitor model \cite{Nguyen18}, which we address in detail in Sec.~\ref{sec:activatorinhibitor}.\\

We include as well an explicit numerical analysis of the resulting associated dynamical system associated to the network \eqref{eq"recipe0ex}. To do so, we firstly inspect the symbolic Jacobian:
\begin{equation}
\begin{split}
    G=\mathbb{S}R=&\begin{pmatrix}
        -1 & 0 & 1 & -1& 0\\
        1 & -1 & 0 & 0&0\\
        1 & 1 & -1 & 0&0\\
        -1 & 0 & 0 & 0 &1\\
    \end{pmatrix}\begin{pmatrix}
        R_{1X} & 0 & 0 & R_{1W}\\
        0 & R_{2Y} & 0 & 0\\
        0 & 0 & R_{3Z} & 0\\
        R_{4X} & 0 & 0 & 0\\
        0 & 0 & 0 & 0 & 0
    \end{pmatrix}\\
    &=\begin{pmatrix}
        -R_{4X}-R_{1X} & 0 & R_{3Z} & R_{1W}\\
        R_{1X} & -R_{2Y} & 0 & R_{1W}\\
        R_{1X}& R_{2Y} & -R_{3Z} & R_{1W}\\
        -R_{1X} & 0 & 0 & -R_{1W}.
    \end{pmatrix}
    \end{split}
\end{equation}
with invertible determinant
\begin{equation}
\operatorname{det}G=R_{4X}R_{2Y}R_{3Z}R_{1W},
\end{equation}
corresponding to \eqref{eq:recipe0CS}. In turn, the leading principal matrix
\begin{equation}
G[k]=
\begin{pmatrix}
        -R_{4X}-R_{1X} & 0 & R_{3Z} \\
        R_{1X} & -R_{2Y} & 0 \\
        R_{1X}& R_{2Y} & -R_{3Z}
    \end{pmatrix}
\end{equation}
has determinant
\begin{equation}
    \operatorname{det}G[k]=(R_{1X}-R_{4X})R_{2Y}R_{3Z}
\end{equation}
and it thus shows that the instability of such principal matrix necessarily requires
\begin{equation}\label{eq:recipe0cond}
    R_{1X}>R_{4X},
\end{equation}
i.e., the rate derivative along the autocatalytic cycle must be greater than the derivative along the outflow reaction, which we can interpret as the dominance of the autocatalytic cycle. Eq.~\ref{eq:recipe0cond} is sufficient to obtain periodic orbits: in view of the steady-state constraints \eqref{eq:recipe0ss} we can then utilize mass action kinetics for the entire network with the exception of the outflow reaction 4, where we need a sublinear rate. A natural choice is Michaelis-Menten, obtaining the following system of ODEs:
\begin{equation}\label{eq:recipe0eqmm}
\begin{cases}
\dot{x}=-k_1xw+k_3z-k_4\dfrac{x}{1+b_4x}\\
\dot{y}=k_1xw-k_2y\\
\dot{z}=k_1xw+k_2y-k_3z\\
\dot{w}=F_5-k_1xw
\end{cases}
\end{equation}
To spot the periodic orbits, we fix the steady state fluxes \eqref{eq:recipe0ss} arbitrarily for $c=1$. Then, recipe \ref{recipe:0main} requires making the unstable principal matrix dominant, via Lemma \ref{lem:mainGDmain}. To do so, the steady state concentration of species $\ce{W}$ must be large enough: this implies that its inverse is small enough, and thus the autocatalytic core, which does not include species $\ce{W}$, becomes dominant. We then fix the steady state concentration to be 
\begin{equation}\label{recipe0:eqmmss}
(\bar{x},\bar{y},\bar{z},\bar{w})=(1,1,1,\bar{w}),
\end{equation}
which, jointly with the fixed steady-state fluxes, implies a choice of rates:
\begin{equation}
\begin{cases}
    k_1=\dfrac{1}{\bar{w}}\\
    F_5=k_2=b_4=1\\
    k_3=k_4=2\\
\end{cases}.
\end{equation}
We underline that the Michaelis-Menten parameters $(k_4,a_4)$ could also be chosen differently, in principle. Now, for all choices of $\bar{w}$, the concentrations \eqref{recipe0:eqmmss} are steady-state concentrations. However, such steady state loses stability for large value of $\bar{w}$, generating a stable periodic orbit appear. See Fig.~\ref{fig:recipe0MM}.

\begin{figure}
    \centering
     \includegraphics[width=0.75\linewidth]{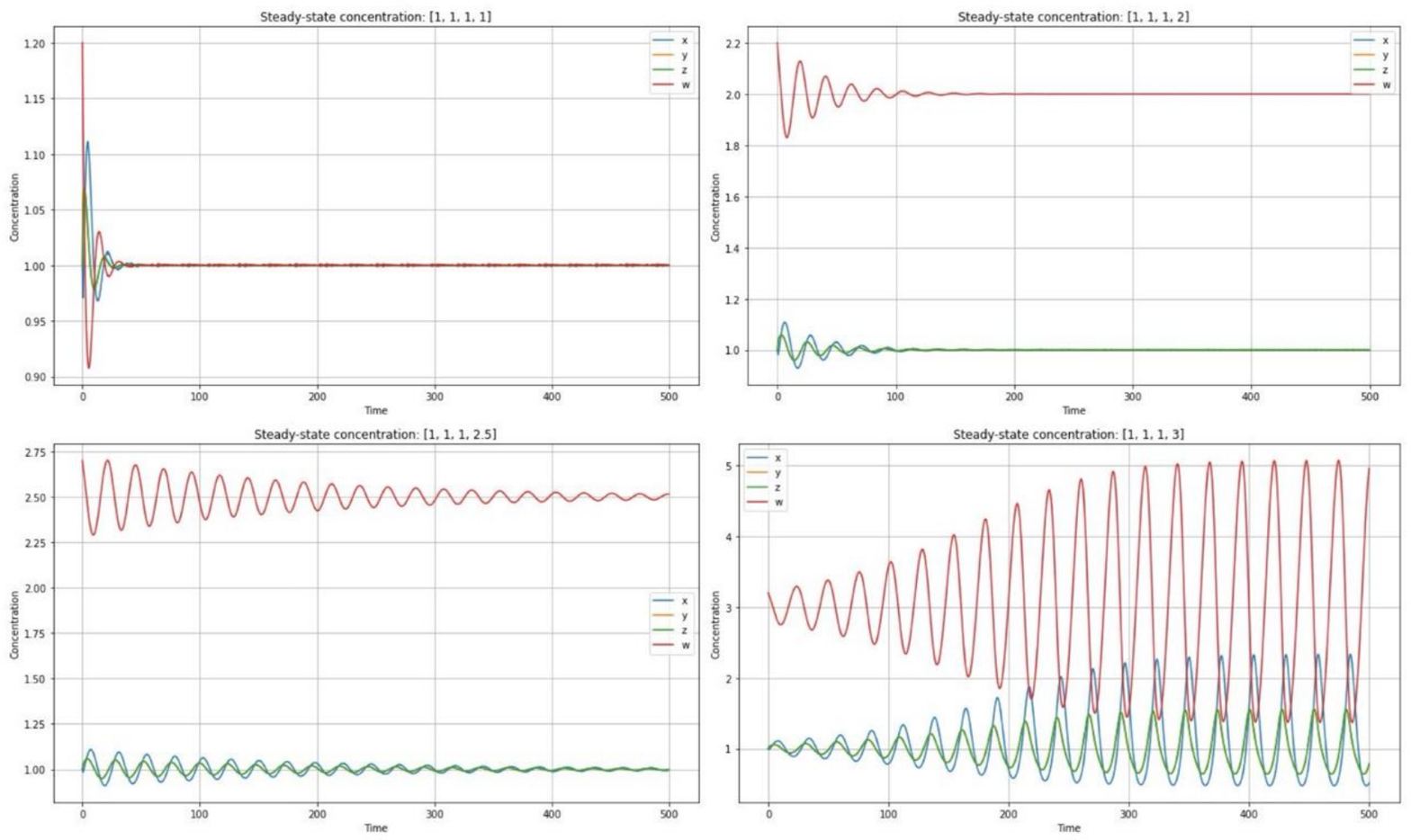}
    \caption{Numerical simulations for the system \eqref{eq:recipe0eqmm}, with parameter values: $k_1=\dfrac{1}{\bar{w}};$ $
    F_5=k_2=b_4=1;$
    $k_3=k_4=2$. This choice of parameters fixes the steady-state value at $(\bar{x},\bar{y},\bar{z},\bar{w})=(1,1,1,\bar{w})$ for any choice of $\bar{w}$.     
    The initial conditions are set accordingly near the steady-state value at $(x_0,y_0,z_0,w_0)=(\bar{x},\bar{y},\bar{z},\bar{w}+0.2)=(1,1,1,\bar{w}+0.2)$. The figure shows the plots for four different values of $\bar{w}$. At $\bar{w}=1$ the steady-state is stable and strongly attracts the dynamics. At $\bar{w}=\{2,2.5\}$ the stability gradually decreases and damped oscillations can be seen. At $\bar{w}=3$ the steady-state is already unstable and a stable periodic orbit has appeared.}
    \label{fig:recipe0MM}
\end{figure}

Moreover, inspired by the fact that any Michaelis-Menten step 
\begin{equation}
    \ce{A}\quad \rightarrow \quad \ce{B}
\end{equation} 
can be seen as the singular limit of a mass action scheme
\begin{equation}
\ce{A}+\ce{E} \quad \rightleftharpoons \quad \ce{I} \quad \rightarrow \quad \ce{B}+\ce{E},
\end{equation}
we can obtain an analogous oscillatory  behavior for the extended mass action network
\begin{equation}\label{eq:recipe0exma}
    \begin{cases}
        \ce{X} +\ce{W} \quad&\underset{1}{\rightarrow}\quad \ce{Y}+\ce{Z}\\
        \ce{Y} \quad&\underset{2}{\rightarrow} \quad\ce{Z}\\
         \ce{Z} \quad&\underset{3}{\rightarrow} \quad\ce{X}\\
         \ce{X} +\ce{E}\quad&\underset{4}{\rightarrow}\quad I\\
         \ce{I}\quad&\underset{5}{\rightarrow}\quad \ce{X} +\ce{E}\\
         \ce{I}\quad&\underset{6}{\rightarrow}\quad\ce{E}\\
         \emptyset &\underset{7}{\rightarrow}\quad \ce{W}
    \end{cases},
\end{equation}
where the Michaelis-Menten outflow 
\begin{equation}
    \ce{X} \quad \rightarrow \quad
\end{equation}
has been substituted by
\begin{equation}
    \ce{X} + \ce{E} \quad \rightleftharpoons \quad \ce{I} \quad \rightarrow \quad \ce{E}.
    \end{equation}
Following identical intuition as for system \eqref{eq:recipe0eqmm}, we can spot periodic oscillations in the associated mass action system:
\begin{equation}\label{eq:recipe0eqma}
\begin{cases}
\dot{x}=-k_1xw+k_3z-k_4xe\\
\dot{y}=k_1xw-k_2y\\
\dot{z}=k_1xw+k_2y-k_3z\\
\dot{w}=F_7-k_1xw\\
\dot{e}=-k_4xe+k_5i+k_6i\\
i'=k_4xe-k_5i-k_6i,
\end{cases}
\end{equation}
for parameter values: 
\begin{equation}
\begin{cases}
    k_1=\dfrac{1}{\bar{w}}\\
    k_2=k_5=k_6=F_7=1\\
    k_3=k_4=2
 \end{cases},
\end{equation}
see Fig.~\ref{fig:recipe0MA}. We stress that Michaelis--Menten represents a limit case of the mass action scheme. As a consequence, it is not a priori guaranteed that periodic orbits observed in the Michaelis-Menten system will persist for a large parameter area, as it is the case in the presented example: this must be checked case-by-case.

%Sec.~\ref{appendixB} in the Appendix Part \hyperlink{part3}{III} collects further structural enlargements (`\emph{decorations}') to the essential construction here presented, which allows for periodic oscillations under mass-action kinetics.

\begin{figure}
    \centering
\includegraphics[width=0.75\linewidth]{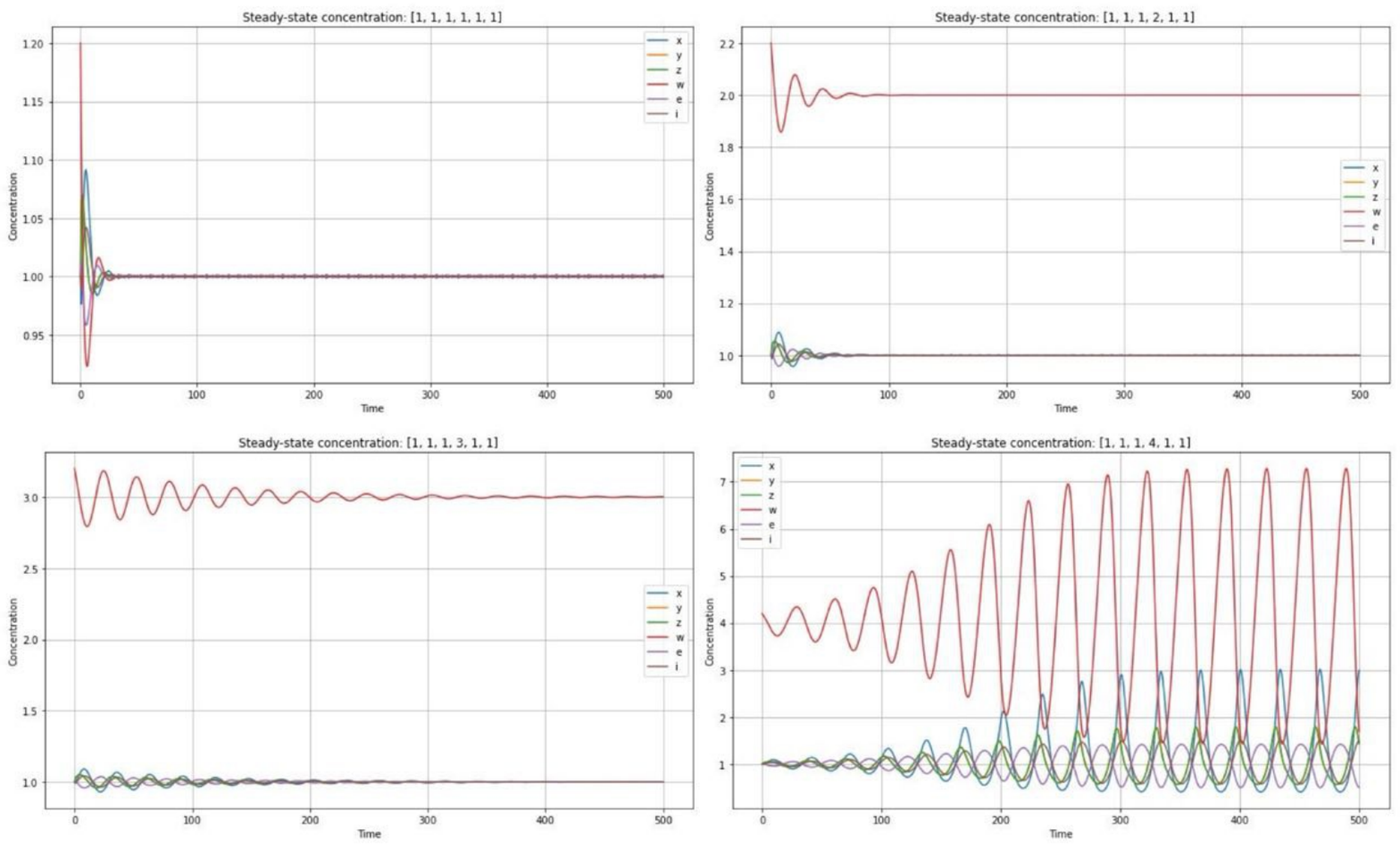}
    \caption{Numerical simulations for the system \eqref{eq:recipe0eqma}, with parameter values: $k_1=\dfrac{1}{\bar{w}};$ $
    k_2=k_5=k_6=F_7=1;$
    $k_3=k_4=2$. This choice of parameters fixes the steady-state value at $(\bar{x},\bar{y},\bar{z},\bar{w},\bar{e},\bar{i})=(1,1,1,\bar{w},1,1)$ for any choice of $\bar{w}$.     
    The initial conditions are set accordingly near the steady-state value at $(x_0,y_0,z_0,w_0,e_0,i_0)=(\bar{x},\bar{y},\bar{z},\bar{w}+0.2,\bar{e},\bar{i})=(1,1,1,\bar{w}+0.2,1,1)$. The figure shows the plots for four different values of $\bar{w}$. At $\bar{w}=1$ the steady-state is stable and strongly attracts the dynamics. At $\bar{w}=\{2,3\}$ the stability gradually decreases and damped oscillations can be seen. At $\bar{w}=4$ the steady-state is already unstable and a stable periodic orbit has appeared.}
    \label{fig:recipe0MA}
\end{figure}

\subsubsection{The same construction for other autocatalytic cores}\addtocontents{toc}{\protect\setcounter{tocdepth}{-1}}

An analogous construction can be carried out identically for the autocatalytic cores of type I, III, IV, obtaining the oscillatory networks listed below. We omit the detailed analysis, as it follows identically the construction done for type II above.

\paragraph{Type I.} Based on an Autocatalytic core of type I:
\begin{equation}
    \begin{cases}
        \ce{X} \quad &\underset{1}\rightarrow \quad \ce{Y}\\
\ce{Y} \quad &\underset{2}\rightarrow \quad 2\ce{X}\\
    \end{cases}\quad \quad \mathbb{S}=\begin{pmatrix}
        -1 & 2\\
        1 & -1
    \end{pmatrix}
\end{equation}
we get the oscillatory network
\begin{equation}
\begin{cases}
\ce{X}+\ce{W} \quad &\underset{1}\rightarrow \quad \ce{Y}\\
\ce{Y} \quad &\underset{2}\rightarrow \quad 2\ce{X}\\
\ce{X}\quad &\underset{3}\rightarrow\\
 &\underset{4}\rightarrow \quad \ce{W}
\end{cases}
\end{equation}
and its stoichiometric matrix
\begin{equation}
    \mathbb{S}=\begin{pmatrix}
        -1 & 2  & -1& 0\\
        1 & -1  & 0&0\\
        -1 & 0  & 0 &1\\
    \end{pmatrix},
\end{equation}
with steady-state flux vector
\begin{equation}
    v=(c,c,c,c),\quad c\in \mathbb{R}_{>0}.
\end{equation}

\paragraph{Type III.} Based on an Autocatalytic core of type III:
\begin{equation}
\begin{cases}
\ce{X} \quad &\underset{1}\rightarrow \quad \ce{Y}+\ce{Z}\\
\ce{Y} \quad &\underset{2}\rightarrow \quad \ce{X}\\
\ce{Z}\quad &\underset{3}\rightarrow\quad\ce{X}\\
\end{cases} \quad \quad
    \mathbb{S}=\begin{pmatrix}
        -1 & 1 & 1\\
        1 & -1 & 0\\
        1 & 0 & -1
    \end{pmatrix}
\end{equation}

we get the oscillatory network
\begin{equation}
\begin{cases}
\ce{X}+\ce{W} \quad &\underset{1}\rightarrow \quad \ce{Y}+\ce{Z}\\
\ce{Y} \quad &\underset{2}\rightarrow \quad \ce{X}\\
\ce{Z}\quad &\underset{3}\rightarrow\quad\ce{X}\\
 \ce{X}\quad&\underset{4}\rightarrow \\
 &\underset{5}\rightarrow \quad \ce{W}
\end{cases}
\end{equation}
and its stoichiometric matrix
\begin{equation}
    \mathbb{S}=\begin{pmatrix}
        -1 & 1 & 1 & -1& 0\\
        1 & -1 & 0 & 0&0\\
        1 & 0 & -1 & 0&0\\
        -1 & 0 & 0 & 0 &1\\
    \end{pmatrix}.
\end{equation}
with steady-state flux vector
\begin{equation}
    v=(c,c,c,c,c),\quad c\in \mathbb{R}_{>0}.
\end{equation}

\paragraph{Type IV.} Based on an Autocatalytic core of type IV:
\begin{equation}
\begin{cases}
\ce{X} \quad &\underset{1}\rightarrow \quad \ce{Y}+\ce{Z}\\
\ce{Y} \quad &\underset{2}\rightarrow \quad \ce{X} +\ce{Z}\\
\ce{Z}\quad &\underset{3}\rightarrow\quad\ce{X}\\
 
\end{cases}
\quad \quad
    \mathbb{S}=\begin{pmatrix}
        -1 & 1 & 1 \\
        1 & -1 & 0 \\
        1 & 1 & -1 
    \end{pmatrix}
\end{equation}
we get the oscillatory network
\begin{equation}
\begin{cases}
\ce{X}+\ce{W} \quad &\underset{1}\rightarrow \quad \ce{Y}+\ce{Z}\\
\ce{Y} \quad &\underset{2}\rightarrow \quad \ce{X} +\ce{Z}\\
\ce{Z}\quad &\underset{3}\rightarrow\quad\ce{X}\\
 \ce{X}\quad&\underset{4}\rightarrow \\
 &\underset{5}\rightarrow \quad \ce{W}
\end{cases}
\end{equation}
and its stoichiometric matrix
\begin{equation}
    \mathbb{S}=\begin{pmatrix}
        -1 & 1 & 1 & -1& 0\\
        1 & -1 & 0 & 0&0\\
        1 & 1 & -1 & 0&0\\
        -1 & 0 & 0 & 0 &1\\
    \end{pmatrix}.
\end{equation}
with steady-state flux vector
\begin{equation}
    v=(c,c,2c,2c,c),\quad c\in \mathbb{R}_{>0}.
\end{equation}

\paragraph{Type V.} In contrast, the construction does not work the same way for an Autocatalytic core of type V, 
\begin{equation}
\begin{cases}
\ce{X}\quad &\underset{1}\rightarrow \quad \ce{Y}+\ce{Z}\\
\ce{Y} \quad &\underset{2}\rightarrow \quad \ce{X} +\ce{Z}\\
\ce{Z}\quad &\underset{3}\rightarrow\quad\ce{X}+\ce{Y}
\end{cases}
 \quad\quad \mathbb{S}=\begin{pmatrix}
        -1 & 1 & 1 \\
        1 & -1 & 1 \\
        1 & 1 & -1 
    \end{pmatrix}
\end{equation}
In fact, the resulting enlarged network
\begin{equation}
\begin{cases}
\ce{X}+\ce{W} \quad &\underset{1}\rightarrow \quad \ce{Y}+\ce{Z}\\
\ce{Y} \quad &\underset{2}\rightarrow \quad \ce{X} +\ce{Z}\\
\ce{Z}\quad &\underset{3}\rightarrow\quad\ce{X}+\ce{Y}\\
 \ce{X}\quad&\underset{4}\rightarrow \\
 &\underset{5}\rightarrow \quad \ce{W}
\end{cases}
\end{equation}
possesses a stoichiometric matrix 
\begin{equation}
    \mathbb{S}=\begin{pmatrix}
        -1 & 1 & 1 & -1& 0\\
        1 & -1 & 1 & 0&0\\
        1 & 1 & -1 & 0&0\\
        -1 & 0 & 0 & 0 &1\\
    \end{pmatrix}
\end{equation}
that does not admit a positive right kernel vector. In particular, the network is not consistent, as it does not admit any steady-state and thus it shall be modified further, e.g. by adding further outflow reactions. We do not pursue this here.

\subsection{Recipe 0 in Literature Oscillators}\label{sec:literature}
\addtocontents{toc}{\protect\setcounter{tocdepth}{-1}}

We analyze few examples from the literature: the classics  \emph{Brusselator} \cite{Prigogine1968} and \emph{Oregonator} \cite{Field1974}; a more recent `\emph{Groningenator}' \cite{Harmsel2023}; an \emph{Activator-Inhibitor} model \cite{Nguyen18}, and a pattern commonly found in Lotka-Volterra systems and Mathematical Epidemiology. We will show how such examples contain the ingredients of Recipe \ref{recipe:0main}. In these examples, we will proceed as follows. We will take a look at the stoichiometry and find autocatalysis in form of autocatalytic cores. As autocatalytic cores are in particular unstable cores (unstable-positive feedbacks), the presence of autocatalysis always implies that the system admits instability \cite{VasStad23}, i.e., there is a parameter choice $R_u$ such that the symbolic Jacobian matrix $G=\mathbb{S}R_u$ is Hurwitz-unstable. We will further show that the determinant of the symbolic Jacobian $\mathbb{S}R$ is nonzero for any choice of parameters in $R$, i.e. the Jacobian is invertible for any choice of parameters. Finally, mostly leveraging Prop.~\ref{prop:stability}, we will show that the system admits stability, i.e., there exists a choice of $R$ there exists a choice $R_s$ of parameters such that the symbolic Jacobian $\mathbb{S}R_s$ is Hurwitz-stable. A network that admits both stability and instability under a Jacobian that is invertible irrespective of the parameters is an instance of Recipe \ref{recipe:0main} and we will thus conclude for periodic orbits.

\subsubsection{A Brusselator with nonambiguity property}\label{sec:bruss}\addtocontents{toc}{\protect\setcounter{tocdepth}{-1}}

We here revisit the Brusselator model \cite{Prigogine1968}, where we introduce a slight modification by adding one intermediate, to make the network nonambiguous.

Under this construction, a typical Brusselator can indeed be captured by
\be
\begin{cases}
2 \ce{X} + \ce{Y} \rightleftharpoons \ce{Z} \\
\ce{Z} \rightleftharpoons 3 \ce{X} \\
\ce{X} + \ce{B}_{res} \rightleftharpoons \ce{Y} + \ce{D}_{res} \\
\ce{X} \rightleftharpoons \ce{X}_{res}
\end{cases}
\ee
where the suffix `res' denotes reservoir species, which may appear in the form of buffer species or rapidly exchange with chemostats  (e.g. phase equilibria, coupled compartments). Upon assuming their constancy, these species can be removed from the mathematical description, and in terms of internal species we obtain the internal subnetwork 
\be
\begin{cases}
2 \ce{X} + \ce{Y} \underset{1}{\overset{2}{\rightleftharpoons}} \ce{Z} \\
\ce{Z} \underset{3}{\overset{4}{\rightleftharpoons}} 3 \ce{X} \\
\ce{X} \underset{5}{\overset{6}{\rightleftharpoons}} \ce{Y} \\
\ce{X} \underset{7}{\overset{8}{\rightleftharpoons}} \emptyset
\end{cases}
\ee
which has no conservation laws. % and one `emergent cycle' sensu Polettini \& Esposito. \cite{polettini_irreversible_2014} (i.e. a cycle of the subnetwork but not the full network).
We now write the stoichiometric matrix $\mathbb{S}$,
\be
\mathbb{S} = 
\begin{pmatrix} 
-2 &  2 & 3 & -3 & -1 & 1& -1  & 1\\
-1 &  1 & 0 &  0 & 1 & -1 & 0 & 0 \\
1  & -1 &-1 &  1 & 0 & 0  & 0 & 0
\end{pmatrix},
\ee
where we spot an \emph{autocatalytic core} of type I \cite{blokhuis20} corresponding to species $(\ce{X,Z})$ and reactions $(1,3)$,
\begin{equation}
\begin{pmatrix}
    -2 & 3\\
    1 & -1
\end{pmatrix},
\end{equation}
 which readily implies that the network admits instability. We then compute the symbolic reactivity matrix $R$
\be
R = 
\begin{pmatrix} 
R_{1X} &  0 & 0 & R_{4X} & R_{5X} & 0& R_{7X}  & 0\\
R_{1Y} &  0 & 0 &  0 & 0 & R_{6Y} & 0 & 0 \\
0  & R_{2Z} &R_{3Z} &  0 & 0 & 0  & 0 & 0
\end{pmatrix}^T,
\ee
with positive parameters $R_{jm}>0$. Note that for nonambiguous networks, the reactivity matrix has a nonzero entry $R_{jm}$ if and only if $\mathbb{S}_{mj}=s_m^j<0$.
One can readily compute the symbolic Jacobian $\mathbb{S}R$.
\begin{equation}
    G=\mathbb{S}R=
    \begin{pmatrix}
        - 2R_{1X} - 3R_{4X} - R_{5X} - R_{7X} & R_{6Y} - 2R_{1Y} & 2R_{2Z} + 3R_{3Z}\\
               R_{5X} - R_{1X} & - R_{1Y} - R_{6Y} &           R_{2Z}\\
               R_{1X} + R_{4X} &         R_{1Y}&   - R_{2Z} - R_{3Z}
    \end{pmatrix}
\end{equation}
with determinant
\begin{equation}
  - R_{7X}(R_{1Y} R_{3Z} + R_{6Y}R_{2Z} + R_{6Y}R_{3Z}),
\end{equation}
which is nonzero for all choices of parameters in $R$. To conclude for periodic orbits, we just need to find values in $R$ such that $G$ is Hurwitz-stable, since the autocatalytic core already guarantees the existence of values for which $G$ is Hurwitz-unstable. To do so, we focus on the 3-CS $\pmb{\kappa}$ defined as follows:
\begin{equation}
    \pmb{\kappa}=(\{ \ce{X,Y,Z} \}, \{1,3,7\}, \{J(\ce{X})=7,J(\ce{Y})=1,J(\ce{Z})=3),
\end{equation}
which gives rise to the associated CS-matrix
\begin{equation}
\mathbb{S}[\pmb{\kappa}]=
    \begin{pmatrix}
        -1 & -2 & 3\\
        0 & -1 & 0 \\
        0 & 1 & -1
    \end{pmatrix},
\end{equation}
with eigenvalues $(-1,-1,-1)$ and thus Hurwitz-stable. We apply Prop.~\ref{prop:stability} to conclude that the network admits stability. In conclusion, the network has an invertible Jacobian for all parameter choices, admits stability and admits instability: Recipe \ref{recipe:0main} applies and we can conclude that the system admits periodic orbits.

%has one child selection and negative determinant and furthermore see that this system admits a 2-by-2 unstable submatrix.

%Alternatively, we can opt to remove one less species - either $\ce{B}$ or $\ce{D}$ - from the description, so that the effective description has neither cycles nor conservation laws. 

%\be
%\begin{cases}
%\ce{X} \rightleftharpoons \emptyset \\
%\ce{X} + \ce{B}_{res} \rightleftharpoons \ce{Y} \\
%2 \ce{X} + \ce{Y} \rightleftharpoons \ce{Z} \\
%\ce{Z} \rightleftharpoons 3 \ce{X} 
%\end{cases}
%\ee

%We can then show an example of stability from the stoichiometric matrix thus obtained

%\be
%\mathbb{S} = 
%\begin{pmatrix} 
%-1 &  0 & 0   &  0    \\
%-1 & -1 & -2  &  3   \\
%1 &  0  & -1  &  0    \\
%0 &  0  &  1  & -1    
%\end{pmatrix}
%\begin{matrix}
%\ce{B} \\ 
%\ce{X} \\ 
%\ce{Y} \\ 
%\ce{Z} 
%\end{matrix}
%\ee

%for which all four eigenvalues are $-1$. The matrix admits a type I  autocatalytic core.

\subsubsection{An Oregonator with nonambiguity property}\label{sec:oreg}\addtocontents{toc}{\protect\setcounter{tocdepth}{-1}}

In the context of a reduced model for oscillations in Belousov–Zhabotinsky reaction, a nonambiguous \emph{Oregonator} model \cite{Field1974} is given by
\be
\begin{cases}
\ce{A} + \ce{Y} \rightleftharpoons \ce{X} \\
\ce{X} + \ce{Y} \rightleftharpoons  \ce{P} \\
\ce{B} + \ce{X} \rightleftharpoons \ce{W} \\
\ce{W} \rightleftharpoons 2 \ce{X} +  \ce{Z} \\
2 \ce{X} \rightleftharpoons \ce{Q} \\
\ce{Z} \rightarrow f \ce{Y} \ \ \ (0 < f < 1).
\end{cases}
\ee
For $f$, the authors consider $f \leq 0.25$\footnote{The value $1/4$ is due to it taking 4 electrons from Cerium(III) species to reduce one bromomalonic acid and yield a single bromide ion. Formally, side reactions with collinear rates result in smaller $f$ (co-production).}. As a model in its own right, one may just leave $f$ as a variable, and we will proceed this way. To simplify the analysis of the dynamical system, we remove species $\ce{P}$ and $\ce{Q}$ since they both only appear as single-species complex in one reversible reaction and we focus on species $\ce{A},\ce{B},\ce{X}, \ce{Y}, \ce{Z}, \ce{W}$\footnote{The intermediate $\ce{W}$ is not part of the original Oregonator, but added to obtain a nonambiguous network that facilitates the presentation.}. Furthermore, reactions are here considered as irreversible, so that we obtain the system
\be
\begin{cases}
\ce{A}+\ce{Y} \underset{1}\rightarrow \ce{X} \\
\ce{X} + \ce{Y} \underset{2}\rightarrow  \emptyset \\
\ce{B}+\ce{X} \underset{3}\rightarrow \ce{W} \\
\ce{W} \underset{4}\rightarrow 2 \ce{X} +  \ce{Z} \\
2 \ce{X} \underset{5}\rightarrow \emptyset \\
\ce{Z} \underset{6}\rightarrow  f\ce{Y} 
\end{cases}
\ee
 We can then write the full stoichiometric matrix
$\mathbb{S}$:
\begin{equation}
    \mathbb{S}=\begin{pmatrix}
        -1 & 0 & 0 & 0 & 0 & 0\\
        0 & 0 & -1 & 0 & 0 & 0\\
        1 & -1 & -1 & 2 & -2 & 0\\
        -1 & -1 & 0 &0 & 0 & 1\\
         0 & 0 &  0 & 1 & 0 & -1\\
        0 & 0 &  1 & -1 & 0 & 0
    \end{pmatrix}.
\end{equation}
Focusing on species $(\ce{X,W})$ and reactions (3,4), we clearly see the presence of an autocatalytic core of type I \cite{blokhuis20}:
\begin{equation}
\begin{pmatrix}
    -1 & 2\\
    1 & -1,
\end{pmatrix}
\end{equation}
which directly implies that the system admits instability. We further write the reactivity matrix:
\begin{equation}
    R=\begin{pmatrix}
        R_{1A} & 0 & 0 & 0 & 0 & 0\\
        0 & 0 & R_{3B} & 0 & 0 & 0\\
        0 & R_{2X} & R_{3X} & 0 & R_{5X} & 0\\
        R_{1Y} & R_{2Y} & 0 &0 & 0 & 0\\
         0 & 0 &  0 & 0 & 0 & R_{6Z}\\
        0 & 0 &  0 & R_{4W} & 0 & 0
    \end{pmatrix}^T.
\end{equation}
The symbolic Jacobian $G$ is thus obtained
\begin{equation}G=\mathbb{S}R=
\begin{pmatrix}
- R_{1A} & 0 & 0 & - R_{1Y} & 0 & 0 \\
0 & - R_{3B} & - R_{3X} & 0 & 0 & 0 \\
R_{1A} & - R_{3B} & - R_{2X} - R_{3X} - 2R_{5X} & R_{1Y} - R_{2Y} & 0 & 2R_{4W} \\
- R_{1A} & 0 & - R_{2X} & - R_{1Y} - R_{2Y} & fR_{6Z} & 0 \\
0 & 0 & 0 & 0 & - R_{6Z} & R_{4W} \\
0 & R_{3B} & R_{3X} & 0 & 0 & - R_{4W}
\end{pmatrix}
\end{equation}
whose determinant is
\begin{equation}   \operatorname{det}G=2R_{1A}R_{3B}R_{4W}R_{5X}R_{2Y}R_{6Z},
\end{equation}
and thus $G$ is invertible irrespective of the parameter choices. Note moreover that the variable $f$ does not play a role in this analysis. Moreover, the fact that the determinant has only one term in its expansion implies that there is only one $6\times 6$ invertible CS-matrix $\mathbb{S}[\pmb{\kappa}]$, 
\begin{equation}\label{eq:oregSk}
\mathbb{S}[\pmb{\kappa}] = \begin{pmatrix} 
-1 &  0   &  0  &   0  &   0  & 0 \\
0  &  -1  &  0  &   0  &   0  & 0 \\
1  &  -1  &  -2 &   -1 &   0 & 2 \\
-1 &  0   &  0 &   -1  &   f & 0 \\
0  &  0   &  0  &   0  &   -1  & 1\\
0  &  1   &  0 &   0  &   0  & -1
\end{pmatrix}
\begin{matrix} 
\ce{A} \\
\ce{B} \\
\ce{X} \\
\ce{Y} \\
\ce{Z} \\
\ce{W} 
\end{matrix}
\end{equation}
associated to a 6-CS $\pmb{\kappa}$ defined on the entire set of species, whose image is the entire set of reactions (6 species for 6 reactions). The Child-Selection bijection $J$ is given by
\begin{equation}
 J:=\begin{cases}
     J(\ce{A})=1\\
     J(\ce{B})=3\\
     J(\ce{X})=5\\
     J(\ce{Y})=2\\
     J(\ce{Z})=6\\
     J(\ce{W})=4\\
 \end{cases}  ,
\end{equation}
and it determines the order of the columns in \eqref{eq:oregSk}. Since $\mathbb{S}[\pmb{\kappa}]$ is Hurwitz-stable with eigenvalues $(-1,-1,-1,-1,-1,-2)$, the network admits stability via Prop.~\ref{prop:stability}. Again, $f$ does not appear. Recipe \ref{recipe:0main} again applies and concludes for periodic orbits, irrespective of the value of $f$.

\subsubsection{Groningenator}\label{sec:groning}\addtocontents{toc}{\protect\setcounter{tocdepth}{-1}}

A small-molecule oscillator appeared in Ref. \cite{Harmsel2023}. For consistency with oscillator naming conventions we refer to it here as the Groningenator, from Groningen, the Dutch city where the oscillator finds its origin. The reactions that have been observed and modeled for the Groningenator are given in Fig. \ref{fig:Groningenator}. The reaction is interpreted as being initiated by a trigger ($\ce{T}$), after which pyridine ($\ce{X}$) can release itself autocatalytically by cleaving FMOC. The reaction is inhibited by a fast and slow inhibitor. 

\begin{figure}[tbhp!]
\centering
\includegraphics[width=0.70\linewidth]{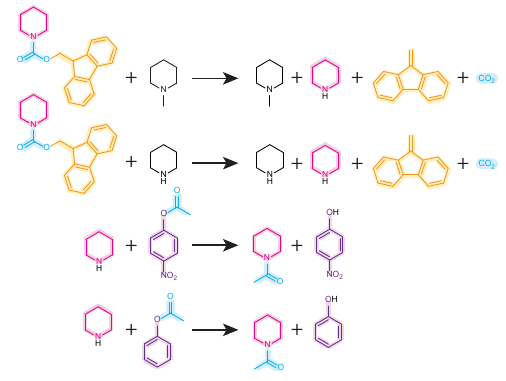}
\caption{Groningenator: chemical reactions used in minimal model in ref \cite{Harmsel2023}, consisting of triggered release, autocatalysis, fast inhibition, slow inhibition. Molecular fragments are colored for clarity.}
\label{fig:Groningenator}
\end{figure}
Mechanistically, these reactions can obviously be decomposed in further steps. For the first two reactions, adding at least one intermediate to the description makes the network nonambiguous, i.e. stoichiometry of reactants and products then coincide with the stoichiometry of overall reactions and ipso facto the stoichiometric matrix then adequately represents allo-and autocatalysis. To this end, we introduce intermediates $\ce{Y}, \ce{Z}$. to obtain
\bea
\ce{FMOC} + \ce{Trig} \rightarrow \ce{Y} \rightarrow \ce{Trig} + \ce{Pyr} + \ce{DBF} + \ce{CO2} \\
\ce{FMOC} + \ce{Pyr} \rightarrow \ce{Z} \rightarrow 2 \ce{Pyr} + \ce{DBF} + \ce{CO2} \\
\ce{Pyr} + \ce{Inh_{fast}} \rightarrow \ce{Pyr_{ac}} + \ce{Waste_{fast}} \\
\ce{Pyr} + \ce{Inh_{slow}} \rightarrow \ce{Pyr_{ac}} + \ce{Waste_{slow}} 
\eea

%\ce{X} : Pyridine
%\ce{W} : pyr-FMOC
%\ce{Y}: Int
%\ce{A}: inh-1
%\ce{B}: inh-2
%\ce{T}: Trigger
We can freely remove $\ce{DBF, CO2, Pyr_{ac}, Waste_{fast}, Waste_{slow}}$ from the description as they are formed irreversibly and not consumed, i.e. no other variables are dependent on their values. Upon removal of a sufficient number of species from the description, we are left with 
%\be
%\begin{cases}
%\ce{X} + \ce{W} \rightleftharpoons \ce{Y} \\
% \ce{Y} \rightleftharpoons 2 \ce{X} \\ 
% \ce{X} + \ce{A} \rightleftharpoons \emptyset \\
 %\ce{X} \rightleftharpoons \emptyset \\
%\ce{T} + \ce{W} \rightleftharpoons \ce{X} \\
% \ce{X} + \ce{B} \rightleftharpoons \emptyset
%\end{cases}
%\ee

\be
\begin{cases}
\ce{W} + \ce{X} \underset{1}\rightarrow \ce{Y} \\
 \ce{Y} \underset{2}\rightarrow 2 \ce{X} \\ 
 \ce{X} + \ce{A} \underset{3}\rightarrow \emptyset \\
 \ce{X} \underset{4}\rightarrow \emptyset \\
\ce{W} + \ce{T} \underset{5}\rightarrow \ce{Z} \\
\ce{Z} \underset{6}\rightarrow \ce{X} + \ce{T} \\
 \ce{X} + \ce{B} \underset{7}\rightarrow\emptyset
\end{cases}
\ee
We consider here for simplicity irreversible reactions, as in the original model. A same conclusion can be reached by adding reversible direction. 
The principal oscillating species here are $\ce{X}, \ce{Y}$ corresponding to Pyridine and an  intermediate. $\ce{A}, \ce{B}$ are inhibitors, $\ce{T}$ a trigger. $\ce{W}$ is Pyridine-Fmoc which serves as a supply from which $\ce{X}$ can be liberated.

We obtain as stoichiometric matrix
\be
\mathbb{S} = \begin{pmatrix}
-1 & 0 & 0 & 0 & -1 & 0 & 0 \\
-1 & 2 &  -1 & -1 &  0 & 1 & -1 \\
 1 & -1 & 0 & 0 & 0 & 0 & 0 \\
0 & 0 & -1 & 0 & 0 & 0 & 0 \\
0 & 0 & 0 & 0 & -1 & 1 & 0 \\
0 & 0 & 0 & 0 & 1 & -1 & 0 \\
0 & 0 & 0 & 0 & 0 & 0 & -1 
\end{pmatrix} \begin{matrix} \ce{W} \\ \ce{X} \\ \ce{Y} \\ \ce{A} \\ \ce{T} \\ \ce{Z} \\ \ce{B}
\end{matrix}.
\ee
We directly see a conserved quantity $L$ composed of $\ce{T}, \ce{Z}$:
\begin{equation}
    \mathcal{L} = x_{\ce{T}} + x_{\ce{Z}},
\end{equation}
corresponding to the conservation law identified by the left kernel vector $w=(0,0,0,0,1,1,0)$. In particular, the Jacobian determinant is identically zero. The mechanism of this third example shares again strong similarities with the both Brusselator and Oregonator analyzed in Sec.~\ref{sec:bruss} and Sec.~\ref{sec:oreg}, respectively. As in both the previous cases we had explicitly performed the computation, we argue here more implicitly. %Indeed, as in the Oregonator, 
To apply Recipe \ref{recipe:0main}, we need to identify a principal submatrix of the symbolic Jacobian $G$, invertible for all choices. Due to the conserved quantity $\mathcal{L}$ involving the trigger, an educated guess consists into looking into the $6\times 6$ principal submatrix $G[\kappa]$ where $\kappa=\{\ce{W,X,Y,A,Z,B}\}$, with the trigger $T$ removed. The corresponding stoichiometry reads:
\be
\mathbb{S}_{\vee T} = \begin{pmatrix}
-1 & 0 & 0 & 0 & -1 & 0 & 0 \\
-1 & 2 &  -1 & -1 &  0 & 1 & -1 \\
 1 & -1 & 0 & 0 & 0 & 0 & 0 \\
0 & 0 & -1 & 0 & 0 & 0 & 0 \\
0 & 0 & 0 & 0 & 1 & -1 & 0 \\
0 & 0 & 0 & 0 & 0 & 0 & -1 
\end{pmatrix} \begin{matrix} \ce{W} \\ \ce{X} \\ \ce{Y} \\ \ce{A} \\\ce{Z} \\ \ce{B}
\end{matrix}.
\ee
We note that this stoichiometry consists of six species for seven reactions. However, we note that the rows corresponding to species $B,Z,A,Y$ contains only one negative entry, i.e. these species participate as reactant in one reaction each, respectively $7,6,3,2$. This directly implies that 
\emph{any} Child-Selection bijection $J$ assigns univocally to $(\ce{Y,A,Z,B})$ reactions $(2,3,6,7)$. In turn, by focusing on the remaining species $(\ce{W,X})$ we conclude that there are only three 6-Child-Selections $\pmb{\kappa}_1,\pmb{\kappa}_2,\pmb{\kappa}_3$ defined on $\kappa=(\ce{W,X,Y,A,Z,B})$, namely:
\begin{equation}
 \pmb{\kappa}_1:=\{J_1(\kappa)=(1,4,2,3,6,7)\},
\end{equation}
with associated CS-matrix
\be
\mathbb{S}[\pmb{\kappa}_1] = \begin{pmatrix}
-1 &0 & 0 & 0  & 0 & 0 \\
-1 & -1 &2 &  -1  & 1 & -1 \\
 1 &0 & -1 & 0 & 0 & 0 \\
0 & 0& 0 & -1  & 0 & 0 \\
0 & 0& 0 & 0 & -1 & 0 \\
0 & 0 & 0 & 0  & 0 & -1 
\end{pmatrix} \begin{matrix} \ce{W} \\ \ce{X} \\ \ce{Y} \\ \ce{A} \\\ce{Z} \\ \ce{B}
\end{matrix},
\ee
which is Hurwitz-stable with eigenvalues $(-1,-1,-1,-1,-1,-1).$
;
\begin{equation}
 \pmb{\kappa}_2:=\{J_2(\kappa)=(5,4,2,3,6,7)\},
\end{equation}
with associated CS-matrix
\be
\mathbb{S}[\pmb{\kappa}_2] = \begin{pmatrix}
-1 &0 & 0 & 0  & 0 & 0 \\
0 & -1 &2 &  -1  & 1 & -1 \\
0 &0 & -1 & 0 & 0 & 0 \\
0 & 0& 0 & -1  & 0 & 0 \\
1 & 0& 0 & 0 & -1 & 0 \\
0 & 0 & 0 & 0  & 0 & -1 
\end{pmatrix} \begin{matrix} \ce{W} \\ \ce{X} \\ \ce{Y} \\ \ce{A} \\\ce{Z} \\ \ce{B}
\end{matrix},
\ee
which is Hurwitz-stable with eigenvalues $(-1,-1,-1,-1,-1,-1).$
;
\begin{equation}
 \pmb{\kappa}_3:=\{J_3(\kappa)=(5,1,2,3,6,7)\},
\end{equation}
with associated CS-matrix
\be
\mathbb{S}[\pmb{\kappa}_2] = \begin{pmatrix}
-1 &-1 & 0 & 0  & 0 & 0 \\
0 & -1 &2 &  -1  & 1 & -1 \\
0 &1 & -1 & 0 & 0 & 0 \\
0 & 0& 0 & -1  & 0 & 0 \\
1 & 0& 0 & 0 & -1 & 0 \\
0 & 0 & 0 & 0  & 0 & -1 
\end{pmatrix} \begin{matrix} \ce{W} \\ \ce{X} \\ \ce{Y} \\ \ce{A} \\\ce{Z} \\ \ce{B}
\end{matrix},
\ee
which is a singular matrix.

By this simple direct analysis we conclude three facts:
\begin{enumerate}
\item The determinant of the principal submatrix $G[\kappa]$ is
  non-zero %invertible
  for all choices of parameters:
\begin{equation}\begin{split}
    \operatorname{det}G[\kappa]=\sum_\mathbf{\pmb{\kappa}}&%=
    \operatorname{det}\mathbb{S}[\pmb{\kappa}_1]R[\pmb{\kappa}_1]+\operatorname{det}\mathbb{S}[\pmb{\kappa}_2]R[\pmb{\kappa}_2]\\
    &=R_{1W}R_{4X}R_{2Y}R_{3A}R_{6Z}R_{7B}+R_{5W}R_{4X}R_{2Y}R_{3A}R_{6Z}R_{7B}
    >0;
\end{split}\end{equation}
\item There exists a choice $R_u$ for which $G[\kappa]$ is Hurwitz-unstable, due to the presence of an autocatalytic core of type I in $\mathbb{S}_{\vee T}$ for species $(X,Y)$ and reactions $(1,2)$:
$$\begin{pmatrix}
    -1 & 2\\
    1 & -1
\end{pmatrix};$$
\item There exists a choice $R_s$ for which $G[\kappa]$ is Hurwitz-stable, via Prop.~\ref{prop:stability},
  since both $\mathbb{S}[\pmb{\kappa}_1]$ and $\mathbb{S}[\pmb{\kappa}_2]$ are Hurwitz-stable. 
\end{enumerate}
Recipe \ref{recipe:0main} concludes then for periodic orbits in the associated dynamical system. 

\subsubsection{Activator-Inhibitor Model}\label{sec:activatorinhibitor}\addtocontents{toc}{\protect\setcounter{tocdepth}{-1}}
The activator-inhibitor model \cite{Nguyen18} represents a simplified CRN model to study regulatory motifs commonly found in biological oscillators. 
The model consists of a substrate $\ce{S}$, an activator $\ce{R}$, an inhibitor $\ce{X}$, phosphorylated enzyme $\ce{M_p}$, $\ce{ATP}$ (referred as $A_t$), $\ce{ADP}$ (referred as $A_d$), and inorganic phosphate $\ce{P}_i$. We have added intermediates $\ce{I_1},\ce{I_2},\ce{I_3}$ to have a nonambiguous model. The reactions read as follows
\begin{equation}
    \begin{cases}
    \ce{M_p}+\ce{S} \quad  &\overset{1}{\underset{2 }\rightleftharpoons} \quad \ce{I_1} \quad \overset{3}{\underset{4}\rightleftharpoons} \quad \ce{M_p} + \ce{R}\\
        \ce{S}\quad & \overset{5}{\underset{6}\rightleftharpoons} \quad \ce{R}\\
        \ce{X}+\ce{R}+\ce{A_T}\quad &\overset{7}{\underset{8}\rightleftharpoons}\quad \ce{I_2}\quad \overset{9}{\underset{10}\rightleftharpoons}\quad\ce{X}+\ce{S}+\ce{A_D}+\ce{P_i}\\
          \ce{M_p}+\ce{S} \quad & \overset{11}{\underset{12}\rightleftharpoons} \quad \ce{I_3} \quad \overset{13}{\underset{14}\rightleftharpoons} \quad \ce{M_p} + \ce{X}\\
\ce{X}+\ce{A_T}\quad&\overset{15}{\underset{16}\rightleftharpoons}\quad\ce{S}+\ce{A_D}+\ce{P_i}\\
 \end{cases},
\end{equation}

\begin{figure}
    \centering
    \includegraphics[width=0.75\linewidth]{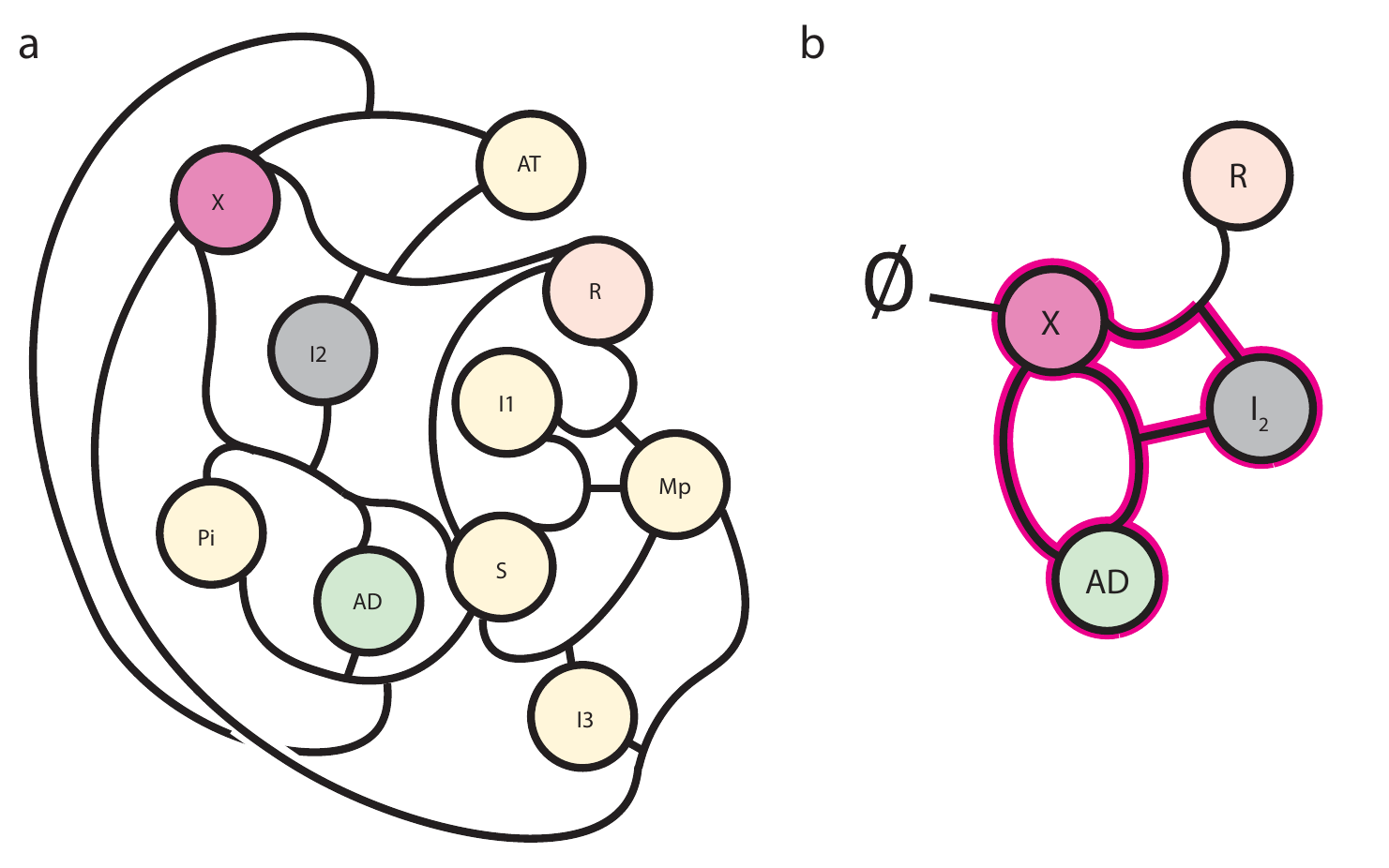}
    \caption{a) CRN for the activator-inhibitor model \cite{Nguyen18}, with added intermediates. b) subnetwork motif from which oscillations can be deduced via Recipe \ref{recipe:0main}.}
    \label{fig:actinh}
\end{figure}

For sake of example, we spot a possible motif for oscillatory behavior by focusing on the stoichiometry arising from species $\ce{R},\ce{X},\ce{A_d},\ce{I_2}$ and associated reactions $7,14,9,16$, obtaining the following stoichiometric matrix:
\begin{equation}
\mathbb{S}=\begin{pmatrix}
    -1 & 0 & 0 & 0\\
    -1 & -1 & 1 & 1\\
     1 & 0 & -1 & 0\\
     0 & 0 &  1 & -1
\end{pmatrix}
\begin{matrix}
    \ce{R}\\
    \ce{X}\\
    \ce{I_2}\\
    \ce{A_d}
\end{matrix},
\end{equation}
which is Hurwitz-stable with eigenvalues $(-1,-1,-1,-1)$. Since $\mathbb{S}$ possesses negative diagonal, we may  indeed interpret it as a CS-matrix $\mathbb{S}[\pmb{\kappa}]$, for the 4-CS
\begin{equation}
    \pmb{\kappa}:=\{\{\ce{R},\ce{X},\ce{A_d},\ce{I_2}\},\{7,9,14,16\},\{J(\ce{R},\ce{X},\ce{A_d},\ce{I_2})=(7,14,9,16)\}\}.
\end{equation} 
Moreover, no other 4-CS-bijection can be defined for such sets of species and reactions. However, besides the 3-principal minors of $\mathbb{S},$ we observe that species $(\ce{X},\ce{I_2},\ce{A_d})$ and reactions $(7,9,16)$ identifies an autocatalytic core of type II:
\begin{equation}\label{eq:ac_in_autocat}
   \mathbb{S}[\pmb{\kappa}']= \begin{pmatrix}
        -1 & 1 & 1\\
        1 & -1 & 0\\
        0 & 1 & -1
    \end{pmatrix},
\end{equation}
associated to the 3-CS
\begin{equation}
\pmb{\kappa}':=\{\{\ce{X},\ce{I_2},\ce{A_d})\}, \{7,9,16\},J'(\ce{X},\ce{I_2},\ce{A_d})=(7,9,16)\}.
\end{equation}
The above two observations guarantee the applicability of Recipe \ref{recipe:0main}, and thus the insurgence of periodic behavior. Note that since the autocatalytic core \eqref{eq:ac_in_autocat} does not appear as principal submatrix of the CS-matrix $\mathbb{S}[\pmb{\kappa}]$ such scheme does not follow Recipe \ref{recipe:1main}, albeit it shares with it a certain degree of chemical intuition.

To make our observations explicit, we consider the reactivity matrix associated to the stoichiometry of $\mathbb{S}$:
\begin{equation}
    R=\begin{pmatrix}
        R_{7R} & 0 & 0 & 0\\
        R_{7X} & R_{14X} & 0& 0\\
        0 & 0 & R_{9I_2}& 0\\
        0 & 0 & 0 & R_{16A_d}
    \end{pmatrix}^T,
\end{equation}
and the product matrix $G=\mathbb{S}R$
\begin{equation}
\begin{pmatrix}
-R_{7R} & -  R_{7X} & 0 & 0\\
-R_{7R} & -  R_{7X} -R_{14X} & R_{9I_2} & R_{16A_d} \\
R_{7R} &   R_{7X} &-R_{9I_2} & 0\\
0 & 0 & R_{9I_2} & -R_{16A_d}
\end{pmatrix}
\end{equation}
which is just the principal minor of the Jacobian matrix of the full system, associated to species $(\ce{R},\ce{X},\ce{I}_2,\ce{A}_d)$ and with any partial derivative $R_{jm}$ with $j\neq 7,14,9,16$ chosen to be at the limit $R_{jm}=0.$

By computing the determinant of $G$, 
\begin{equation}
\operatorname{det}G=R_{7R}R_{14X}R_{9I_2}R_{16A_d}
\end{equation}
we indeed confirm that $G$ is always invertible. On the other hand, the principal minor:
\begin{equation}
\operatorname{det}\begin{pmatrix}
    -  R_{7X} -R_{14X} & R_{9I_2} & R_{16A_d} \\
   R_{7X} &-R_{9I_2} & 0\\
 0 & R_{9I_2} & -R_{16A_d}
\end{pmatrix}=(R_{7X}-R_{14X})R_{9I_2}R_{16A_d}
\end{equation}
is of sign $(-1)^{3-1}$ for $R_{7X}>R_{14X}$, namely when autocatalysis becomes dominant, and thus the system admits instability. Recipe \ref{recipe:0main} applies and implies the occurrence of periodic orbits. 

Similar constructions can be also obtained, for instance, starting from species 
$(\ce{R},\ce{X},\ce{I_2},\ce{A_t}),$ 
$(\ce{R},\ce{X},\ce{I_2},\ce{P_i}),$
$(\ce{S},\ce{X},\ce{I_2},\ce{A_d}),$
$(\ce{S},\ce{X},\ce{I_2},\ce{A_t}),$
$(\ce{S},\ce{X},\ce{I_2},\ce{P_i}),$
whose analysis we omit here for brevity and for close analogy to the treated choice.

\subsubsection{Autocatalytic reaction, Mathematical Epidemiology, and Lotka-Volterra}\label{ex:ME}\addtocontents{toc}{\protect\setcounter{tocdepth}{-1}}
A particular simple autocatalytic scheme always give rise to oscillations under parameter-rich kinetics, following Recipe \ref{recipe:0main}. 
\begin{equation}
    \begin{cases}
        \ce{X} + \ce{Y} \quad &\underset{1}{\rightarrow} \quad \ce{2Y}\\
        \ce{Y} \quad &\underset{2}{\rightarrow} \quad ...
    \end{cases}
\end{equation}

The dots `...' in reaction 2 indicates any combination of species, possibly empty, that does not contain neither $\ce{X}$ nor $\ce{Y}$. Recipe \ref{recipe:0main} implies the following corollary.
\begin{cor}\label{cor:ME}
If a network possesses two reactions of the form $1$ and $2$, then
the associated system admits nonstationary periodic solutions
\end{cor}
\proof
Consider the following rescale for the reactivity matrix $R$: Any entry $R_{jm}\neq R_{1X},R_{1Y},R_{2Y}$ is of order $\varepsilon$. Then consider the principal submatrix of the Jacobian $G=\mathbb{S}R$ associated to species $\ce{X,Y}$. We get
$$G[\{X,Y\}]=\begin{pmatrix}
    -R_{1X} +O(\varepsilon) & -R_{1Y} +O(\varepsilon) \\
    R_{1X}+ O(\varepsilon) & R_{1Y}-R_{2Y}+O(\varepsilon)
\end{pmatrix}$$
At $\varepsilon=0$ the determinant of the matrix is $\operatorname{det}R_{1X}R_{2Y}\neq 0$, for all choices of $R_{1X}, R_{1Y}, R_{2Y}$. On the other hand, fix $R_{1X}=R_{2Y}=1$. For $R_{1Y}$ small enough, $G[\{X,Y\}]$ is Hurwitz-stable. For $R_{1Y}$ big enough, $G[\{X,Y\}]$ is Hurwitz-unstable. Via the  same continuity argument as in proofs of Thm.~\ref{thm:mainmain} and Recipe \ref{recipe:0main}, the system admits nonstationary periodic.
\endproof

A weaker version of Corollary \ref{cor:ME} was proved independently in
\cite{VAA24ME} in the context of Mathematical Epidemiology. It is worth
noting that such a pattern is omnipresent %indeed in epidemics models,
where $X$ typically represents the susceptible individuals $\mathsf{S}$,
and $Y$ the infected ones $\mathsf{I}$. The autocatalytic reaction 1 then
represents the basic infection-interaction where one infected individual
meets a susceptible individual: The encounter results in two infected
individuals. We underline, however, that the same paper \cite{VAA24ME}
shows that the statement does not typically hold if the network is endowed
with mass action kinetics instead than with a parameter-rich kinetics. For
example, for the standard Lotka--Volterra system \be
\begin{cases}
\ce{X} + \ce{Y} \underset{1}\rightarrow 2 \ce{Y},\\
\ce{Y} \underset{2}\rightarrow \emptyset\\
\ce{X} \underset{3}\rightarrow 2 \ce{X}, 
\end{cases}
\ee contains reactions of the form 1 and 2 as in Cor.~\ref{cor:ME}, but
there is no choice of parameters for which the steady-state Jacobian is
Hurwitz-stable or Hurwitz-unstable. In this case it is well-known
\cite{Murray:2007} that the
spectrum of the Jacobian is purely imaginary for all choices of the
parameters and the Lotka-Volterra system possesses a continuum of periodic
solutions: Any choice of nonstationary initial condition lies on a periodic
trajectory, and all the periodic trajectories are non-hyperbolic. One way
to look at it is that the associated dynamical system can be written as
\begin{equation}
    \begin{cases}
    \dot{x}=k_3x-k_1xy=x(k_3-k_1y)\\
    \dot{y}=k_1xy-k_2y=y(k_1x-k_2),
    \end{cases}
\end{equation}
and thus steady-state constraints automatically imply an identically zero  trace for the steady-state Jacobian matrix, which points to purely-imaginary eigenvalues. Clearly, by considering a different nonlinearity as Michaelis-Menten, such feature does not necessarily happen and Cor.~\ref{cor:ME} applies. 

It is also worth noting that upon addition of one single intermediate in reaction 1,
\begin{equation}
    \ce{X} + \ce{Y} \underset{1'}\rightarrow \ce{Z} \underset{1''}\rightarrow 2 \ce{Y}
\end{equation}
we get that the steady-state Jacobian of the associated mass action system becomes hyperbolic, and thus the continuum of periodic orbits is lost if we consider such simple enlargement. We stress that this is not in contradiction with Banaji's inheritance results \cite{ba23splitting}, which are valid for hyperbolic steady-states and periodic orbits, only. 

We include here an explicit analysis of this fact:
\begin{equation}
    \begin{cases}
    \dot{x}=k_3x-k_{1'}xy\\
    \dot{y}=-k_{1'}xy+2k_{1''}z-k_2y\\
    \dot{z}=k_{1'}xy-k_{1''}z
    \end{cases}
\end{equation}
The positive steady-state parametrization reads:
\begin{equation}
    \bar{x}=\dfrac{k_2}{k_{1'}},
    \quad \bar{y}=\dfrac{k_3}{k_{1'}},\quad \bar{z}=\dfrac{k_2k_3}{k_{1'}k_{1''}},
\end{equation}
which leads to the Jacobian matrix, evaluated at $(\bar{x},\bar{y},\bar{z})$:

\begin{equation}
\begin{split}
    G(\bar{x},\bar{y},\bar{z})&=\begin{pmatrix}
        k_3-k_{1'}y & -k_{1
    '}x & 0\\
    -k_{1'}y & -k_2 -k_{1'}x& 2k_{1''}\\
    k_{1'}y & k_{1'}x & -k_{1''}
\end{pmatrix}\Bigg|_{(x,y,z)=(\bar{x},\bar{y},\bar{z})}\\
&=
\begin{pmatrix}
        0 & -k_{2} & 0\\
    -k_3 & -2k_2 & 2k_{1''}\\
    k_3 & k_2 & -k_{1''}
\end{pmatrix}=\begin{pmatrix}
        0 & -1 & 0\\
    -1 & -2 & 2\\
    1 & 1 & -1
\end{pmatrix}\begin{pmatrix}
 k_3 & 0 & 0\\
    0 & k_2 & 0\\
    0 & 0 & k_1
\end{pmatrix}.
\end{split}
\end{equation}
Therefore, the $D$-hyperbolicity of $$\begin{pmatrix}
        0 & -1 & 0\\
    -1 & -2 & 2\\
    1 & 1 & -1
\end{pmatrix},$$
is equivalent to the dynamical hyperbolicity of $\bar{x},\bar{y},\bar{z}$ for all choices of reaction rates. To do so, we exclude purely-imaginary eigenvalues of 
$$
\begin{pmatrix}
        0 & -k_{2} & 0\\
    -k_3 & -2k_2 & 2k_{1''}\\
    k_3 & k_2 & -k_{1''}.
\end{pmatrix}
$$
Computing $$\operatorname{det}G(\bar{x},\bar{y},\bar{z})=-k_{1''}k_2k_3\neq0$$
excludes real zero eigenvalues for all choices of parameters. Indirectly assume that $G(\bar{x},\bar{y},\bar{z})$ has purely-imaginary eigenvalues. This would imply the existence of real $\omega\in \mathbb{R}$ such that $\omega i$ is a root of the characteristic polynomial:
\begin{equation}
 g(\omega i):=\operatorname{det}(G(\bar{x},\bar{y},\bar{z})-\omega i\:Id)=
 \omega^3i+(2k_2 + k_{1''})\omega^2 + k_2k_3\omega i - k_2k_3k_{1''}.
\end{equation}
However, since real part of the characteristic polynomial
$$\Re (g(\omega i))=(2k_2 + k_{1''})\omega^2-k_2k_3k_{1''}$$
and imaginary part 
$$\Im (g(\omega i))=\omega^3+k_2k_3\omega$$
need both to be zero, the \emph{resultant} of such two polynomials in $\omega$ must be zero, as well, since they share a root. However:
\begin{equation}
\operatorname{resultant}(\Re (g(\omega i)),\Im (g(\omega i))= 
4k_2^4k_3^2k_{1''}+ k_2^3k_3^3k_{1''}^3 + 4k_2^3k_3^2k_{1''}^2 + k_2^2k_3^2k_{1''}^3 \neq 0,
\end{equation}
for all choices of parameters. Thus purely imaginary eigenvalues of the Jacobian $G(\bar{x},\bar{y},\bar{z})$ are excluded.

\section{Hidden explicit catalysis may trigger oscillations}\label{ex:hidcat}
\addtocontents{toc}{\protect\setcounter{tocdepth}{0}}

We consider the following stoichiometric matrix:
\begin{equation}
\mathbb{S}=
\begin{pmatrix}
-1 & 0 & 0 & 0 & 1 \\
1 & -1 & 0 & 0 & 0 \\
0 & 1 & -1  & 0 & 0\\
0 & 0 & 1& -1 & 0\\
0 & 0 & 0  & 1 & -1\\
\end{pmatrix}.
\end{equation}

If the network is assumed to be nonambiguous, that is, if we assume absence of explicit catalysts in each reaction, such stoichiometry corresponds to a network consisting of a simple monomolecular cycle.

\begin{equation}\label{eq:exprimum}
\ce{X_1} \quad\underset{1}{\rightarrow}\quad \ce{X_2} 
\quad\underset{2}{\rightarrow}\quad \ce{X_3}
\quad\underset{3}{\rightarrow}\quad \ce{X_4}
\quad\underset{4}{\rightarrow}\quad \ce{X_5}
\quad\underset{5}{\rightarrow}\quad \ce{X_1}
\end{equation}

It is easy to see that there is no oscillatory core, and more in general that the associated dynamical system does not admit any Hopf bifurcation. As a side comment, in the language of Feinberg's theory \cite{Fei19}, the network \eqref{eq:exprimum} has \emph{deficiency zero}, and it is weakly reversible. This implies that the associated mass-action system admits a unique positive complex-balanced steady-state in each stoichiometric compatibility class. Such steady-state is asymptotically stable, due to the presence of a global Lyapunov function in the positive orthant.

In this section, we show that - without the nonambiguity assumption -  it is possible to design a network with a net (!) stoichiometry as $\mathbb{S}$ above, which admits oscillations. To add explicit catalysis in a proper way, we assume that the rate of each reaction $j$ also depends on the concentration of the species $\ce{X_{j-1}}$.  This way the associated ODE system \eqref{eq:dynamics}  reads:
\begin{equation}\label{eq:explicit}
\begin{cases}
\dot{x}_1=-r_1(x_1,x_5)+r_5(x_4,x_5)\\
\dot{x}_2=-r_2(x_1,x_2)+r_1(x_1,x_5)\\
\dot{x}_3=-r_3(x_2,x_3)+r_2(x_1,x_2)\\
\dot{x}_4=-r_4(x_3,x_4)+r_3(x_2,x_3)\\
\dot{x}_5=-r_5(x_4,x_5)+r_4(x_3,x_4)\\
\end{cases},
\end{equation}
or -- in general --
\begin{equation}\label{eq:monotonecyclic}
    \dot{x}_i=-r_i(x_i, x_{i-1}) + r_{i-1}(x_{i-2},x_{i-1}).
\end{equation}
Note that the stoichiometric matrix of system \eqref{eq:explicit} coincides
with $\mathbb{S}$ above. Note, that this model violates the nonambiguity assumption, since the reactivities here depend on the catalysts
  that are not represented in $\mathbb{S}$, i.e., the reactivity matrix $R$ has non-zero entries that are 
  \emph{not implied} by the stoichiometric matrix $\mathbb{S}$. The steady-state constraints are
particularly easy:
\begin{equation}\label{eq:eqconstraints}
    r_i(x)=r_{i-1}(x)
\end{equation}
which means that the unique steady-state flux vector $v$ for the stoichiometric matrix $S$
is the right kernel vector $$v=(c,c,c,c,c,c)^T,$$
with $c>0$. The vector $v$ identifies also the unique conserved quantity, i.e. the left kernel vector. This implies 
$$x_1(t)+x_2(t)+x_3(t)+x_4(t)+x_5(t)=C.$$

For simplicity of presentation, let us consider explicit polynomial kinetics of the following form:
$$r_i=k_i x_i x_{i-1}^n, \quad \text{with $n\ge2$},$$
which gives rise to the system:
\begin{equation}\label{eq:explicitsystem}
\begin{cases}
\dot{x}_1=-k_1x_1x_5^n+k_5x_4^nx_5\\
\dot{x}_2=-k_2x_1^nx_2+k_1x_1x_5^n\\
\dot{x}_3=-k_3x_2^nx_3+k_2x_1^nx_2\\
\dot{x}_4=-k_4x_3^nx_4+k_3x_2^nx_3\\
\dot{x}_5=-k_5x_4^nx_5+k_4x_3^nx_4\\
\end{cases}\end{equation}
Such kinetics can be seen as generalized mass-action for the network with reactions 
$$\ce{X}_i + \ce{X}_{i-1} \quad \rightarrow \quad \ce{X}_{i+1} + \ce{X}_{i-1},$$
or as standard mass-action for reactions
$$\ce{X}_i + n \ce{X}_{i-1} \quad \rightarrow \quad \ce{X}_{i+1} + n \ce{X}_{i-1}.$$

In any case, we write the Jacobian of such systems evaluated at steady-states. To do so, we follow a standard parametrization procedure discussed in detail in Clarke's \emph{Stoichiometric Network Analysis} \cite{ClarkeSNA}. The Jacobian $G(\bar{x})$ evaluated at a steady-state $\bar{x}$ reads:
\begin{equation}\label{eq:Jac}
G(\bar{x})=\mathbb{S}\operatorname{diag}(c,c,c,c,c,c) K^T\operatorname{diag}(1/{\bar{x}_1},...,1/{\bar{x}_5}),
\end{equation}
where $K$ is the so-called $5\times 5$ \emph{kinetic matrix} defined in this case as:
$$K_{im}=\begin{cases}
n\quad\text{if $m={i-1};$}\\
    1\quad\text{if $m=i$;}\\
    0\quad\text{otherwise.}
\end{cases}.$$
Fix $c=1$, without loss of generalities. Calling $D:=\operatorname{diag}(1/\bar{x}_1,...,1/\bar{x}_5)$ we get:
\begin{equation}\label{eq:Jac2}
\begin{split}
G(\bar{x})&=\mathbb{S}K^T D\\
&=\begin{pmatrix}
    -1 & 0 & 0 & 0 & 1\\
    1 & -1 & 0 & 0 & 0\\
    0 & 1 & -1 & 0 & 0\\
    0 & 0 & 1 & -1 & 0\\
    0 & 0 & 0 & 1 & -1\\
\end{pmatrix} \begin{pmatrix}
1 & 0 & 0 & 0 & n\\
n & 1 & 0 & 0 & 0\\
0 & n & 1 & 0 & 0\\
0 & 0 & n & 1 & 0\\
0 & 0 & 0 & n & 0\\
\end{pmatrix} D\\
&= \begin{pmatrix}
    -1 & 0 & 0 & n & 1-n\\
    1-n & -1 & 0 & 0 & n\\
    n & 1-n & -1 & 0 & 0\\
    0 & n & 1-n & -1 & 0\\
    0 & 0 & n & 1-n & -1\\
\end{pmatrix}D\\
&=CD.
\end{split}
\end{equation}
We can think at the values of the positive steady-state itself as parameters. Thus $D$ can be thought as a positive diagonal matrix with parametric entries $1/\bar{x}$, directly bridging to $D$-stability concepts and Sec. \ref{sec:linearalgebra}. Detailed analysis of this type of procedure for mass action systems can be found in \cite{Vassena2025}. We further note that $C$ is also a $P^-_0$ matrix. However, due to the presence of a conserved quantity, $\operatorname{det}C=0$ and thus $C$ is not a $P^-_{FF}$ matrix. The Fisher-Fuller Theorem, \ref{thm:FF}, does not readily apply. However, we can consider the $4\times4$ principal submatrix:
$$C[1:4]=\begin{pmatrix}
    -1 & 0 & 0 & n\\
    1-n & -1 & 0 & 0\\
    n & 1-n & -1 & 0\\
    0 & n & 1-n & -1\\
\end{pmatrix}\text{ with $\operatorname{det}C[1:4]=n^4- n^3 + n^2 - 1$}.$$
It can be shown that for $n\ge2$ $C[1:4]$ is unstable. As this $C[1:4]$ is
further a $P^-$ matrix, the Fisher-Fuller Theorem applies and we can
conclude that $G$ is a $D$-Hopf matrix, and thus Thm.~\ref{thm:mainmain}
guarantees that the system admits nonstationary periodic orbits. See
Fig.~\ref{fig:hiddencatalysis} for a numerical simulations of the periodic
orbits. For simplicity, we have here focused on one single explicit
example. Similar features appear for variations on the same example, for
example, for larger $n$ or different kinetics.

\begin{figure}[ht]
    \centering
    \includegraphics[width=0.75\linewidth]{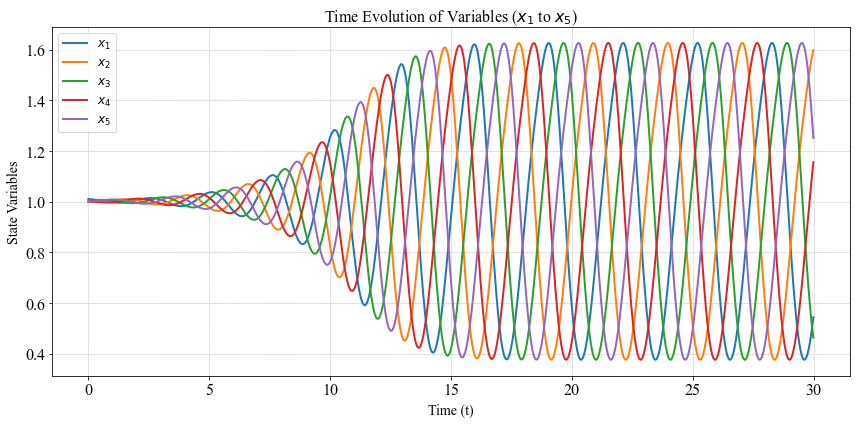}
    \caption{A numerical simulation of system \eqref{eq:explicitsystem}, with constants $k_1=k_2=k_3=k_4=k_5=1$, and $n=2$. The unique steady-state $\bar{x}=(1,1,1,1,1)$ is unstable. Nearby initial conditions $x_0=(1.01, 1, 1, 1, 1)$ shows convergence to a stable periodic orbit.}
\label{fig:hiddencatalysis}
\end{figure}

\paragraph{Consequences and conclusions.} We have thus shown that a
innocent-looking stoichiometry, even deficiency zero if interpreted
verbatim, can give rise to periodic oscillations if proper ``hidden''
catalysts are added. This observation highlights a more general issue regarding modelling
  assumptions. As noted above, ``hidden'' catalysts violate the nonambiguity assumption
  on non-zero reactivities. In a similar vein, \cite{MuFlSt:22}
  characterized chemically feasible networks in terms of thermodynamic
  soundness. Here, the underlying assumption was again that reactions are
  specified nonambiguously by the stoichiometric matrix, i.e., there are no hidden reactants or products.
  If reactants are considered to be ``buffered'', i.e., supplied by an
  infinite external reservoir, their concentrations become constant and
  thus their reactivities vanish. This eliminates them from the reaction
  rate model. For the purpose of investigating reaction kinetics, buffered
  reactants can also be dropped from the stoichiometric matrix
  $\mathbb{S}$. From a chemical point of view, however, we are then considering an
  open system, and hence the network can admit nonstationary periodic
  solutions, which are impossible in a closed chemical reaction system.
  The choice of reaction kinetics, therefore, may imply that there are
  hidden reactants and products. In fact, nothing in our definitions of the
  kinetic laws ensures that they are consistent with the constraints
  imposed by thermodynamics. It remains a topic for future research
  to disentangle in full generality the relationships between reaction
  kinetics, thermodynamics, and external in/out fluxes of reactants and
  energy, see e.g.\ \cite{polettini_irreversible_2014}.

To modify our example above to make it \emph{thermodynamically sound} sensu \cite{MuFlSt:22}, we add
reversible reactions. We then leverage standard \emph{inheritance} results
by Banaji \cite{Ba18} which guarantee that a nondegenerate periodic
solution in an irreversible system persists if reversible reactions are
added. This implies the existence of periodic solutions in systems with
reversible stoichiometry as:
\begin{equation}
\mathbb{S}=
\begin{pmatrix}
-1 & 1 & 0 & 0 & 0 & 0 & 0 & 0 & 1 & -1\\
1 & -1 & -1 & 1 & 0 & 0 & 0 & 0 & 0& 0 \\
0 & 0 & 1 & -1 & -1 & 1  & 0 & 0 & 0& 0\\
0 & 0 & 0 & 0 & 1& -1 & -1 & 1 & 0& 0\\
0 & 0 & 0 & 0 & 0 &0  & 1& -1 & -1& 1\\
\end{pmatrix},
\end{equation}
which is thermodynamically sound according to \cite{MuFlSt:22}: one can realize closed systems with this
stoichiometry. However, oscillations cannot persist in such a truly closed system. This
example illustrates that a deviation of only (effective) rate constants
from their closed-system constraints can suffice to produce
oscillations.

\vspace{3cm}
%Such a situation can be realized by driving reactions
%(e.g. photochemically) or when the described system is a partial
%description of reactions that further include hidden chemostatted species.

%Therefore, this example contradicts thermodynamics solely at the kinetic level while preserving a thermodynamically valid stoichiometry. 

%Here, oscillations are maintained by having rate constants 

 \textbf{\Large \hypertarget{part3}{PART III}: Appendix. How to find oscillations in silico}

As an example, we explicitly construct an oscillatory system starting from an oscillatory core. This presentation remains abstract and serves solely to illustrate how our methods can identify oscillations in systems incorporating our framework. The realistic design of chemical oscillations involves far greater detail and complexity, which lies beyond the scope of this work. For simplicity, we consider the smallest of our selection, i.e. Oscillatory Core (I,a), which can be summarized in the following stoichiometric matrix:
\begin{equation}
\begin{pmatrix}
-1 & 2 & 1\\
1 & -1 & 1\\
0 & -1 & -1
\end{pmatrix}.
\end{equation} 
As remarked in Obs.~\ref{obs:invertible}, oscillatory cores are always invertible, and consequently do not possess any right kernel vector, i.e., a system whose stoichiometric matrix consists of the oscillatory core alone does not admit any steady-state, for any choice of reaction rates. The cheapest modification to the system that admits a steady-state, without adding new species or complexes, simply considers each of the reactions in the oscillatory core as reversible:
\begin{equation}
\begin{cases}
\ce{X}&\underset{4}{\overset{1}{\rightleftharpoons}} \quad \ce{Y}\\
\ce{Y}+\ce{Z} &\underset{5}{\overset{2}{\rightleftharpoons}} \quad  2\ce{X}\\
\ce{Z}&\underset{6}{\overset{3}{\rightleftharpoons}} \quad \ce{X}+\ce{Y}
\end{cases},
\end{equation}
with stoichiometric matrix:
\begin{equation}
\mathbb{S}=
\begin{pmatrix}
-1 & 2 & 1 & 1 & -2 & -1\\
1 & -1 & 1 & -1 & 1 & -1\\
0 & -1 & -1& 0 & 1 & 1
\end{pmatrix}.
\end{equation}
This way the associated dynamical systems,
\begin{equation}\label{eq:appendix}
\begin{cases}
\dot{x}=-r_1(x)+2r_2(y,z)+r_3(z)+r_4(y)-2r_5(x)-r_6(x,y)\\
\dot{y}=r_1(x)-r_2(y,z)+r_3(z)-r_4(y)+r_5(x)-r_6(x,y)\\
\dot{z}=-r_2(y,z)-r_3(z)+r_5(x)+r_6(x,y)
\end{cases}
\end{equation}
admits a steady-state for 
\begin{equation}
\begin{cases}
\bar{r}_1(\bar{x})=\bar{r}_4(\bar{y})\\
\bar{r}_2(\bar{y},\bar{z})=\bar{r}_5(\bar{x})\\
\bar{r}_3(\bar{z})=\bar{r}_6(\bar{x},\bar{y}).
\end{cases}
\end{equation}
Our framework requires using a kinetics that is parameter-rich, and we choose for this demonstration rate functions found for irreversible Michaelis--Menten kinetics, Eq.~\ref{MMeq}. For convenience, we consider the case where the rate functions can be factorized into rational functions with numerator and denominator being polynomials of degree 1. This entails no loss of generality while simplifying notation and leaving enough parameters for parameter-richness.\footnote{for instance, a rate like $r_5(x)$ for $\ce{X} + \ce{X} \rightarrow ..$ would usually have as its MM expression $a_5 \frac{x}{1+c_{5,1}x + c_{5,2}x^2}$ and ipso facto depend on three parameters. The factorization into $a_5 \left( \frac{x}{1+b_{5}x}\right)^2$ reduces the number of parameters to 2, effectively introducing a parametric constraint, but still leaving enough freedom for parameter-richness. }

Under such modeling assumption, the reaction rates become:
\begin{equation}
\begin{cases}
r_1(x)=a_1 \dfrac{x}{1+b_1x}\\
\\
r_2(y,z)=a_2 \dfrac{y}{1+b_2^y y} \dfrac{z}{1+b_2^z z}\\
\\
r_3(z)=a_3 \dfrac{z}{1+b_3z}\\
\\
r_4(y)=a_4 \dfrac{y}{1+b_4y}\\
\\
r_5(x)=a_5 \bigg(\dfrac{x}{1+b_5x}\bigg)^2\\
\\
r_6(x,y)=a_6\dfrac{x}{1+b_6^x x}  \dfrac{y}{1+b_6^y y} \\
\end{cases}
\end{equation}
We leverage the parameter-richness of the kinetics: we fix a priori, arbitrarily and to our convenience, reference values of the steady-state 
\begin{equation}
(\bar{x},\bar{y},\bar{z})=(1,1,1)
\end{equation}
and of its fluxes 
\begin{equation}
(\bar{r}_1(1),\bar{r}_2(1,1),\bar{r}_3(1),\bar{r}_4(1),\bar{r}_5(1),\bar{r}_6(1,1))=(2,2,2,2,2,2).
\end{equation}
In order to establish this value choice as a reference steady state, independent of the parameters, we express the parameters $\mathbf{a}$ in terms of the parameters $\mathbf{b}$ through the following parametrization $\mathbf{a}(\mathbf{b})$
\begin{equation}
\begin{cases}
a_1=2(1+b_1)\\
a_2=2(1+b_2^y)(1+b_2^z)\\
a_3=2(1+b_3)\\
a_4=2(1+b_4)\\
a_5=2(1+b_5)^2\\
a_6=2(1+b_6^x)(1+b_6^y).
\end{cases}
\end{equation}
It is straightforward to check that that such choice of the parameters $\mathbf{a}$ makes $(\bar{x},\bar{y},\bar{z})=(1,1,1)$ a steady-state with fluxes $(2,2,2,2,2,2)$ irrespective of the choice of $\mathbf{b}$. The parameters $\mathbf{b}$, however, still can decide the relative size of the derivatives that appear in the Jacobian, and can thus be used to harness the system dynamics towards the unstable region where periodic orbits arise. We recall that the linear stoichiometric argument for our result, based on the concept of $D$-Hopf matrix,  is essentially that the product
\begin{equation}
\begin{pmatrix}
-1 & 2 & 1\\
1 & -1 & 1\\
0 & -1 & -1
\end{pmatrix}\begin{pmatrix}
1 & 0 & 0\\
0 & 1 & 0\\
0 & 0 & \beta
\end{pmatrix}=\begin{pmatrix}
-1 & 2 & \beta\\
1 & -1 & \beta\\
0 & -1 & -\beta
\end{pmatrix}
\end{equation}
is Hurwitz-stable for $\beta=1$ and Hurwitz-unstable for $\beta>0$ small enough, while being invertible for any $\beta>0$. 

We want to mimic this idea for the symbolic Jacobian matrix $\mathbb{S}R$ of the system \eqref{eq:appendix}:
\begin{equation}
\begin{split}\mathbb{S}R=&
\begin{pmatrix}
-1 & 2 & 1 & 1 & -2 & -1\\
1 & -1 & 1 & -1 & 1 & -1\\
0 & -1 & -1& 0 & 1 & 1
\end{pmatrix}\begin{pmatrix}
R_{1X} & 0 & 0\\
0 & R_{2Y} & R_{2Z}\\
0 & 0 & R_{3Z}\\
0 & R_{4Y} & 0\\
R_{5X} & 0 & 0\\
R_{6X} & R_{6Y} & 0
    \end{pmatrix}\\
=&
    \begin{pmatrix}
        -R_{1X}-2R_{5X}-R_{6X} & 2R_{2Y}+R_{4Y}-R_{6Y} & 2R_{2Z}+R_{3Z}\\
R_{1X}+R_{5X}-R_{6X} & -R_{2Y}-R_{4Y}-R_{6Y} & -R_{2Z}+R_{3Z}\\
R_{5X}+R_{6X} & -R_{2Y}+R_{6Y} & -R_{2Z}-R_{3Z}
    \end{pmatrix}
    \end{split}
\end{equation}
To do so, we firstly pick the following choices for the derivatives $R_{jm}$:
\begin{equation}
    \begin{cases}
        R_{1X}=1\\
        R_{2Y}=1\\
        R_{3Z}=\beta\\
        R_{2Z}=R_{4Y}=R_{5X}=R_{6X}=R_{6Y}=0.1
    \end{cases}
\end{equation}
where we have highlighted precisely the derivatives in correspondence of the oscillatory core, as the symbolic Jacobian evaluated at this choice of parameters is
\begin{equation}
G^*=
\begin{pmatrix}
        -1-0.3 & 2 & 0.2+\beta\\
1 & -1-0.2& -0.1+\beta\\
0.2 & -1+0.1 & -0.1-\beta
    \end{pmatrix}
\end{equation}
Consider $\beta\in[0.08,1]$. At $\beta=1$ the eigenvalues of $G^*$ are approximately $(-2.66, -0.47 \pm 0.39i)$ and thus $G^*$ is Hurwitz-stable. In contrast, at $\beta=0.08$ the eigenvalues are   $(-2.72, \mathbf{0.02 \pm 0.18i})$: A Hopf bifurcation occurred on the way. 
By direct inversion of the Michaelis--Menten nonlinearity \eqref{MMeq}, we can then find the explicit values for the parameters $\mathbf{b}$:
\begin{equation}
    \begin{cases}
        b_1=1;\\
        b_{2}^y=1;\\
        b_2^z=19;\\
        b_3^z\in[1,24] \text{ which corresponds to }R_{3z}=\beta\in[1,0.08]\\
        b_{4}=19;\\
        b_{5}=19;\\
        b_6^x=19;\\
        b_6^y=19.
    \end{cases}
\end{equation}
and we can check numerically the dynamics for $\beta=0.08$ to find indeed a stable periodic orbits, see Fig.~\ref{fig:appendix}.
We see how the choice of parameters for the oscillatory regime precisely highlight the autocatalytic core 
\begin{equation}
    \begin{pmatrix}
    -1 & 2\\
    1 & -1
\end{pmatrix},
\end{equation}
i.e., the associated derivatives $R_{1X}$ and $R_{2Y}$ are the only of of order 1, while all the others are chosen small enough. 

\begin{figure}
    \centering
    \includegraphics[width=0.75\linewidth]{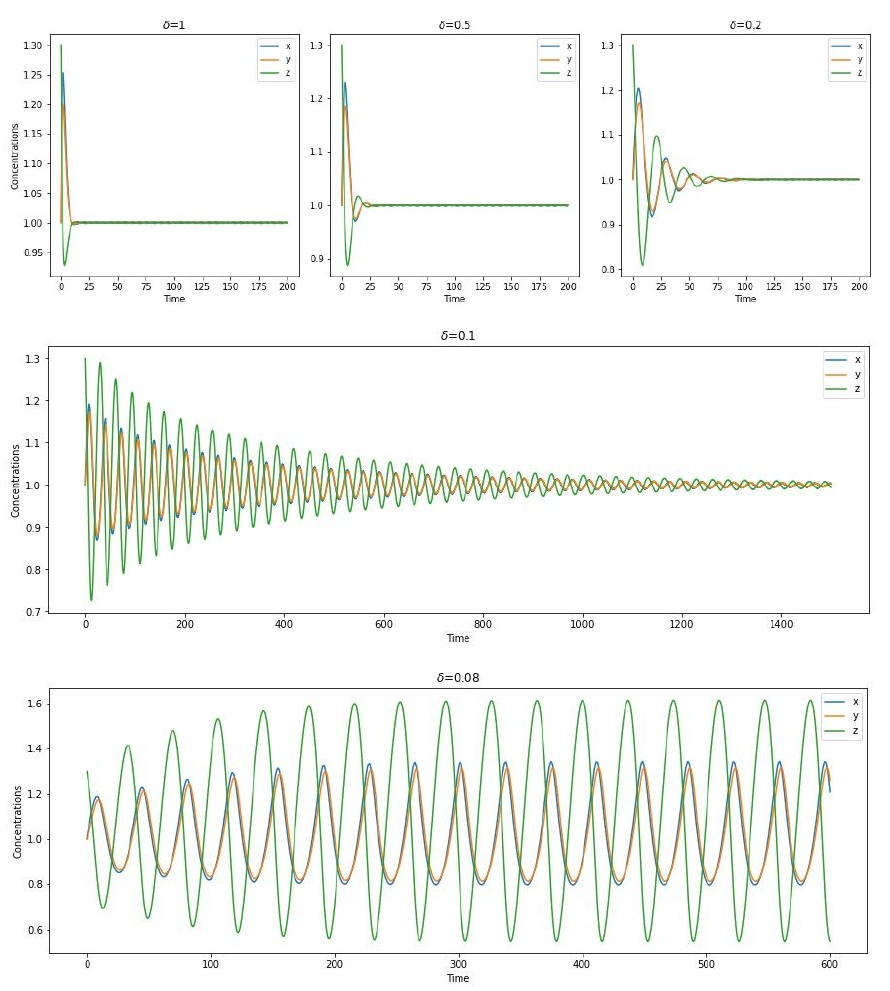}
    \caption{Numerical simulations for decreasing values of $\beta=1,0.5,0.2,0.1,1.$ Initial conditions are chosen always $(x_0,y_0,z_0)=(1,1,1.3)$ For $\beta=1$ the stable steady-state $(\bar{x},\bar{y},\bar{z})$ is strongly attracting. Its stability weakens with decreasing values of $\beta$. At $\beta=1.1$ we clearly see a damped oscillations that slowly converges to the steady-state. At $\beta=0.8$ the steady-state is already unstable, and a stable periodic orbits appears.}
    \label{fig:appendix}
\end{figure}

\end{document}